\documentclass[a4paper,english,10pt]{article}
\usepackage[left=.9in, right=.9in, bottom=1.2in, top=1in,dvips]{geometry}
\makeatletter

\AtBeginDocument{\setlength\abovedisplayskip{5pt plus 2pt minus 1pt}
\setlength\belowdisplayskip{5pt plus 2pt minus 1pt}}

\usepackage[utf8]{inputenc}
\usepackage[T1]{fontenc}
\usepackage{lmodern}
\usepackage{indentfirst}
\usepackage[affil-it]{authblk}
\usepackage{rotating}
\usepackage{tocloft}
\usepackage{titlesec}
 \titleformat{\subparagraph}[hang]{\normalfont}{\thesubparagraph}{0pt}{\underline}
 \titleformat{\paragraph}[hang]{\normalfont}{\theparagraph}{0pt}{\myuline}
\usepackage{titletoc}
\usepackage{color, soul}
 \usepackage{xcolor}
 \usepackage[citebordercolor={green}]{hyperref}
\usepackage{array}
\usepackage{tabularx}
\usepackage{amssymb}
\usepackage{amsmath}
\usepackage{bbm}

\usepackage{amsthm,amsfonts,amssymb,euscript}
\usepackage{latexsym, multicol, fancybox}
\usepackage{amsmath, amsthm, amssymb, bm}
\usepackage{caption}
\usepackage{psfrag}
\usepackage{setspace}
\usepackage{mathrsfs}
\usepackage{epsfig}
\usepackage{verbatim}
\usepackage[affil-it]{authblk}
\usepackage{ipa}
\usepackage{tikz}

\newtheorem{theorem}{Theorem}[section]
\newtheorem{lemma}[theorem]{Lemma}
\newtheorem{proposition}[theorem]{Proposition}

\newtheorem{definition}[theorem]{Definition}
\newtheorem{remark}[theorem]{Remark}

\newtheorem*{theorem*}{Theorem}
\newtheorem*{prop*}{Proposition}

\newcommand{\bea}{\begin{eqnarray}}
\newcommand{\eea}{\end{eqnarray}}
\def\beaa{\begin{eqnarray*}}
\def\eeaa{\end{eqnarray*}}

\newcommand{\MM}{\mathcal{M}}
\def\DD{{\mathcal D}}
\def\TT{{\mathcal T}}
\def\PP{\mathcal{P}}

\def\D{{\bf D}}
\def\F{{\mathbf{F}}}
\def\R{{\bf R}}
\def\W{{\bf W}}
\def\g{{\bf g}}

\def\CCC{{\Bbb C}}
\def\Ddot{\dot{\D}}
\def\RR{\mathcal{R}}

\def\a{{\alpha}}
\def\al{{\alpha}}
\def\b{{\beta}}
\def\be{{\beta}}
\def\ga{\gamma}

\def\de{\delta}
\def\De{\Delta}
\def\ep{\epsilon}

\def\la{\lambda}
\def\La{\Lambda}
\def\om{\omega}
\def\vphi{\varphi}
\def\si{\sigma}
\def\th{\theta}
\def\ze{\zeta}
\def\nab{\nabla}

\def\trch{{\mbox tr}\, \chi}
\def\chih{{\hat \chi}}
\def\chib{{\underline \chi}}
\def\chibh{{\underline{\chih}}}
\def\etab{{\underline \eta}}
\def\omb{{\underline{\om}}}
\def\bb{{\underline{\b}}}
\def\aa{{\underline{\a}}}
\def\xib{{\underline \xi}}
\def\Xib{\underline{\Xi}}
\def\lap{{\triangle}}
\def\atr{\,^{(a)}\mbox{tr}}
\def\trchb{{\tr \,\chib}}
\def\atrch{\atr\chi}
\def\atrchb{\atr\chib}
\def\rhod{\,\dual\rho}

\def\ZZ{\mathcal{Z}}

\def\qf{\frak{q}}
\def\pf{\mathfrak{p}}

\def\Ffr{\mathfrak{F}}
\def\Bfr{\mathfrak{B}}

\def\sk{\mathfrak{s}}

\def\bF{\,^{(\mathbf{F})} \hspace{-2.2pt}\b}
\def\bbF{\,^{(\mathbf{F})} \hspace{-2.2pt}\bb}
\def\rhoF{\,^{(\mathbf{F})} \hspace{-2.2pt}\rho}

\def\BF{\,^{(\mathbf{F})} \hspace{-2.2pt} B}
\def\BBF{\,^{(\mathbf{F})} \hspace{-2.2pt}\underline{B}}
\def\PF{\,^{(\mathbf{F})} \hspace{-2.2pt} P}

\def\Lb{{\,\underline{L}}}

\def\pr{{\partial}}
\def\les{\lesssim}

\def\c{\cdot}
\def\dual{{\,\,^*}}
\def\div{{\mbox div\,}}
\def\curl{{\mbox curl\,}}
\def\lot{\mbox{ l.o.t.}}
\def\Hb{\,\underline{H}}
\def\Ab{\underline{A}}
\def\Bb{\underline{B}}
\def\Xb{\underline{X}}

\def\tr{\mbox{tr}}
\def\hot{\widehat{\otimes}}

\def\squared{\dot{\square}}
\def\lab{\label}

\def\DDov{\ov{\DD}}

\def\nn{\nonumber}
\def\ov{\overline}

\def\DDs{ \, \DD \hspace{-2.4pt}\dual    \mkern-20mu /}

\def\DDs{ \, \DD \hspace{-2.4pt}\dual    \mkern-20mu /}
\def\DDd{ \, \DD \hspace{-2.4pt}    \mkern-8mu /}

\def\Kh{\,^{(h)}K}

\newcommand{\nabb}{\nab\mkern-13mu /\,}

\def\OO{\mathcal{O}}
\def\KK{\mathcal{K}}
\def\QQ{\mathcal{Q}}
\def\LL{\mathcal{L}}
\def\GG{\mathcal{G}}
\def\BB{\mathcal{B}}
\def\EE{\mathcal{E}}
\def\AA{\mathcal{A}}
\def\VV{\mathcal{V}}
\def\SS{\mathcal{S}}
\def\FF{\mathcal{F}}
\def\UU{\mathcal{U}}
\def\WW{\mathcal{W}}
\def\piX{\, ^{(X)}\pi}
\def\That{\widehat{T}}

\def\HH{\mathcal{H}}
\def\SSS{\mathbb{S}}

\def\gadot{\dot{\gamma}}
\def\e{{\bf e}}
\def\k{{\bf k}}

\def\g{{\bf g}}

\def\lz{{\bf \ell_z}}
\def\aund{{\underline{a}}}
\def\bund{{\underline{b}}}
\def\cund{{\underline{c}}}

\def\psia{{\psi_{\underline{a}}}}
\def\psib{{\psi_{\underline{b}}}}
\def\psic{{\psi_{\underline{c}}}}

\def\pfa{{\pf_{\underline{a}}}}
\def\pfb{{\pf_{\underline{b}}}}
\def\pfc{{\pf_{\underline{c}}}}
\def\qfa{{\qf_{\underline{a}}}}
\def\qfb{{\qf_{\underline{b}}}}
\def\qfc{{\qf_{\underline{c}}}}

     \def\dk{\mathfrak{d}}

\def\Db{{\dot{\D}}}

\def\Sb{{S_\bund}}
\def\Sc{{S_\cund}}
\def\RRtp{{\widetilde{\RR}}}

\newcommand{\Edth}{\textnormal{\DH}}

\allowdisplaybreaks

\begin{document}

 \title{\LARGE \textbf{The Carter tensor and the physical-space analysis \\ in perturbations of Kerr-Newman spacetime}}
 
 \author[1]{{\Large Elena Giorgi\footnote{egiorgi@princeton.edu}}}

\affil[1]{\small  Princeton University, Gravity Initiative, Princeton, 
United States \vspace{0.2cm} \ }

\maketitle

\begin{abstract} 

The Carter tensor is a Killing tensor of the Kerr-Newman spacetime, and its existence implies the separability of the wave equation. Nevertheless, the Carter operator is known to commute with the D'Alembertian only in the case of a Ricci-flat metric. We show that, even though the Kerr-Newman spacetime satisfies the non-vacuum Einstein-Maxwell equations, its curvature and electromagnetic tensors satisfy peculiar properties which imply that the Carter operator still commutes with the wave equation. 
This feature allows to adapt to Kerr-Newman the physical-space analysis of the wave equation in Kerr by Andersson-Blue \cite{And-Mor}, which avoids frequency decomposition of the solution by precisely making use of the commutation with the Carter operator.

We also extend the mathematical framework of physical-space analysis to the case of the Einstein-Maxwell equations on Kerr-Newman spacetime, representing coupled electromagnetic-gravitational perturbations of the rotating charged black hole. The physical-space analysis is crucial in this setting as the coupling of spin-1 and spin-2 fields in the axially symmetric background prevents the separation in modes as observed by Chandrasekhar \cite{Chandra}, and therefore represents an important step towards an analytical proof of the stability of the Kerr-Newman black hole.

\end{abstract}

\bigskip\bigskip

\tableofcontents

\section*{Introduction}
\addcontentsline{toc}{section}{Introduction}

The Kerr-Newman spacetimes $(\MM_{\g_{M, a, Q}}, \g_{M, a, Q})$ \cite{Newman} are a three-parameter family of solutions to the Einstein-Maxwell equations. They are the most general explicit black hole solutions to the Einstein equation, representing rotating and charged black holes, where $M$ is the mass of the black hole, $a$ its angular momentum and $Q$ its charge, in the (sub-extremal and extremal) range $a^2+Q^2\leq M^2$. The Kerr-Newman metric in Boyer-Lindquist coordinates $(t, r,  \th, \vphi)$ takes the form
\beaa
\g_{M, a, Q}=-\frac{\Delta}{|q|^2}\left( dt- a \sin^2\th d\vphi\right)^2+\frac{|q|^2}{\Delta}dr^2+|q|^2 d\th^2+\frac{\sin^2\th}{|q|^2}\left(a dt-(r^2+a^2) d\vphi \right)^2,
\eeaa
where $ \Delta = r^2-2Mr+a^2+Q^2$ and $ |q|^2=r^2+a^2\cos^2\th$.
The Kerr-Newman spacetimes generalize the Kerr solution (for vanishing charge), the Reissner-Nordstr\"om solution (for vanishing angular momentum) and the Schwarzschild solution (for both $a=Q=0$),  and are expected to be the final state of gravitational collapse \cite{NMR}.

The problem of stability of the domain of outer communication of the Kerr-Newman spacetime as solution to the Einstein-Maxwell equations is an important open problem in General Relativity and has implications for the physical relevance of the Kerr-Newman solution as a realistic black hole. Tremendous progress towards the proof of stability of the black hole solutions has been made in the past fifteen years, with works which encompass scalar, electromagnetic, gravitational perturbations of Schwarzschild, Reissner-Nordstr\"om, Kerr and Kerr-Newman spacetimes, from mode stability, to linear and fully non-linear stability, see for example \cite{DR11}\cite{TeukolskyDHR}\cite{KS}\cite{Kerr-lin1}\cite{Kerr-lin2}\cite{Y-R}\cite{DHRT}\cite{KS-Kerr} and references therein. 

\medskip

In \cite{Giorgi7} we started a program aimed to prove the linear stability of the Kerr-Newman family to coupled electromagnetic-gravitational perturbations. Numerical works strongly support stability for the Kerr-Newman family \cite{numerics1}\cite{numerics2}\cite{Mark}\cite{Wang2}, but nevertheless an analytical proof of even its weakest version, the mode stability, has not been obtained. 
 Mode solutions to the wave equation are solutions of the  separated form 
\bea\label{mode-solution}
\psi(r, t, \theta, \vphi)= e^{- i \omega t} e^{i m\vphi} R(r) S(\th),
\eea
where
 $\omega \in \mathbb{C}$ is the time frequency, and $m \in \mathbb{Z}$ is the azimuthal mode. Mode stability consists in showing that there are no solutions of the form \eqref{mode-solution} with bounded initial energy which are exponentially growing in time. 
 
 The above separated form for $\omega \in \mathbb{R}$ and $m \in \mathbb{Z}$ is related to the Fourier transform of the solution with respect to the symmetries of the spacetime (the stationary vectorfield $\partial_t$ and the axially symmetric $\partial_\vphi$), and therefore corresponds to its frequency decomposition. In addition to the two Killing vectorfields, the Kerr-Newman metric admits a Killing tensor, called the Carter tensor \cite{Carter}, which represents a hidden symmetry of the spacetime and which provides a third constant of motion which allows the full integrability of the geodesic flow. 
Because of such integrability, functions of the form \eqref{mode-solution} are solutions to the wave equation as long as $R(r)$ and $S(\th)$ respectively satisfy a radial and an angular ODE. The Carter separation introduces, in addition to the frequency $\omega$ and $m$, a real frequency parameter $\lambda_{m \ell}( a \omega)$ parametrized by $\ell \in \mathbb{N}_0$, which are the eigenvalues of an associated elliptic operator, whose eigenfunctions are the \textit{oblate spheroidal harmonics} $S_{\om m \ell}$ \cite{Teukolsky}.

As observed by Chandrasekhar \cite{Chandra} in the 80s, the methods developed in black hole perturbation theory for the other black hole solutions, such as Kerr or Reissner-Nordstr\"om, involving the separation in modes as in \eqref{mode-solution} and the proof of their mode stability, did not seem to be applicable for treating the coupled electromagnetic-gravitational perturbations of Kerr-Newman spacetime. In fact, in decomposing the system of equations in separated forms \eqref{mode-solution}, the electromagnetic field (of spin $\pm1$) and the gravitational field (of spin $\pm2$) get separated in spheroidal harmonics $S^{[s]}_{\om m \ell}$ of different spins $s$, which do not easily interact with the coupling operators appearing as consequence of the Einstein-Maxwell equations. On the other hand, for gravitational perturbations of Kerr there is no interaction between spheroidal harmonics of different spins, as only the spin $\pm2$ gravitational field is considered. In the case of the spherically symmetric Reissner-Nordstr\"om, the spheroidal harmonics reduce to the standard spherical harmonics, which commute with the coupling operators thanks to the Hodge theory on the sphere. Instead, in the most general case of electromagnetic-gravitational perturbations of Kerr-Newman, 
the interaction between spin-1 and spin-2 fields in the axially symmetric background prevents the separability in modes, and their  frequency decomposition is not possible due to the coupling of the equations of different spins, see Section \ref{section:spheroidal}. For more details, see also \cite{Chandra}\cite{Giorgi7}. 

As a consequence of the non separability in modes, we expect that the proof of the linear stability of the Kerr-Newman family necessitates a \textbf{physical-space analysis} of the coupled system of equations describing the interaction between gravitational and electromagnetic radiations, which in particular does not rely on decomposition in spheroidal harmonics. 

\medskip

The goal of this paper is to derive the mathematical framework to obtain estimates for the scalar wave equation
\bea\label{wave-equation-intro}
\square_{\g_{M, a, Q}} \psi=0
\eea
on the exterior region of Kerr-Newman black hole \textit{exclusively involving physical-space analysis}, i.e. no frequency decomposition of the solution, and which is robust enough to be adapted to the system of coupled Teukolsky and Regge-Wheeler equations obtained in our \cite{Giorgi7}, in the framework of the proof of linear stability. A physical-space analysis has also the advantage of being potentially more adaptable to the non-linear stability problem, see \cite{KS-Kerr} and \cite{GKS2}.
We prove the following:

\begin{theorem*} Boundedness of the energy flux through a global foliation on the exterior of Kerr-Newman spacetime $(\MM_{\g_{M, a, Q}}, \g_{M, a, Q})$ with $|a| \ll M$, and a suitable version of local integrated energy decay can be obtained  \textbf{exclusively through a physical-space analysis} for solutions to the wave equation \eqref{wave-equation-intro} arising from bounded initial energy on a suitable Cauchy surface.
\end{theorem*}

More precisely, we obtain estimates through a physical-space analysis for axially symmetric solutions in the sub-extremal range $a^2+Q^2< M^2$, and  for general solutions in the slowly rotating range $|a| \ll M$, see Theorem \ref{theorem:general}. 

We then show that the physical-space analysis obtained for the wave equation can be extended to the coupled system of generalized Regge-Wheeler equations for $|a| \ll M$ obtained in \cite{Giorgi7} describing perturbations of Kerr-Newman black hole, see Proposition \ref{proposition:modified-OO-system}.

\medskip

\subsubsection*{Previous results on the wave equation in black hole backgrounds}

Boundedness and decay properties for scalar wave equations in Kerr and Kerr-Newman have been obtained in \cite{DR10}\cite{DR11b}\cite{DR13}\cite{Tataru}\cite{And-Mor} in the slowly rotating case, and in \cite{DRSR}\cite{Civin} in the full sub-extremal range. Nevertheless, the only proof obtained exclusively through a physical-space analysis is  the proof of Andersson-Blue \cite{And-Mor} in slowly rotating Kerr spacetimes. We now give a summary of previous results on the wave equation on black hole backgrounds.

In the spherically symmetric Schwarzschild and Reissner-Nordstr\"om  spacetimes, the Killing vectorfield $\partial_t$ is timelike everywhere in the exterior region, and the orbital null geodesics, i.e. null geodesics that neither fall into the black hole nor escape to infinity and which are then an obstruction to decay, all asymptote to a physical-space timelike cylinder, called the \textit{photon sphere}. The first property implies that superradiance is not present, as the vectorfield $\partial_t$ 
 defines a positive definite conserved energy. The second property implies that the trapping region of the spacetime does not depend on the frequency of the solution. 
  As the orbital null geodesics are restricted to a physical-space hypersurface, given by $\{ r=r_{trap}\}$ for some $r_{trap}$ outside the black hole ($r_{trap}=3M$ in Schwarzschild), spacetime energy estimates can be obtained through a vectorfield of the form $\FF(r) \partial_r$, with $\FF$ vanishing at $r=r_{trap}$, and the analysis of the wave equation can be performed in physical-space.
Physical-space analysis of the wave equation in spherically symmetric backgrounds have been obtained in \cite{KW}\cite{BS}\cite{DR09a}\cite{DR09}. See also \cite{Stogin}\cite{AAG}\cite{AAG1}, and \cite{extremal-1}\cite{extremal-2}\cite{AAG2} in the case of extremal Reissner-Nordstr\"om. 
Similarly, the study of the Maxwell equations has been obtained in \cite{BlueMax}\cite{ST}\cite{Federico}.

In the case of gravitational perturbations of Schwarzschild spacetime, the analysis of the separated mode solutions was obtained by the physics community  in the 70s,  see \cite{Regge-Wheeler}\cite{Bardeen-Press}\cite{Chandra}.
The first quantitative result for the linear gravitational perturbations of Schwarzschild through a physical-space analysis of the Teukolsky equation was obtained by Dafermos-Holzegel-Rodnianski in \cite{DHR}, see also \cite{mu-tao}\cite{Johnson}. The physical-space analysis of the linearized gravity exploited, in addition to the conserved energy and the spacetime estimates which degenerate at the photon sphere, a hierarchy of wave-like equations, from the Teukolsky to the Regge-Wheeler equation \cite{DHR}.   For non-linear gravitational perturbations of Schwarzschild, see the work by Klainerman-Szeftel \cite{KS} under the class of axially symmetric polarized perturbations, and the recent work by Dafermos-Holzegel-Rodnianski-Taylor \cite{DHRT} for the full non-linear stability.

In the case of electromagnetic-gravitational perturbations of sub-extremal Reissner-Nordstr\"om spacetime, the analysis of the separated mode solutions was also obtained by the physics community in \cite{Moncrief1}\cite{Moncrief2}\cite{Chandra}. The quantitative results for the linear electromagnetic-gravitational perturbations through the physical-space analysis of the Teukolsky system, and its derived Regge-Wheeler system, were obtained in our series of works \cite{Giorgi4}\cite{Giorgi5}\cite{Giorgi6}\cite{Giorgi7a}.

\medskip

In the axially symmetric Kerr and Kerr-Newman spacetimes, the analysis of the wave equation is complicated by two factors: the presence of an ergoregion, and therefore superradiance of the solution, and the dependence of the trapping region on the frequency of the solution. The Killing vectorfield $\partial_t$ becomes spacelike in a region outside the event horizon known as the \textit{ergoregion}, hence its conserved energy is not positive definite everywhere. Moreover, the trapped null geodesics are not confined to a hypersurface in physical-space, but rather cover an open region of the spacetime which depends on the energy and angular momentum of the geodesics, and therefore the trapping degeneracy for the wave equation depends on the frequency of the solution. For this reason the high frequency obstruction to decay given by the trapping region cannot be described by the classical vectorfield method \cite{Alinhac}. 

These problems have been first overcome in the case of small angular momentum, $|a| \ll M$, in which case both the superradiance and the trapping simplify. In \cite{DR11b}, Dafermos-Rodnianski showed that the superradiance can be controlled by analyzing the solution in its separated form \eqref{mode-solution} and decomposing it into its superradiant and non-superradiant parts. Crucially the superradiant part is not trapped, and it satisfies a local energy decay identity obtained by perturbing the one in Schwarzschild. In \cite{DR10}\cite{DR13}, Dafermos-Rodnianski also overcame the problem of capturing the trapping region using frequency-localized generalizations of the Morawetz multipliers obtained in Schwarzschild. Even though the null geodesics are not localized on a physical-space hypersurface, they are localized in frequency-space, as the potential of the radial ODE has a unique simple maximum in the trapped frequency range, whose value $r_{trap}(a\omega, m, \lambda_{m\ell})$ depends on the frequency parameters. This allows for the construction of a frequency-space analogue of the current $\FF(r) \partial_r$ which vanishes at $r_{trap}$ for each triple of trapped frequencies. See also the independent approach by Tataru-Tohaneanu in \cite{Tataru}.

In \cite{And-Mor}, Andersson-Blue obtained the first analysis of solutions exclusively in physical space  in slowly rotating Kerr spacetime. This approach, as the ones in frequency-decomposition \cite{DR10}\cite{DR13}\cite{Tataru}, makes use of the Carter hidden symmetry in Kerr, but not through the separation of the solution as in \eqref{mode-solution}, but rather as a physical-space commutator to the wave equation. In particular, as Killing tensors commute with the D'Alembertian operator in Ricci-flat spacetimes, the second order differential operator associated to the Carter tensor can be used as a symmetry operator in addition to the Killing vectorfields $\partial_t$ and $\partial_\vphi$. This allows to obtain a local energy decay identity at the level of three derivatives of the solution which degenerate near $r=3M$, as trapped null geodesics lie within $O(|a|)$ of the photon sphere $r=3M$.

In passing from the slowly rotating case $|a| \ll M$ to the full sub-extremal range in Kerr or Kerr-Newman, there is an intermediate step which present many simplifications, which is the case of axially symmetric solutions to the wave equation, i.e. $\partial_\vphi \psi=0$, for $a^2+Q^2 <M^2$. For those solutions, superradiance is effectively absent and the trapping region collapses to a physical-space hypersurface. Although $\partial_t$ still fails to be everywhere timelike, its associated energy through the horizon is non-negative as the Hawking vectorfield $\That:=\partial_t +\frac{a}{r^2+a^2}\partial_\vphi$ is the null generator of the horizon and timelike everywhere outside it. In particular, for axially symmetric solutions, the Hawking vectorfield $\That$ behaves like the Killing vectorfield $\partial_t$. As the dependence on the azimuthal frequency $m$ becomes trivial, the axially symmetric trapped null geodesics all asymptote towards a hypersurface $\{ r=r_{trap}\}$ in physical-space, where $r_{trap}$ is defined as the largest root of the polynomial (see Section \ref{section:trapped-null-geodesics})
\bea\label{intro:definition-TT}
 \TT:= r^3-3Mr^2 + ( a^2+2Q^2)r+Ma^2,
\eea
 and therefore the construction of the current $\FF(r) \partial_r$ simplifies, see \cite{DR10} for the frequency-space construction, and see \cite{Stogin} for a construction entirely in physical space. The analysis of the axially symmetric solutions to the wave equation in frequency space has also been extended to the extremal Kerr $|a|=M$ by Aretakis in \cite{extremal-kerr}.

The only result at this day which extends the local energy decay estimates to the full sub-extremal range $|a| < M$ for the wave equation in Kerr is the work by Dafermos-Rodnianski-Shlapentokh-Rothman \cite{DRSR} in frequency-space. See also \cite{Civin} for the sub-extremal $a^2+Q^2<M^2$ Kerr-Newman. In \cite{DRSR}, Dafermos-Rodnianski-Shlapentokh-Rothman perform a careful construction of frequency-dependent multiplier currents for the separated solutions, and crucially make use of the structure of trapping, i.e. the existence of a  simple maximum $r_{trap}(a\omega, m, \lambda_{m\ell})$ for the radial potential, and the fact that superradiant frequencies are not trapped, which they show it holds in the full sub-extremal range $|a|<M$. They then make use of a continuity argument in $a$ to justify the future integrability necessary to perform the Fourier transform in time.

In the case of electromagnetic or gravitational perturbations of Kerr spacetime, the analysis of the mode stability was obtained by the physics community in the 70s, see \cite{Teukolsky}\cite{Whiting}\cite{Chandra}. For a quantitative mode stability in sub-extremal Kerr see \cite{Yakov}, and in extremal Kerr see \cite{Rita}. 

For quantitative results for the linearized electromagnetic and gravitational perturbations of slowly rotating Kerr spacetime, see \cite{AB}\cite{TeukolskyDHR}\cite{ma1}\cite{ma2}, where a hierarchy of wave-like equations from Teukolsky to a generalized Regge-Wheeler equation is exploited. See also \cite{Kerr-lin1}\cite{Kerr-lin2}.
For the analysis of the linearized gravity and electromagnetic perturbations in the sub-extremal range Kerr $|a|<M$ see the work by Shlapentokh-Rothman-Teixeira da Costa \cite{Y-R}. For coupled linear-gravitational perturbations of Kerr-Newman spacetime, see our \cite{Giorgi7} for the derivation of the Teukolsky and Regge-Wheeler system of equations governing the perturbations.

For non-linear gravitational perturbations of Kerr, see the formalism developed in \cite{GKS}, and the recent work by Klainerman-Szeftel \cite{KS-Kerr}.
In the presence of positive cosmological constant, the non-linear stability of slowly rotating Kerr-de Sitter and Kerr-Newman-de Sitter spacetimes have been obtained in \cite{Hintz-Vasy}\cite{Hintz-M}.

Finally, electromagnetic-gravitational perturbations of Kerr-Newman spacetime have been considered in \cite{Chandra} and asymptotic solutions were obtained in \cite{Lee}. In \cite{Mark}, the Teukolsky equations were derived in the phantom gauge. In \cite{Giorgi7} we derived the coupled system of equations for gauge-invariant perturbations of Kerr-Newman. See also the recent \cite{MG} for the construction of an energy functional for axisymmetric perturbations of Kerr-Newman.

\subsubsection*{The physical-space analysis of the wave equation in Kerr-Newman}

As mentioned above, most of the results for scalar, electromagnetic and gravitational perturbations of Kerr or Kerr-Newman spacetimes rely on the separability in modes and the frequency-decomposition of the solution. Even though these methods are very effective, and they are at the present moment the only ones which allow for the analysis in the sub-extremal range for general solutions  \cite{DRSR}\cite{Y-R}, they are nevertheless not well suited for the analysis of coupled electromagnetic-gravitational perturbations of Kerr-Newman spacetime, as separability in modes cannot be obtained  in that case (see the introduction of \cite{Giorgi7} for more details). The notable exception among the above-mentioned methods is the physical-space analysis for the wave equation in slowly rotating Kerr by  Andersson-Blue \cite{And-Mor}, which makes crucial use of the Carter tensor in Kerr and the fact that the differential operator associated to the Carter tensor \textbf{commutes with the D'Alembertian operator in Ricci-flat spacetimes} \cite{Carter2}. 

The Carter tensor \cite{Carter} is a symmetric 2-tensor $K$ on Kerr and Kerr-Newman spacetimes which satisfies the Killing equation, i.e.
\bea\label{eq:Killing-eq-intro}
\D_{(\mu}K_{\nu\gamma)}=0
\eea
where $\D$ is the covariant derivative of the metric. In virtue of the Killing equation \eqref{eq:Killing-eq-intro}, the associated differential operator $\KK(\psi):= \D_\mu (K^{\mu\nu} \D_\nu \psi)$ commutes with the D'Alembertian operator $\square_{\g}=\D_\mu \D^\mu$ in the case of Ricci-flat metric. We say that $\KK$ is a \textit{symmetry operator} for the wave equation.
In \cite{And-Mor}, Andersson-Blue develop a \textit{generalized vectorfield method} which allows for commutations with second order differential operators, and then apply it to the Carter differential operator $\KK$ and its elliptic counterpart, together with the Killing vectorfields of the Kerr metric, to derive energy and Morawetz estimates for the solution.

The main obstruction to the application of Andersson-Blue's method \cite{And-Mor} to the case of Kerr-Newman spacetime is that, even though the metric admits a Killing tensor, its associated differential operator does not in general commute with the wave equation. 
The first result of this paper is to show that, even though the Kerr-Newman spacetime is not Ricci-flat, the Carter differential operator $\KK$ associated to the Carter tensor still \textbf{commutes with the D'Alembertian operator of Kerr-Newman}.  Interestingly enough, the commutation property is not a direct consequence of the Einstein-Maxwell equations, but rather of peculiar properties of
the curvature and electromagnetic components in Kerr-Newman. We prove the following, see Theorem \ref{theo:Carter-operator-commutes-KN}:

\begin{theorem*} Even though the Kerr-Newman metric is not Ricci-flat, the Carter differential operator $\KK$ is still a symmetry operator for the wave equation.
\end{theorem*}

The above Theorem then allows to extend the physical-space analysis of Andersson-Blue \cite{And-Mor} to Kerr-Newman spacetime. More precisely, we show that the Carter differential operator $\KK$ is given by, see Proposition \ref{prop:KK-OO}, 
\bea\label{intro-operator-KK}
\KK &=& -a^2 \cos^2\th \ \square_{\g_{M, a, Q}} + \OO,
\eea
where $\OO$ is a \textit{modified Laplacian} for the Kerr-Newman metric, given in Boyer-Lindquist coordinates by
\beaa
\OO &=& \frac{1}{\sin\th} \pr_\th(\sin\th\pr_\th)+\frac{1}{\sin^2\th} \pr^2_\vphi +2a \partial_t\partial_\vphi+a^2\sin^2\th\pr^2_t.
\eeaa
In order to obtain a physical-space analysis of the wave equation in Kerr-Newman spacetime  necessary to tackle the electromagnetic-gravitational perturbations of Kerr-Newman, we prove and make crucial use of the fact that the modified Laplacian $\OO$ obtained from the Carter differential operator  is a conformal symmetry operator for the wave equation, see Proposition \ref{prop:KK-OO}:
\bea\label{intro:commutation-OO}
[\OO, (r^2+a^2\cos^2\th) \square_{\g_{M, a, Q}}]=0.
\eea
This allows to apply Andersson-Blue's generalized vectorfield method \cite{And-Mor} to the case of Kerr-Newman. We now briefly recall how the physical-space analysis is obtained in \cite{And-Mor}. 

The \textit{vectorfield method} is a robust geometrical approach to obtain energy estimates for solutions to the wave equation. 
  The energy-momentum tensor associated to the wave equation is given by
\beaa
\QQ[\psi]_{\mu\nu}&=& \pr_\mu\psi \pr_\nu \psi -\frac 1 2 \g_{\mu\nu} \pr_\lambda \psi \pr^\lambda \psi,
\eeaa
and the wave equation is satisfied if and only if the divergence of the energy-momentum tensor $\QQ[\psi]$ vanishes. For a vectorfield $X$, called multiplier, the current associated to $X$  is defined by
 \beaa
 \PP_\mu^{(X)}[\psi]&=&\QQ[\psi]_{\mu\nu} X^\nu,
  \eeaa
and its divergence is then given by
  \beaa
  \D^\mu \PP_\mu^{(X)}[\psi]= \frac 1 2 \QQ[\psi]  \c\piX,
 \eeaa
 where $\piX_{\mu\nu}=\D_{(\mu} X_{\nu)}$ is the deformation tensor of the vectorfield $X$. Recall that if $X$ is a Killing vectorfield, then $\piX=0$. 
 
 The main idea to derive estimates through the vectorfield method is to use the divergence identity for the current $ \PP_\mu^{(X)}[\psi]$ for appropriate vectorfields $X$, and obtain
\beaa
\int_{\MM}   \D^\mu \PP_\mu^{(X)}[\psi]= \int_{\partial \MM} \PP_\mu^{(X)}[\psi] \cdot n_{\partial \MM}
\eeaa
for some causal domain $\MM$.

To obtain local integrated energy decay, one wants to use a radial vectorfield  $X=\FF(r) \pr_r$, for a well chosen function $\FF$, such that the divergence of the current $  \D^\mu \PP_\mu^{(X)}[\psi]$ above is positive definite. Nevertheless, for general solutions to the wave equation, because of the complicated structure of trapping described above, this cannot be obtained. In fact, for $X=\FF(r) \pr_r$ the above current gives 
\beaa
  \D^\mu \PP_\mu^{(\FF \partial_r)}[\psi]&=&\AA |\pr_r\psi|^2 + \UU^{\a\b}(\pr_\a \psi )(\pr_\b \psi ),
\eeaa
for some positive coefficient $\AA$, and
where $\UU^{\a\b}(\pr_\a \psi )(\pr_\b \psi )$ contains only derivatives in $t, \theta, \vphi$ which are degenerate at trapping. More precisely, one obtains, see \eqref{eq:expressions-AA-VV},
\beaa
  \UU^{\a\b}(\pr_\a \psi )(\pr_\b \psi )&=& \frac{u \TT}{(r^2+a^2)^3}  |\nab \psi|^2-  u \frac{2ar}{(r^2+a^2)^2} \That( \psi) \partial_\vphi \psi
  \eeaa
  where $u$ is an increasing function of $r$ which changes sign at $r_{trap}$, the root of the polynomial $\TT$ defined in \eqref{intro:definition-TT}. Therefore the coefficient of the angular derivatives $|\nab \psi|^2$, given by $\frac{u \TT}{(r^2+a^2)^3}$, is strictly positive and presents a degeneracy of multiplicity $2$ at $r=r_{trap}$. On the other hand, the term $ u \frac{2ar}{(r^2+a^2)^2} \That( \psi) \partial_\vphi \psi$ does not have a definite sign, and for this reason one fails to obtain a positive definite current by using the vectorfield $X=\FF(r) \partial_r$.

Observe that in the case of axially symmetric solutions, the term without definite sign vanishes, as $\partial_\vphi \psi=0$, and Morawetz and energy estimates can be obtained in physical space, see Section \ref{section:axysimmetry}, without recurring to the commutation with the Carter operator.
In this case, the difficulty is in defining the function $u$ for which all the terms of the divergence, together with a zero-th order term, are positive in the sub-extremal range $a^2+Q^2<M^2$, as obtained in Kerr by Stogin \cite{Stogin}.

In the case of general solutions, Andersson-Blue \cite{And-Mor} introduced a \textit{generalized vectorfield method} which allows for higher-order symmetry operators as multipliers. In virtue of the commutation property \eqref{intro:commutation-OO}, 
the set of the second order operators $\SS_\aund$, for $\aund=1,2,3,4$, given by
\beaa
\SS_1=\partial_t^2, \qquad \qquad \SS_2=\partial_{t} \partial_{\vphi}, \qquad \qquad \SS_3=\partial_\vphi^2, \qquad \qquad \SS_4=\OO
\eeaa
are conformal symmetry operators for the wave equation in Kerr-Newman spacetime, and the commuted-solutions
\beaa
\psia:= \SS_\aund(\psi), \qquad \aund=1,2,3,4,
\eeaa
are also solutions to the wave equation.
 The generalized energy-momentum tensor  is then defined as
\beaa
\QQ[\psi]_{\aund \bund \mu\nu}&=& \pr_\mu\psia \pr_\nu \psib -\frac 1 2 \g_{\mu\nu} \pr_\lambda \psia \pr^\lambda \psib, \qquad \text{ for $\aund, \bund=1,2,3,4$. }
\eeaa
  Let $\mathbf{X}$ be a double-indexed collection of vector fields $\mathbf{X}=\{ X^{\underline{a} \underline{b}}\}$, the generalized current associated to $\mathbf{X}$ is defined by 
 \beaa
 \PP_\mu^{(\mathbf{X})}[\psi]&=&\QQ[\psi]_{\aund \bund \mu\nu} X^{\aund\bund\nu},  
  \eeaa
and its divergence is given by
  \beaa
  \D^\mu \PP_\mu^{(\mathbf{X})}[\psi]&= \frac 1 2 \QQ[\psi]_{\aund\bund}  \c \D_{(\mu} X^{\aund\bund}_{\nu)}.
 \eeaa

As for the  standard vectorfield method, the goal is to apply the divergence identity to the above generalized currents for appropriate double-indexed collection of vector fields $\mathbf{X}$, where one sums over the underlined indices $\aund$. 
Just like in the standard vectorfield method,  when applied to $\mathbf{X}=\FF^{\aund\bund}(r) \partial_r$,  the generalized current gives
\beaa
  \D^\mu \PP_\mu^{(\FF^{\aund\bund} \partial_r)}[\psi]  &=\AA^{\aund\bund} \, \pr_r\psia \pr_r\psib + \UU^{\a\b\aund\bund} \,( \pr_\a \psia) \, (\pr_\b \psib  )
 \eeaa
for some positive coefficients $\AA^{\aund\bund}$, and
where $\UU^{\a\b\aund\bund} \,( \pr_\a \psia) \, (\pr_\b \psib  )$ contains only derivatives in $t, \theta, \vphi$ of the $\psia$ which are degenerate at trapping.

The main advantage in going to the higher-order multipliers in the generalized vectorfield method is the fact that now one has the flexibility of interchanging the derivatives applied to $\psi$ through an integration by parts, and the trapped term $\UU^{\a\b\aund\bund} \,( \pr_\a \psia) \, (\pr_\b \psib  )$ can in fact be rewritten as a positive definite term. More precisely, one can write (see Lemma \ref{lemma:integration-parts})
                \beaa
 \UU^{\a\b\aund\bund}\pr_\a \psia \, \pr_\b \psib        &=& \frac 1 2  h   \big( |\partial_t \Psi|^2 + |\nab \Psi|^2 \big) + \text{boundary terms},
      \eeaa
      for a positive function $h$ and where $\Psi$ is a trapped linear combinations of second order derivatives of $\psi$, schematically given by, see \eqref{eq:Psi-choice-z},
\beaa
 \Psi&=& - \frac{2\TT}{ (r^2+a^2)^3}\pr_t^2 \psi  -\frac{2\TT}{(r^2+a^2)^3}  \OO(\psi)+  \frac{4ar}{(r^2+a^2)^2}   \That(\psi) \pr_\vphi\psi.
\eeaa
One can then use the above positivity to express the generalized current as a positive definite current for the original $\psi$ for small angular momentum $|a| \ll M$ as in \cite{And-Mor}.

\subsubsection*{Application to the Einstein-Maxwell equations}
 
In order to apply the previous techniques for the wave equation to the more interesting case of coupled electromagnetic-gravitational perturbations of Kerr-Newman, we need to extend the above physical-space analysis to the case of coupled generalized Regge-Wheeler equations describing the perturbations.

As a consequence of the Einstein-Maxwell equations, in \cite{Giorgi7} we showed that there exists a pair of gauge-invariant tensorial quantities, denoted $\pf$ and $\qf^\F$, which satisfy the following schematic system of equations, see Theorem \ref{main-theorem-RW} for the precise formulation,
\bea\label{eq:intro-gRW}
\begin{split}
 \square_\g\pf-i  \frac{2a\cos\th}{|q|^2}\nab_T \pf  -V_1  \pf &=4Q^2 \frac{\ov{q}^3 }{|q|^5} \left(  \ov{\DD} \c  \qf^\F  \right) + \lot \\
\square_\g\qf^\F-i  \frac{4a\cos\th}{|q|^2}\nab_T \qf^\F -V_2  \qf^\F &=-   \frac 1 2\frac{q^3}{|q|^5} \left(  \DD \hot  \pf  \right) + \lot
\end{split}
 \eea
 where $q=r+i a \cos\th$, $\ov{q}=r-i a \cos\th$, $|q|^2=r^2+a^2\cos^2\th$, and $\ov{\DD} \c$ and $\DD \hot$ are angular operators responsible for the coupling between the two quantities. In Section \ref{section:spheroidal}, we explain why those angular operators on the right hand side of the above equations prevent the separability in modes in Kerr-Newman.
 
 Nevertheless, the precise structure of the right hand side has good properties when interpreted in \textit{physical-space}, in terms of the energy-momentum tensor of the equations. Indeed, the combined energy-momentum tensor defined as
 \beaa
 \QQ[\pf, \qf^\F]_{\mu\nu}&:=& \QQ[\pf]_{\mu\nu}+8Q^2 \QQ[\qf^\F]_{\mu\nu},
 \eeaa
has good divergence properties. More precisely, the highest-order coupling terms at the level of divergence of the energy-momentum tensor cancel out in physical-space, see Section \ref{section:cancellation}.

There is still one problem in extending the Andersson-Blue method described above to the case of the generalized Regge-Wheeler equations as in \eqref{eq:intro-gRW}. The Carter differential operator and the modified Laplacian $\OO$, which is a conformal symmetry for the D'Alembertian $\square_{\g}$ is not a conformal symmetry for the system of equations, again because of the presence of the angular operators on the right hand side. Those angular operators do not commute with the modified Laplacian $\OO$, but rather their commutator involves a modified Gauss curvature term, denoted $\Kh$. The scalar $\Kh$ is defined as a curvature component of the (non-integrable) horizontal structure associated to the principal null frame in Kerr or Kerr-Newman, see \cite{GKS}, and it reduces to the Gauss curvature of the spheres in the case of spherically symmetric background.

Even though the modified Laplacian $\OO$ does not commute with the right hand side of the generalized Regge-Wheeler equations, their symmetric structure allows to define modified Laplacian operators involving Gauss curvature which do commute with the system \eqref{eq:intro-gRW}. They are given by 
\beaa
\widehat{\pf}:=\big(  \OO +(c+3)|q|^2\Kh \big) \pf , \qquad \widehat{\qf^\F}:=\big(  \OO +c \ |q|^2\Kh \big) \qf^\F
\eeaa
for any real number $c$, and they can be combined to the other symmetry operators to apply Andersson-Blue's method to the system of perturbations of Kerr-Newman.

We finally show that with the above definition of modified symmetry operators the integration by parts which allows to create positive definite terms in $\UU^{\a\b\aund\bund}\pr_\a \psia \, \pr_\b \psib$ can still be performed in this more general setting.

\subsubsection*{Structure of the paper}

This paper is organized as follows.

 In Section \ref{section:Kerr-Newman},  we define Kiling tensors and their associated differential operators in a general manifold, and we compute the commutator with the D'Alembertian operator in terms of the Ricci curvature of the metric. We then show that the Carter operator is a symmetry operator in Kerr-Newman.

 In Section \ref{subsection-preliminaries-wave} we collect the main properties of the wave equation in Kerr-Newman and the vectorfield method. We also derive the equations of trapped null geodesics in Kerr-Newman, and obtain energy and local decay estimates in physical space for axially symmetric solutions to the wave equation in the sub-extremal range $a^2+Q^2< M^2$ and for general solutions for slowly rotating $|a| \ll M$ in Kerr-Newman spacetime.

In Section \ref{Einstein-Maxwell-equations}, we show how to apply the above physical-space analysis to the generalized Regge-Wheeler system describing the coupled electromagnetic-gravitational perturbations of Kerr-Newman black hole.

\bigskip

\noindent\textbf{Acknowledgements.} The author is grateful to Sergiu Klainerman for useful discussions. The author acknowledges the support of NSF grant DMS 2006741.

\section{The Carter tensor in Kerr-Newman spacetime}\label{section:Kerr-Newman}

In this section, we recall the definitions of Killing tensors and associated Killing differential operators for a general Lorentzian manifold $(\MM, \g)$, without any assumption on the curvature of $\g$. 
We then collect the commutator of the Killing differential operator with the D'Alembertian $\square_\g$ of the metric $\g$.

We then recall the main properties of the Kerr-Newman spacetime, and explicitly define its Killing tensor, known as Carter tensor. We will show that, because of the special structure of the Carter tensor and the electromagnetic and curvature components, the associated Carter operator is a symmetry of the wave equation.
Using a convenient expression for the Carter operator, we then define a conformal symmetry operator which can be interpreted as a modified Laplacian for the Kerr-Newman metric.

\subsection{Killing tensors and differential operators}\label{section-Killing}

Let $(\MM, \g)$ be a Lorentzian\footnote{The computations in this section are valid in any pseudo-Riemannian manifold $(\MM, \g)$.} manifold and let $\D$ denote the covariant derivative of $\g$. Recall that a vectorfield $X$ on $\MM$ is called \textit{Killing} if the Lie derivative of the metric with respect to $X$ vanishes, i.e. $\piX:= \mathcal{L}_X \g_{\mu\nu}=\D_{(\mu} X_{\nu)}=0$, where $\piX_{\mu\nu}$ is called the \textit{deformation tensor} of the vectorfield $X$. 

 The notion of Killing vectorfield can be extended to 2-tensors.

\begin{definition}
A Killing tensor for $(\MM, \g)$ is a symmetric $2$-tensor $K$ which satisfies the Killing equation:
\bea\label{Killing-equation}
\D_{(\mu} K_{\nu \rho)}=0.
\eea
\end{definition}

In the case of a Killing vectorfield $X$, the flow of $X$ represents a local isometry of $(\MM, \g)$, and similarly a hidden symmetry is associated to a Killing tensor $K$.

\medskip

Let $(\MM, \g)$ be a Lorentzian manifold possessing a Killing tensor $K$. 
We define a second order differential operator associated to $K$.
\begin{definition}\label{definition-Killing-operator} Given a Killing tensor $K$,  its associated second order differential operator $\KK$ applied to any tensor $\Psi$ in $\MM$ is defined by
\bea\label{definition-KK}
\mathcal{K}(\Psi)= \D_{\mu}( K^{\mu\nu} \D_\nu(\Psi)).
\eea
\end{definition}

As a consequence of $K$ being Killing, the above operator $\KK$ enjoys favorable properties of commutation with the D'Alembertian operator $\square_\g=\D^\mu \D_\mu$, which are also related to the Ricci curvature of the metric $\g$.

\begin{proposition}\label{commutation-KK-square} Let $(\MM, \g)$ be a Lorentzian manifold with a Killing tensor $K$. Then the commutator between the differential operator $\KK$ and the D'Alembertian operator $\square_\g$ for a scalar function $\phi$ is given by
\beaa
[\KK, \square_\g] \phi&=& \Big(\big( \D^\a \R-\frac 4 3  \D^\mu {{\R}_{\mu}}^\a \big) K_{\a\nu}+ \frac 2 3 \big({{\R}_{\mu}}^\ep  \D^\mu  K_{\ep \nu}-   {{\R}_{\nu}}^\ep  \D^\mu K_{\mu \ep}-  \D^\a {\R^\ep}_{\nu}   K_{\a \ep}\big) \Big) \D^\nu \phi
\eeaa
where $\R$ denotes the Ricci curvature or the scalar curvature depending if it appears a 2-tensor or a scalar respectively. 
\end{proposition}
\begin{proof} See Appendix \ref{proof-prop-commutators}.
\end{proof}

\subsection{The case of vacuum and electrovacuum spacetimes}

We specialize the above commutator to the case of metrics satisfying the Einstein vacuum equation or the Einstein-Maxwell equation.

\begin{definition}\label{definition:eve-maxwell} A Lorentzian manifold $(\MM, \g)$ is a solution to the \textup{Einstein vacuum equation} if its Ricci curvature vanishes identically, i.e.
\bea
\R_{\mu\nu}=0.
\eea
A Lorentzian manifold $(\MM, \g)$ is a solution to the \textup{Einstein-Maxwell equation} if 
\bea\label{Einstein-Maxwell}
\R_{\mu\nu}=2 \F_{\mu \lambda} {\F^{\lambda}}_\nu - \frac 1 2 \g_{\mu\nu} \F^{\alpha\beta} \F_{\alpha\beta}, 
\eea
where $\F$ is a 2-form on $\MM$, denoted electromagnetic tensor, satisfying the Maxwell equations:
\bea\label{Maxwell}
\D_{[\mu} \F_{\nu \lambda]}=0, \qquad \D^\mu \F_{\mu\nu}=0.
\eea

We will refer to the above as \textup{vacuum} and \textup{electrovacuum} spacetime respectively.
\end{definition}

 As a consequence of Proposition \ref{commutation-KK-square}, if a vacuum spacetime possesses a Killing tensor $K$, then 
 the differential operator $\KK$ commutes with the D'Alembertian of $\g$:
 \beaa
 [\KK, \square_\g]\phi=0.
 \eeaa
For this reason, $\KK$ is referred to as a \textit{symmetry operator}, as it sends solution to the wave equation to solutions: if $\square_\g \phi=0$ then $\square_\g (\KK(\phi))=0$.

As a consequence of the Einstein-Maxwell equations \eqref{Einstein-Maxwell} and \eqref{Maxwell}, the curvature of an electrovacuum spacetime satisfies
\beaa
\D^\mu \R_{\mu\nu}=0, \qquad \R=0.
\eeaa
Then if an electrovacuum spacetime possesses a Killing tensor $K$, the commutator between $\KK$ and $\square_\g$ according to Proposition \ref{commutation-KK-square} becomes
\bea\label{commutator-electrovacuum}
[\KK, \square_\g] \phi&=& \frac 2 3\big({{\R}_{\mu}}^\ep  \D^\mu  K_{\ep \nu}-   {{\R}_{\nu}}^\ep  \D^\mu K_{\mu \ep}-  \D^\a {\R^\ep}_{\nu}   K_{\a \ep}\big) \D^\nu \phi.
\eea
Moreover, the Maxwell equations \eqref{Maxwell} do not imply the vanishing of the right hand side of \eqref{commutator-electrovacuum}.
Therefore in this case $\KK$ cannot be interpreted as a symmetry operator, as the commutator depends on the form of the Killing tensor $K$.

Famous examples of vacuum and an electrovacuum spacetimes which possess a Killing tensor are the Kerr and Kerr-Newman solutions respectively. The Killing tensor in Kerr was discovered by Carter \cite{Carter}, and it is then referred to as \textit{Carter tensor}.

\subsection{The Kerr-Newman spacetime}

We review here the Kerr-Newman metric and associated properties, see also \cite{Civin}. The Kerr-Newman metric depends on three physical parameters: the mass $M$, the angular momentum $a$ and the charge $Q$. We consider here the subextremal family of Kerr-Newman spacetimes which represent a charged rotating black hole.

\subsubsection{The manifold and the metric}

For $a^2+Q^2 <M^2 $, the Kerr-Newman metric in Boyer-Lindquist coordinates $(t, r,  \th, \vphi)$ takes the form
\bea\label{metric-KN}
\g_{M, a, Q}=-\frac{\Delta}{|q|^2}\left( dt- a \sin^2\th d\vphi\right)^2+\frac{|q|^2}{\Delta}dr^2+|q|^2 d\th^2+\frac{\sin^2\th}{|q|^2}\left(a dt-(r^2+a^2) d\vphi \right)^2,
\eea
where\footnote{Observe that what we denote by $|q|^2=r^2+a^2\cos^2\th$ is normally denoted in the literature as $\rho^2$. We avoid using the letter $\rho$ for this metric component as it is also used as a curvature component later.}
\bea
q&=&r+ i a \cos\th \label{definition-q}, \qquad |q|^2=r^2+a^2\cos^2\th,
\eea
and
\beaa
\Delta &=& r^2-2Mr+a^2+Q^2=(r-r_{+}) (r-r_{-}), \qquad r_{\pm}=M\pm \sqrt{M^2-a^2-Q^2}.
\eeaa

We recall the ambient manifold with boundary $\MM$, diffeomorphic to $\mathbb{R}^{+} \times \mathbb{R} \times \mathbb{S}^2$, and the Kerr star coordinates $(t^{*}, r, \th, \vphi^*)$ with relations $t(t^*, r)=t^*-\overline{t}(r)$, $\vphi(\vphi^*, r)=\vphi^*-\overline{\vphi}(r)$ modulo $2\pi$. For the explicit form see \cite{Civin}. When expressed in Kerr star coordinates, the metric \eqref{metric-KN} extends smoothly to the event horizon $\mathcal{H}^+$ defined as the boundary $\partial \MM=\{ r=r_{+}\}$.

The metric $\g_{M, a, Q}$, together with an appropriate 2-form $\F$, satisfies the Einstein-Maxwell equations \eqref{Einstein-Maxwell} and \eqref{Maxwell}.
Observe that the Kerr-Newman family reduces to the Kerr metric when $Q=0$, to the Reissner-Nordstr\"om metric when $a=0$, and  to the Schwarzschild metric for $a=Q=0$.

\subsubsection{The Killing vectorfields}

The coordinate vectorfields $T=\partial_{t^*}$ and $Z=\partial_{\vphi^*}$ coincide with the coordinate vectorfields $\partial_t$ and $\partial_\vphi$ in Boyer-Lindquist coordinates, which are Killing for the metric \eqref{metric-KN}. The stationary Killing vectorfield $T=\partial_t$ is asymptotically timelike as $r \to \infty$, and spacelike close to the horizon, in the ergoregion $\{ \Delta - a^2\sin^2\th <0\}$.

Since $T=\partial_t$ is not timelike in the ergoregion, we recall instead the definition of  \textit{Hawking vectorfield}, which is timelike in the whole exterior and null on the horizon.
 \begin{proposition} 
 \lab{Lemma:Hawking-vf}
 The Hawking vectorfield $\That$
 \bea
\lab{define:That}
\That:&=&\pr_t+\frac{a}{r^2+a^2} \pr_\vphi
\eea
 is  timelike        for $\{r>r_+\}$ and null   on the horizon $\{r=r_+\}$.  More precisely
   \bea\label{eq:g-That-That}
   \g_{M, a, Q}(\That, \That)&=& -\De\frac{|q|^2}{ (r^2+a^2)^2}.
   \eea   
   \end{proposition} 
 \begin{proof}
Denoting $\g=\g_{M, a, Q}$ and using that
 \beaa
 \g_{tt}=   -\frac{\Delta-a^2\sin^2\theta}{|q|^2}, \qquad \g_{t\vphi}=
\frac{ \Delta-(r^2+a^2) }{|q|^2} a\sin^2\theta    , \qquad \g_{\vphi\vphi}=\frac{(r^2+a^2)^2-a^2\sin^2\th \De}{|q|^2}\sin^2\theta
 \eeaa
 we obtain 
  \beaa
   \g(\That, \That) &=&\g_{tt} + \frac{2a}{r^2+a^2} \g_{t\phi}+\frac{a^2}{(r^2+a^2)^2} \g_{\vphi\vphi}\\
   &=&  -\frac{\Delta-a^2\sin^2\theta}{|q|^2} + \frac{2a^2}{r^2+a^2} \frac{ \Delta-(r^2+a^2) }{|q|^2} \sin^2\theta+\frac{a^2}{(r^2+a^2)^2}  \frac{(r^2+a^2)^2-a^2\sin^2\th \De}{|q|^2}\sin^2\theta\\
      &=&  -\frac{\Delta}{|q|^2}   + \frac{a^2}{r^2+a^2} \frac{ 2\Delta}{|q|^2} \sin^2\theta -\frac{a^2}{(r^2+a^2)^2}  \frac{a^2\sin^2\th \De}{|q|^2}\sin^2\theta\\
            &=&  -\frac{\Delta}{|q|^2 (r^2+a^2)^2} ((r^2+a^2)^2  - 2a^2 (r^2+a^2) \sin^2\theta +a^4\sin^4\th)= -\frac{\Delta |q|^4}{|q|^2 (r^2+a^2)^2} \\
            &=& -\frac{\Delta |q|^2}{ (r^2+a^2)^2},
   \eeaa
 as stated.
   \end{proof}
As  a consequence  of the above, we also   deduce that  the Killing vectorfield 
\bea
\That_\HH:=T+\om_\HH Z, \qquad \text{with} \quad \om_\HH=\frac{a}{r_{+}^2+a^2}=\frac{a}{2Mr_{+}-Q^2},
\eea 
where $\om_\HH$ is the angular velocity  of the horizon,     is null on the horizon and timelike in a small neighborhood of it in the exterior. Note that along $\mathcal{H}^+$, we have
\beaa
\D_{\That_\HH} \That_\HH=\kappa \That_{\HH}, \qquad \kappa=\frac{r_{+}-r_{-}}{2(r_{+}^2+a^2)},
\eeaa
where $\kappa$ is the surface gravity, which is positive in the sub-extremal range and vanishes in the extremal case.

\subsubsection{The principal null frames and horizontal structures}

The Kerr-Newman metric is a spacetime of Petrov Type D, i.e. its Weyl curvature can be diagonalized with two linearly independent eigenvectors, the so-called principal null directions.

The vectorfields 
\bea
\lab{eq:Out.PGdirections-Kerr}
L=e_4=\frac{r^2+a^2}{\Delta}\pr_t+\pr_r+\frac{a}{\Delta}\pr_\vphi, \qquad \underline{L}=e_3=\frac{r^2+a^2}{|q|^2}\pr_t-\frac{\Delta}{|q|^2}\pr_r+\frac{a}{|q|^2}\pr_\vphi
\eea
define principal null directions, with the normalization $\g(e_3, e_4)=-2$. The vectorfield $\frac{\De}{r^2+a^2} e_4$ extends smoothly to $\HH^{+}$ to be parallel to the null generator, while $\frac{r^2+a^2}{\De}e_3$ extends smoothly to $\HH^{+}$ to be transversal to it.

We complete the above null frame with the following vectorfields
\bea
\lab{eq:Out.PGdirections-Kerr-12}
  e_1=\frac{1}{|q|}\pr_\th,\quad e_2=\frac{a\sin\th}{|q|}\pr_t+\frac{1}{|q|\sin\th}\pr_\vphi,
\eea
which represent an orthonormal frame on the orthogonal space spanned by $e_3$ and $e_4$.

The frame $\{e_a, e_3, e_4\}$ for $a=1,2$ as defined in \eqref{eq:Out.PGdirections-Kerr} and \eqref{eq:Out.PGdirections-Kerr-12} describes a null frame which satisfies
\beaa
&&\g(e_3, e_3)=\g(e_4, e_4)=0, \qquad \g(e_3, e_4)=-2, \\
&& \g(e_a,e_3)=\g(e_a, e_4)=0, \qquad \g(e_a, e_b)=\de_{ab}, \qquad a=1,2.
\eeaa

We say that a vectorfield $X$ is horizontal if $\g(X, e_4)=\g(X, e_3)=0$.  Observe that the commutator of two horizontal vectorfields may fail to be horizontal. We say that $(L, \underline{L})$ is integrable if the commutator of two horizontal vectorfield is horizontal.

In the case of Kerr-Newman spacetime, the principal null frame \eqref{eq:Out.PGdirections-Kerr} is not integrable, and its horizontal vector space does not span a sphere, but rather a 2-plane distribution. We refer to it as a horizontal structure. On it, we can define the horizontal covariant derivative by the projection of the covariant derivative to the horizontal structure: for  $X$ in the tangent space of $M$ and Y horizontal 
 \bea\label{eq:def-horizontal-cov-der}
 \Ddot_X Y&:=&^{(h)}(\D_X Y)=\D_X Y+ \frac 1 2 \g(\D_X Y,\Lb)L+ \frac 1 2   \g(\D_X Y,L)\Lb,
 \eea
where $^{(h)}(\D_X Y)$ is the projection to the horizontal structure. We denote $\nab_4 Y=$$^{(h)}(\D_L Y)$,  $\nab_3 Y=$$^{(h)}(\D_\Lb Y)$, $\nab_a Y=$$^{(h)}(\D_{e_a} Y)$.  

Similarly, we can extend the above definition to horizontal covariant $k$-tensors.
We denote the set of horizontal tensors on $\MM$ by $\sk_k$. We define the duals of $f \in \sk_1$ and $u \in \sk_2$ by
\beaa
\dual f_{a}&=&\in_{ab}f_b,\qquad (\dual u)_{ab}=\in_{ac} u_{cb},
\eeaa
where $\in_{ab}=\frac 1 2 \in_{ab34}$ is the induced volume form on the horizontal structure.  For $f \in \sk_1$ and $u \in \sk_2$
we  define the frame dependent   operators,   
\beaa
\div f &=&\de^{ab}\nab_b f_a,\qquad 
\curl f=\in^{ab}\nab_af_b,\\
(\nab\hot f)_{ba}&=&\frac 1 2 \big(\nab_b f_a+\nab_a  f_b-\de_{ab}( \div f)\big)\\
(\div u)_a&=& \de^{bc} \nab_b u_{ca}.
\eeaa

\begin{definition}   Given  an orthonormal basis of horizontal vectors $e_1, e_2$ we define 
   the Hodge type operators,  see \cite{Ch-Kl}:
\begin{itemize}
\item  $\DDd_1 $ takes $\sk_1$ into\footnote{Recall that $\sk_0$ refers to pairs of scalar functions $(a,b)$.}  $\sk_0$:
\beaa
\DDd_1 \xi =(\div\xi, \curl \xi),
\eeaa
\item  $\DDd_2 $ takes $\sk_2$ into $\sk_1$:
\beaa
(\DDd_2 \xi)_a =\nab^b \xi_{ab},
\eeaa
\item $\DDs_1$ takes  $\sk_0 $ into  $\sk_1$:
\beaa
\DDs_1( f, f_*) &=& -\nab_a f+\in_{ab}  \nab_b  f_*,
\eeaa
\item 
$\DDs_2 $ takes  $\sk_1$ into $\sk_2$:
\beaa
\DDs_2 \xi&=& -\nab\hot \xi.
\eeaa
\end{itemize}   
\end{definition}

See \cite{GKS} or \cite{Giorgi7} for more details.

\subsubsection{The Ricci, electromagnetic and curvature components}

As in \cite{Ch-Kl}, we use standard notations to define the Ricci coefficients of the null pair frame as
  \beaa
\chib_{ab}&=&\g(\D_a\Lb, e_b),\qquad \chi_{ab}=\g(\D_aL, e_b),\\
\xib_a&=&\frac 1 2\g(\D_\Lb\Lb, e_a),\qquad \xi_a=\frac 1 2 \g(\D_L L, e_a),\\
\omb&=&\frac 1 4 \g(\D_\Lb\Lb, L),\qquad\quad  \om=\frac 1 4\g(\D_L L, \Lb),\qquad \\
\etab_a&=&\frac 1 2\g(\D_L\Lb, e_a),\qquad \quad \eta_a=\frac 1 2 \g(\D_\Lb L, e_a),\qquad\\
 \ze_a&=&\frac 1 2 \g(\D_{a}L,\Lb).
\eeaa
using the short hand notation $\D_a=\D_{e_a}, a=1,2$. 

Observe that in the case of Kerr and Kerr-Newman spacetime, the 2-tensors $\chi_{ab}$ and $\chib_{ab}$ associated to the principal null frame are not symmetric, as a consequence of the fact that the space which is orthogonal to the principal null frame is not integrable, see 
\cite{GKS}.
Following \cite{GKS}, we  introduce the  notations
\beaa
\trch:=\de^{ab}\chi_{ab}, \qquad  \trchb:=\de^{ab} \chib_{ab},\\
\atr\chi:=\in^{ab} \chi_{ab}, \qquad \atr\chib:=\in^{ab} \chi_{ab}.
\eeaa
In particular we can write
\beaa
\chi_{ab}&=&\chih_{ab}+\frac 1 2 \de_{ab}\, \trch+\frac 12 \in_{ab}\atrch,\\
\chib_{ab}&=&\chibh_{ab}+\frac 1 2 \de_{ab}\, \trchb+\frac 12 \in_{ab}\atrchb,
\eeaa
where $\chih$ and $\chib$ is the symmetric traceless part of $\chi$ and $\chib$ respectively.

We define the electromagnetic components relative to the null frame as
\beaa
\bF_a&=&\F(e_a, e_4), \qquad \bbF_a=\F(e_a, e_3) \\
\rhoF&=&\frac 1 2 \F(e_3, e_4),\qquad \dual\rhoF=
\frac 1 2   \dual \F (e_3, e_4)
\eeaa
where $\F$ is the electromagnetic tensor in the Einstein-Maxwell equations  \eqref{Einstein-Maxwell} and \eqref{Maxwell} and $\dual \F$ denotes the Hodge dual on $(\MM, \g)$ of $\F$, defined by $\dual \F_{\a\b}=\frac 1 2 \in_{\mu\nu\a\b} \F^{\mu\nu}$. 

We define the curvature components relative to the null frame as 
\beaa
\a_{ab}&=&\W(e_4,e_a,e_4,e_b), \qquad \aa_{ab}=\W(e_3,e_a,e_3,e_b),\\
\b_a&=&\frac 1 2 \W(e_a,e_4,e_3,e_4), \qquad \bb_a=\W(e_a, e_3, e_3, e_4)\\
 \rho&=&\frac 1 4 \W(e_3,e_4,e_3,e_4),\qquad \dual\rho=
\frac 1 4  \dual \W(e_3,e_4,e_3,e_4)
\eeaa
where $\W$ is the Weyl curvature, and $\dual \W$ denotes the Hodge dual on $(\MM, \g)$ of $\F$, defined by ${}^{\ast}\W_{\al\be\mu\nu}=\frac{1}{2}{\in_{\mu\nu}}^{\rho\si}\W_{\al\be\rho\si}$.

\begin{lemma}\label{lemma:Kerr-Newman}  The Kerr-Newman metric has the following values of  the Ricci, electromagnetic and curvature components.
\begin{itemize}
\item The following quantities vanish:\footnote{The vanishing of the quantities corresponds to the fact that the Kerr-Newman spacetime is  of Petrov Type D.}
\bea
&&\chih=\chibh=\xi=\xib=0,  \qquad \a=\b=\bb=\aa=0, \qquad  \bF= \bbF=0.
\eea
The non-vanishing electromagnetic and curvature components take the following values:
\bea
&&\rhoF^2+\dual\rhoF^2=\frac{Q^2}{|q|^4}\\
&&\rho= \frac{1}{|q|^6} (-2Mr^3+2Q^2r^2+6M  a^2 \cos^2\th r-2Q^2a^2\cos^2\th),\\
 &&\dual \rho= \frac{a\cos\th }{|q|^6} (6Mr^2-4Q^2 r- 2Ma^2 \cos^2\th).
\eea
\item The Ricci coefficients defined with respect to the  principal null frame \eqref{eq:Out.PGdirections-Kerr} take the following values:
\bea
&& \trch=\frac{2r}{|q|^2},\quad \atrch=\frac{2a\cos\th}{|q|^2}
, \qquad  \trchb=-\frac{2r\Delta}{|q|^4}, \quad \atrchb=\frac{2a\Delta\cos\th}{|q|^4},\\
&& \omb = \frac{a^2\cos^2\th(r-M)+Mr^2-a^2r-Q^2r}{|q|^4}, \qquad \om=0, \qquad \etab=-\ze.
\eea

\item In the orthonormal frame $e_a$ for $a=1,2$ defined in \eqref{eq:Out.PGdirections-Kerr-12}, the Ricci coefficients have the following components:
\beaa
\eta_1&=& -\frac{a^2\sin\th \cos\th}{|q|^3}, \qquad \eta_2=\frac{a\sin\th r}{|q|^3}, \\
\dual \eta_1&=& \frac{a\sin\th r}{|q|^3}, \qquad \dual \eta_2= \frac{a^2\sin\th \cos\th}{|q|^3},\\
\etab_1&=& -\frac{a^2\sin\th \cos\th }{|q|^3}, \qquad \etab_2 =-\frac{a\sin\th r}{|q|^3},\\
\dual \etab_1&=& -\frac{a\sin\th r}{|q|^3}, \qquad \dual \etab_2 =\frac{a^2\sin\th \cos\th}{|q|^3}.
\eeaa

\item The derivatives of the coordinates $r$ and $\theta$ with respect to the frame defined in \eqref{eq:Out.PGdirections-Kerr} and \eqref{eq:Out.PGdirections-Kerr-12} satisfy the following relations:
\bea\label{eq:derivatives-r-th}
 e_{3} (r )=\frac{|q|^2}{2r} \trchb  , \qquad e_{4} (r )=\frac{|q|^2}{2r} \trch, \qquad e_a(r)=0, \\
e_a  (a^2\cos^2\th)=|q|^2 (\eta+\etab)_a, \qquad e_3 (a^2\cos^2\th)=e_4 (a^2\cos^2\th)=0.
\eea
\item Finally, the orthonormal frame $e_a$ for $a=1,2$ defined in \eqref{eq:Out.PGdirections-Kerr-12}
satisfy the following:
\bea\label{eq:expressions-nab11}
\nab_{e_1} e_1&= \nab_{e_1} e_2=0, \qquad \nab_{e_2} e_1 =\La e_2, \qquad \nab_{e_2} e_2=-\La e_1,
\eea
where
\bea
\La := \frac{r^2+a^2}{|q|^3}\cot\th
\eea
\end{itemize}
\end{lemma}
\begin{proof} See \cite{Giorgi7} and \cite{GKS}.
\end{proof}

\subsubsection{Complex notations}

 We denote by $\sk_k(\CCC)$ the set of complex anti-self dual $k$-tensors on $\MM$. More precisely, $a+i b\in \sk_0(\CCC)$ is a complex scalar function on $\MM$ with $(a, b) \in \sk_0$, $F=f+ i \dual f \in \sk_1(\CCC)$ is a complex anti-self dual 1-tensor on $\MM$ with $f \in \sk_1$, and $U=u+ i \dual u \in \sk_2(\CCC)$ is a complex anti-self dual symmetric traceless 2-tensor on $\MM$ with $u \in \sk_2$.

 We extend the definitions for the Ricci, electromagnetic and curvature components to the complex case by using the anti-self dual tensors, by defining 
\beaa
 X=\chi+i\dual\chi, \quad \Xb=\chib+i\dual\chib, \quad H=\eta+i\dual \eta, \quad \Hb=\etab+i\dual \etab, \quad \Xi=\xi+i\dual\xi, \quad \Xib=\xib+i\dual\xib\\
 \BF= \bF + i \dual \bF,  \quad  \BBF= \bbF + i \dual \bbF, \qquad \PF=\rhoF + i \dual \rhoF \\
A=\a+i\dual\a,  \quad \Ab=\aa+i\dual\aa, \qquad  B=\b+i\dual\b, \quad \Bb=\bb+i\dual\bb,\qquad  P=\rho+i\dual\rho. 
\eeaa

We also define the complexified version of the $\nab$ horizontal derivative as
\beaa
\DD= \nab + i \dual \nab, \qquad \DDov= \nab- i \dual \nab
\eeaa
More precisely,
\begin{itemize}
\item For $a+ib \in \sk_0(\CCC)$ 
\beaa
\DD(a+ib) &:=& (\nabla+i\dual\nabla)(a+ib), \qquad \DDov(a+ib) := (\nabla-i\dual\nabla)(a+ib)\
\eeaa

\item For $F= f+ i \dual f \in\sk_1 (\mathbb{C})$, 
\beaa
\DDov\c(f+i\dual f) &:=& (\nabla-i\dual\nabla)\c(f+i\dual f)= 2 \big(\div f + i \curl f \big), \\
\DD\hot(f+i\dual f) &:=& (\nabla+i\dual\nabla)\hot(f+i\dual f)=2\big(\nab \hot f + i \dual (\nab \hot f)\big).
\eeaa

\item For  $U= u + i \dual u \in \sk_2 (\mathbb{C})$, 
\beaa
\DDov (u+i\dual u) &:=& (\nabla- i\dual\nabla) (u+i\dual u)= 2\big( \div u + i \dual(\div u) \big)
\eeaa
\end{itemize}

 For $F \in \sk_1(\CCC)$, the operator $- \DD\hot$ is formally adjoint to the operator $\DDov \c U$ applied to $U \in \sk_2(\CCC)$, as shown in \cite{Giorgi7}.
 \begin{lemma}[Lemma 2.11 in \cite{Giorgi7}]\label{lemma:adjoint-operators}
 For $F=f+i\dual f \in \sk_1(\CCC)$ and $U=u+i\dual u \in \sk_2(\CCC)$, we have
   \bea
 ( \DD \hot   F) \c   \ov{U}  &=&  -F \c (\DD \c \ov{U}) -( (H+\Hb ) \hot F )\c \ov{U} +\D_\a (F \c \ov{U})^\a.
 \eea
 \end{lemma}

\subsection{The Carter tensor and differential operator in Kerr-Newman}

We define the Carter tensor of Kerr-Newman spacetime and show that it is a Killing tensor for its metric. We then also define the Carter differential operator associated to it and show that, even though Kerr-Newman is not Ricci-flat, the Carter operator is a symmetry operator in Kerr-Newman.

\begin{definition} The Carter tensor associated to $(\MM, \g_{M, a, Q})$ is the following symmetric 2-tensor $K$ defined by
\bea\label{definition-carter}
K&=& -(a^2\cos^2\th)  \g_{M, a, Q} +|q|^2 \left( e_1 \otimes  e_1 + e_2 \otimes e_2 \right) 
\eea
where $e_1$ and $e_2$ are defined in \eqref{eq:Out.PGdirections-Kerr-12}.
\end{definition}

By defining the symmetric tensor 
\bea\label{eq:first-definition-O}
O^{\a\b}&:=& |q|^2 ( e_1^\a e_1^\b + e_2^\a e_2^\b)
\eea
we can write from \eqref{definition-carter},
\bea\label{eq:expression-K-O}
K&=&  -(a^2\cos^2\th)  \g_{M, a, Q} + O.
\eea

Since $\g_{ab}=\de_{ab}$, $\g_{a3}=\g_{a4}=0$ for $a=1,2$, and $\g_{34}=-2$, we obtain from \eqref{definition-carter} that the tensor $K$ has the following  components:
\beaa
K_{ab}= r^2 \de_{ab}, \qquad  K_{34}= 2(a^2\cos^2\th) , \qquad K_{a3}=K_{a4}=K_{33}=K_{44}=0.
\eeaa
In particular 
\beaa
\tr K=2(r^2 -a^2\cos^2\th).
\eeaa

\begin{proposition}\label{prop:Carter-tensor-Killing} The Carter tensor $K$ defined by \eqref{definition-carter} satisfies
\beaa
\D_{(\mu} K_{\nu \rho)}=0
\eeaa
and therefore it is a Killing tensor for $(\MM, \g_{M, a, Q})$. 
\end{proposition}
\begin{proof} We decompose the symmetric 3-tensor $\Pi_{\mu\nu\rho}=\D_{(\mu} K_{\nu \rho)}=\frac 1 3(\D_{\mu} K_{\nu\rho}+\D_{\nu} K_{\rho\mu}+\D_{\rho} K_{\mu\nu} ) $ in the null frame by making use of the  Ricci formulae \cite{Ch-Kl},
\bea\label{conn-coeff}
\D_a e_b&=&\nab_a e_b+\frac 1 2 \chi_{ab} e_3+\frac 1 2  \chib_{ab}e_4,\nn\\
\D_a e_4&=&\chi_{ab}e_b -\ze_a e_4,\nn\\
\D_a e_3&=&\chib_{ab} e_b +\ze_ae_3,\nn\\
\D_3 e_a&=&\nab_3 e_a +\eta_a e_3+\xib_a e_4,\nn\\
\D_3 e_3&=& -2\omb e_3+ 2 \xib_b e_b,\label{ricci}\\
\D_3 e_4&=&2\omb e_4+2\eta_b e_b,\nn\\
\D_4 e_a&=&\nab_4 e_a +\etab_a e_4 +\xi_a e_3,\nn\\
\D_4 e_4&=&-2 \om e_4 +2\xi_b e_b,\nn\\
\D_4 e_3&=&2 \om e_3+2\etab_b e_b.\nn
\eea 
We compute
\beaa
3 \Pi_{abc}&=&\D_{a} K_{bc}+\D_{b} K_{c a}+\D_{c} K_{ab} \\
&=&\nab_{a} K_{bc}- K_{\D_a b c}-K_{b \D_ac}+\nab_{b} K_{c a}-K_{\D_b c a}- K_{c \D_b a}+\nab_{c} K_{ab} - K_{\D_c a b} - K_{a \D_c b}\\
&=&\nab_{a} K_{bc}+\nab_{b} K_{c a}+\nab_{c} K_{ab} =\nab_{a}(r^2 \de_{bc})+\nab_{b} (r^2 \de_{c a})+\nab_{c}(r^2 \de_{ab} )\\
&=&2 r \left( \nab_{a}r \de_{bc}+\nab_{b} r \de_{c a}+ \nab_{c}r \de_{ab} \right)=0,
\eeaa
since $e_1(r)=e_2(r)=0$. 
We compute
\beaa
3 \Pi_{ab3}&=&\D_{a} K_{b3}+\D_{b} K_{3 a}+\D_{3} K_{ab} \\
&=&\nab_{a} K_{b3}-K_{\D_a b 3} - K_{b \D_a 3}+\nab_{b} K_{3 a}- K_{\D_b 3 a}-K_{3 \D_b a}+\nab_{3} K_{ab} -K_{\D_3 a b } - K_{a \D_3 b}\\
&=&- \frac 1 2 \chib_{ab} K_{4 3} - \chib_{ac} K_{b c}- \chib_{bc} K_{c a}-\frac 1 2 \chib_{ba} K_{3 4}+\nab_{3} K_{ab} \\
&=&- ( \chib_{ab}+\chib_{ba}) (a^2\cos^2\th)  - r^2 \chib_{ac} \de_{b c}- r^2\chib_{bc} \de_{c a}+\nab_{3} (r^2 \de_{ab} )\\
&=&- ( \chib_{ab}+\chib_{ba})(r^2+ a^2\cos^2\th)  +\nab_{3} (r^2  )\de_{ab}\\
&=& (2r e_{3} (r )-|q|^2 \trchb  ) \de_{ab} -2 \chibh_{ab}|q|^2 =0,
\eeaa
where we used \eqref{eq:derivatives-r-th}, and symmetrically for $\Pi_{ab4}=0$. 
We compute
\beaa
3 \Pi_{a 34}&=&\D_{a} K_{34}+\D_{3} K_{4 a}+\D_{4} K_{a3} \\
&=&\nab_{3} K_{4a}-K_{\D_3 4 a}-K_{4 \D_3 a}+\nab_{4} K_{a 3}-K_{\D_4 a 3}- K_{a \D_4 3}+\nab_{a} K_{34} -K_{\D_a 3 4}- K_{3 \D_a 4}\\
&=&-2\eta_b K_{ ab }-\eta_a  K_{34}-\etab_a K_{34}-2\etab_b K_{a b}+\nab_{a} K_{34} \\
&=&\nab_{a} K_{34} -2( \eta_b+\etab_b) K_{ ab }-(\eta_a +\etab_a) K_{34}\\
&=&\nab_{a} 2(a^2\cos^2\th) -2( \eta_b+\etab_b) r^2 \de_{ab}-(\eta_a +\etab_a)  2(a^2\cos^2\th)\\
&=&2e_{a} (a^2\cos^2\th) -2( \eta_a+\etab_a) |q|^2 =0.
\eeaa
where we used \eqref{eq:derivatives-r-th}.
We compute
\beaa
3 \Pi_{a 33}&=&\D_{a} K_{33}+2\D_{3} K_{3 a}=\nab_{a} K_{33}- 2 K_{\D_a 3 3}+2\nab_{3} K_{3 a}-2 K_{\D_3 3 a}- 2 K_{3 \D_3 a}\\
&=&-4\xib_b K_{ b a}- 2\xib_a K_{3  4}=-4\xib_b r^2 \de_{ b a}- 4\xib_a (a^2 \cos^2\th)= - 4 |q|^2 \xib_a=0,
\eeaa
and similarly for $\Pi_{a44}=0$.
We compute
\beaa
3 \Pi_{34 3}&=&2\D_{3} K_{43}+\D_{4} K_{33}=2\nab_{3} K_{43}-2 K_{\D_3 43} -2 K_{4 \D_3 3}+\nab_{4} K_{33}-2K_{\D_4 3 3}\\
&=&2\nab_{3} K_{43}-4\omb  K_{ 43} +4\omb K_{43}=4e_{3}(a^2\cos^2\th) =0,
\eeaa
and similarly for $\Pi_{434}=0$. Finally we compute
\beaa
\Pi_{333}&=& \D_3 K_{33}=\nab_3 K_{33} - 2 K_{\D_3 3 3}=0
\eeaa
and similarly for $\Pi_{444}=0$. 
\end{proof}

Following Definition \ref{definition-Killing-operator}, we define the Carter differential operator in Kerr-Newman as the second order differential operator $\KK$ associated to $K$ in \eqref{definition-carter} for a scalar function $\psi$ as
\beaa
\mathcal{K}(\psi)= \D_{\mu}( K^{\mu\nu} \D_\nu(\psi)).
\eeaa
 We now prove that because of special considerations of the metric and curvature of Kerr-Newman, the commutator  between $\KK$ and $\square_{\g_{M, a, Q}}$ vanishes.

\begin{theorem}\label{theo:Carter-operator-commutes-KN} In Kerr-Newman spacetime, the Carter differential operator commutes with $\square_{\g_{M, a, Q}}$, i.e. for a scalar function $\psi$ we have
\bea
[\KK, \square_{\g_{M, a, Q}}]\psi=0.
\eea
In particular, the Carter differential operator $\KK$ is a symmetry operator for Kerr-Newman spacetime.
\end{theorem}
\begin{proof}  Since the Kerr-Newman metric is an electrovacuum spacetime, the commutator between $\KK$ and $\square_{\g_{M, a, Q}}$  satisfies \eqref{commutator-electrovacuum}, i.e.
\beaa
[\KK, \square_\g] \psi&=& \frac 2 3\big({{\R}_{\mu}}^\ep  \D^\mu  K_{\ep \nu}-   {{\R}_{\nu}}^\ep  \D^\mu K_{\mu \ep}-  \D^\a {\R^\ep}_{\nu}   K_{\a \ep}\big) \D^\nu \psi=:\frac 2 3 \left(I_1-I_2-I_3 \right).
\eeaa
From the Einstein-Maxwell equation \eqref{Einstein-Maxwell}, we compute the Ricci curvature of the Kerr-Newman metric, which is given by
\beaa
\R_{a3}&=& 2\dual \rhoF \dual \bbF_a - 2\rhoF\bbF_a=0, \qquad \R_{a4}=2\dual \rhoF \dual \bF_a + 2\rhoF\bF_a=0, \\
\R_{33}&=& 2 \bbF \c \bbF =0, \qquad \R_{44}= 2 \bF\cdot  \bF=0
\eeaa
and
\beaa
\R_{34}&=&2\left(\rhoF^2 +\dual \rhoF^2\right)=\frac{2Q^2}{|q|^4} , \qquad \R_{ab}=\left( \rhoF^2+\dual \rhoF^2 \right) \de_{ab} =\frac{Q^2}{|q|^4} \de_{ab} ,
\eeaa
where we used the values in Lemma \ref{lemma:Kerr-Newman}. Using the Ricci formulae \eqref{conn-coeff}, we deduce
\beaa
\D_a {\R}_{b c}&=& \nab_a {\R}_{b c}-\frac 1 2 \chi_{ab} \R_{3 c}-\frac 1 2 \chib_{ab}\R_{4 c}-\frac 1 2 \chi_{ac} \R_{b3} - \frac 1 2 \chib_{ac}\R_{b 4}\\
&=& \nab_a\left(\frac{Q^2}{|q|^4} \de_{bc} \right)=-2 (|q|^2)^{-3} \nab_a\left(| q|^2 \right) Q^2 \de_{bc}=-2 (\eta_a+\etab_a) \frac{Q^2}{|q|^4} \de_{bc},\\
\D_4\R_{3c}&=& \nab_4 \R_{3c}- 2\om \R_{3c}- 2\etab_b \R_{bc}-\etab_c\R_{34}=- 2\etab_b \frac{Q^2}{|q|^4} \de_{bc}-\etab_c \frac{2Q^2}{|q|^4}= - \frac{4Q^2}{|q|^4}\etab_c,\\
\D_a \R_{b3}&=& \nab_a \R_{b3}-\frac 1 2 \chi_{ab} \R_{33 }-\frac 1 2 \chib_{ab}\R_{34}-\chib_{ac}  \R_{bc}-\ze_a\R_{b3} =- \trchb  \frac{Q^2}{|q|^4} \de_{ab},\\
\D_3 \R_{43}&=& \nab_3 \R_{43}-2\eta_b \R_{b 3}= \nab_3 \left(\frac{2Q^2}{|q|^4}\right)= - \trchb \frac{4Q^2}{|q|^4},\\
\D_4 {\R}_{33}&=& \nab_4 {\R}_{33}-4\om \R_{33}-4\etab_b \R_{b 3}=0.
\eeaa
Similarly, as computed in Proposition \ref{prop:Carter-tensor-Killing}, we recall
\beaa
\D_a K_{b c}=\D_4 K_{33}=\D_3 K_{43}=0, \qquad \D_a K_{b3}=-\frac 1 2 \trchb |q|^2 \de_{ab}, \qquad \D_4K_{3c}= -2\etab_c |q|^2.
\eeaa
Also, from \eqref{eq:divergence-K} and using \eqref{eq:derivatives-r-th}, we have
\begin{equation}\label{eq:divergence}
\begin{split}
\D^\mu K_{\mu c}&= -\frac 1 2 \D_c (\tr K)= -\frac 1 2 e_c (2(r^2 -a^2\cos^2\th))= (\eta_c + \etab_c) |q|^2,\\
\D^\mu K_{\mu 3}&= -\frac 1 2 \D_3 (\tr K)= -\frac 1 2 e_3 (2(r^2 -a^2\cos^2\th))= - \trchb |q|^2, \\
\D^\mu K_{\mu4}&= -\frac 1 2 \D_4 (\tr K)= -\frac 1 2 e_4 (2(r^2 -a^2\cos^2\th))= - \trch |q|^2.
\end{split}
\end{equation}
We compute $I_1={{\R}_{\mu}}^\ep  \D^\mu  K_{\ep \nu}\D^\nu \phi$.
\beaa
I_1&=&{{\R}_{\mu}}^\ep  \D^\mu  K_{\ep \nu}  \D^\nu \psi= {{\R}_{\mu}}^a  \D^\mu  K_{a \nu}  \D^\nu \psi+{{\R}_{\mu}}^3  \D^\mu  K_{3 \nu}  \D^\nu \psi+{{\R}_{\mu}}^4  \D^\mu  K_{4 \nu}  \D^\nu \psi\\
&=& {{\R}_{b}}^a  \D^b  K_{a \nu}  \D^\nu \psi-\frac 1 2 {\R}_{34}  \D^3  K_{3 \nu}  \D^\nu \psi-\frac 1 2 {\R}_{43}  \D^4  K_{4 \nu}  \D^\nu \psi\\
&=& \R^{ab}  \D_b  K_{a c}  \D^c \psi+\R^{ab}  \D_b  K_{a 3}  \D^3 \psi+\R^{ab}  \D_b  K_{a 4}  \D^4 \psi\\
&&+\frac 1 4 {\R}_{34}  \D_4  K_{3 a}  \D^a \psi+\frac 1 4 {\R}_{34}  \D_4  K_{3 3}  \D^3 \psi+\frac 1 4 {\R}_{34}  \D_4  K_{3 4}  \D^4 \psi\\
&&+\frac 1 4 {\R}_{43}  \D_3  K_{4 a}  \D^a \psi+\frac 1 4 {\R}_{43}  \D_3  K_{4 3}  \D^3 \psi+\frac 1 4 {\R}_{43}  \D_3  K_{4 4}  \D^4 \psi
\eeaa
which gives, using the above values, 
\beaa
I_1&=&-\R^{ab} \frac 1 2 \trchb |q|^2 \de_{ab}  \D^3 \psi-\R^{ab} \frac 1 2 \trch |q|^2 \de_{ab}  \D^4 \psi-\frac 1 2 {\R}_{34} \etab_a |q|^2  \D^a \psi-\frac 1 2 {\R}_{43}  \eta_a |q|^2  \D^a \psi\\
&=&\frac 1 2 \frac{Q^2}{|q|^2}  ( \trchb \D_4 \psi+\trch\D_3 \psi)- \frac{Q^2}{|q|^2} (\eta_a+\etab_a  ) \D^a \psi.
\eeaa
We compute $I_2= {{\R}_{\nu}}^\ep  \D^\mu K_{\mu \ep} \D^\nu \psi$.
\beaa
 I_2&=& {{\R}_{\nu}}^\ep  \D^\mu K_{\mu \ep} \D^\nu \psi=   {{\R}_{\nu}}^c  \D^\mu K_{\mu c} \D^\nu \psi+  {{\R}_{\nu}}^3  \D^\mu K_{\mu 3} \D^\nu \psi+  {{\R}_{\nu}}^4  \D^\mu K_{\mu 4} \D^\nu \psi\\
  &=&   \R^{ac}  (\eta_c + \etab_c) |q|^2 \D_a \psi+\frac 1 4   {\R}_{34} ( - \trchb |q|^2) \D_4 \psi+\frac 1 4  {\R}_{43} ( - \trch |q|^2) \D_3 \psi\\
    &=&-\frac 1 2  \frac{Q^2}{|q|^2}    (  \trchb \D_4 \psi+ \trch \D_3 \psi)+  \frac{Q^2}{|q|^2}  (\eta^a + \etab^a)  \D_a \psi=-I_1.
\eeaa
Finally, we compute $I_3$.
\beaa
I_3&=& \D^\a {\R^\ep}_{\nu}   K_{\a \ep} \D^\nu \psi =  \D^\a {\R^c}_{\nu}   K_{\a c} \D^\nu \psi+ \D^\a {\R^3}_{\nu}   K_{\a 3} \D^\nu \psi+ \D^\a {\R^4}_{\nu}   K_{\a 4} \D^\nu \psi\\
 &=&  \D_a \R_{c\nu}   K^{a c} \D^\nu \psi+\frac 1 4  \D_3 \R_{4\nu}   K_{4 3} \D^\nu \psi+ \frac 1 4 \D_4 \R_{3\nu}   K_{3 4} \D^\nu \psi\\
 &=&  \D_a \R_{cb}   K^{a c} \D^b \psi+ \D_a \R_{c3}   K^{a c} \D^3 \psi+ \D_a \R_{c4}   K^{a c} \D^4 \psi\\
 &&+\frac 1 4  \D_3 \R_{4a}   K_{4 3} \D^a \psi+\frac 1 4  \D_3 \R_{43}   K_{4 3} \D^3 \psi+\frac 1 4  \D_3 \R_{44}   K_{4 3} \D^4 \psi\\
 &&+ \frac 1 4 \D_4 \R_{3a}   K_{3 4} \D^a \psi+ \frac 1 4 \D_4 \R_{33}   K_{3 4} \D^3 \psi+ \frac 1 4 \D_4 \R_{34}   K_{3 4} \D^4 \psi
\eeaa
which gives, using the above values, 
\beaa
 I_3 &=&-2 \frac{Q^2}{|q|^4}  r^2  (\eta_a+\etab_a) \D^a \psi- \trchb  \frac{Q^2}{|q|^4} 2r^2 \D^3 \psi - \trch  \frac{Q^2}{|q|^4} 2r^2 \D^4 \psi\\
 &&-\frac 1 4  \frac{4Q^2}{|q|^4}\eta_a  2(a^2\cos^2\th)  \D^a \psi-\frac 1 4  \trchb \frac{4Q^2}{|q|^4}   2(a^2\cos^2\th)  \D^3 \psi\\
 &&- \frac 1 4 \frac{4Q^2}{|q|^4}\etab_a   2(a^2\cos^2\th)  \D^a \psi- \frac 1 4 \trch \frac{4Q^2}{|q|^4}    2(a^2\cos^2\th)  \D^4 \psi\\
 &=&\frac{Q^2}{|q|^2} ( \trchb \D_4 \psi +\trch \D_3 \psi)-2 \frac{Q^2}{|q|^2}   (\eta^a+\etab^a) \D_a \psi=2I_1.
\eeaa
We therefore obtain
\beaa
[\KK, \square_\g] \psi&=&\frac 2 3 \left(I_1-I_2-I_3 \right)=0,
\eeaa
as stated.
\end{proof}

\subsection{The modified Laplacian in Kerr-Newman}\label{section:modified-laplacian}

Even though the Carter differential operator $\KK$ is a symmetry operator for the Kerr-Newman metric, it is convenient to extract from it an elliptic operator which we identify as a \textit{modified Laplacian} in Kerr-Newman, see also \cite{GKS}. Such modified Laplacian is then proved to be a \textit{conformal symmetry operator}, as it is a symmetry operator for the conformal rescaling of the Kerr-Newman metric $|q|^2 \g_{M, a, Q}$.

\begin{proposition}\label{prop:KK-OO} The Carter differential operator $\KK$ in Kerr-Newman is given by
\bea\label{operator-KK}
\KK &=& -(a^2 \cos^2\th) \ \square_{\g_{M, a, Q}} + \OO,
\eea
where $\OO$ denotes a second order differential operator given by
\bea\label{definition-OO-lap}
\OO(\psi) &=& |q|^2 \left(\lap \psi + (\eta+\etab) \c \nab \psi   \right),
\eea
where $\lap\psi=\de^{ab} \nab_a \nab_b \psi$, for $a, b=1,2$. 
We call $\OO$ the \textup{modified Laplacian} of the Kerr-Newman metric.

Moreover, the modified Laplacian $\OO$ is a \textup{conformal symmetry operator}, i.e. for a scalar function $\psi$ we have
\bea\label{eq:OO-commutes-}
[\OO, |q|^2\square_{\g_{M, a, Q}}]\psi=0.
\eea
\end{proposition} 
\begin{proof} 
We compute 
\beaa
\mathcal{K}(\psi)= \D_{\mu}( K^{\mu\nu} \D_\nu\psi)= K^{\mu\nu} \D_{\mu}\D_\nu\psi+\D_{\mu} K^{\mu\nu} \D_\nu\psi
\eeaa
Using \eqref{eq:divergence} and the definition of $K$ \eqref{definition-carter}, we obtain
\beaa
\KK(\psi)&=&(-(a^2\cos^2\th)  \g^{\mu\nu} +|q|^2(e_1 \otimes e_1+e_2 \otimes e_2)^{\mu\nu}) \D_{\mu}\D_\nu\psi\\
&&+( \eta^a+\etab^a) |q|^2  \nab_a\psi+\frac 1 2|q|^2 \trch   \nab_3\psi+\frac 1 2 |q|^2 \trchb   \nab_4\psi\\
&=&-(a^2\cos^2\th) \square_{\g_{M, a, Q}} \psi + |q|^2 \de^{ab} \D_{a}\D_b\psi\\
&&+( \eta^a+\etab^a) |q|^2  \nab_a\psi+\frac 1 2|q|^2 \trch   \nab_3\psi+\frac 1 2 |q|^2 \trchb   \nab_4\psi\\
&=&-(a^2\cos^2\th) \square_{\g_{M, a, Q}} \psi + |q|^2\de^{ab}( \nab_b \nab_a\psi -\frac 1 2 \chi_{ba} \nab_{3} \psi -\frac 1 2 \chib_{ba} \nab_{4}\psi )\\
&&+( \eta^a+\etab^a) |q|^2  \nab_a\psi+\frac 1 2|q|^2 \trch   \nab_3\psi+\frac 1 2 |q|^2 \trchb   \nab_4\psi,
\eeaa
which gives
\beaa
\KK(\psi)&=&-(a^2\cos^2\th) \square_{\g_{M, a, Q}} \psi + |q|^2 \left(\de^{ab} \nab_b \nab_a\psi +( \eta^a+\etab^a)  \nab_a\psi\right),
\eeaa
where the last two terms define the operator $\OO(\psi)$. This proves the first part of the Proposition.

Using \eqref{operator-KK} to write $\OO=\KK+(a^2\cos^2\th) \square_{\g}$, we deduce using Theorem \ref{theo:Carter-operator-commutes-KN}, 
 \beaa
 \, [\OO,  \square_{\g}]\psi&=&  [ \KK,  \square_{\g}]\psi+  [ (a^2 \cos^2\th) \ \square_{\g} ,  \square_{\g}]\psi=  [ (a^2 \cos^2\th) \ \square_\g ,  \square_\g]\psi .
 \eeaa
 Recall, see for example \cite{GKS}, for a scalar function $f$ we have
 \bea\label{eq:wave-eq-34}
 \square_\g f&=& -\frac 1 2 (\nab_3 \nab_4 +\nab_4 \nab_3)f+\left( \omb-\frac 1 2 \trchb \right) \nab_4 f +\left(\om-\frac 1 2 \trch \right)\nab_3 f +\lap f + (\eta+\etab)\c \nab f.
 \eea
 Consequently, we have
  \bea\label{square-f-Q-1}
\square_\g( f h)&=& \square_\g(f) h+f \square_\g(h)- \nab_3 f \nab_4h- \nab_4f \nab_3 h +2\nab f \c \nab h.
\eea
We then obtain, 
 \begin{equation}\label{eq:comm-OO-square}
 \begin{split}
 \, [\OO,  \square_\g]\psi &= (a^2 \cos^2\th) \ \square_\g (  \square_\g\psi) - \square_\g ((a^2\cos^2\th) \square_\g \psi)\\
 &= -\square_\g(a^2\cos^2\th )   \square_\g\psi-2\nab (a^2\cos^2\th ) \c \nab (  \square_\g\psi) 
 \end{split}
 \end{equation}
 where we used that $e_3(a^2\cos^2\th)=e_4(a^2\cos^2\th)=0$. Also,
using \eqref{eq:wave-eq-34} we compute
 \beaa
\square_\g (a^2\cos^2\th )&=&\lap (a^2\cos^2\th )+ (\eta+\etab) \c\nab(a^2\cos^2\th )=|q|^{-2} \OO(a^2\cos^2\th).
\eeaa
This gives
 \beaa
 \, [\OO,  \square_\g]\psi  &=& -|q|^{-2} \OO(a^2\cos^2\th)  \square_\g\psi-2\nab (a^2\cos^2\th ) \c \nab (  \square_\g\psi).
 \eeaa
 We can finally deduce from \eqref{eq:comm-OO-square}:
 \beaa
 \, [\OO, |q|^2\square_\g]\psi&=& \OO(|q|^2 \square_\g \psi) - |q|^2 \square_\g( \OO(\psi))\\
 &=& \OO(|q|^2) \square_\g \psi + |q|^2[\OO,  \square_\g]\psi +2|q|^2 \nab (|q|^2) \c \nab (\square_\g \psi)\\
  &=& \OO(|q|^2) \square_\g \psi - \OO(a^2\cos^2\th)  \square_\g\psi-2|q|^2\nab (a^2\cos^2\th ) \c \nab (  \square_\g\psi)  +2|q|^2 \nab (|q|^2) \c \nab (\square_\g \psi)\\
  &=& 0
 \eeaa
 where we used that $\nab (|q|^2)=\nab(r^2+a^2\cos^2\th)=\nab(a^2\cos^2\th)$ and $\OO(|q|^2)=\OO(a^2\cos^2\th)$.
 This proves the Proposition. 
\end{proof}

We finally express the modified Laplacian $\OO$ explicitly in Boyer-Lindquist coordinates in Kerr-Newman.

\begin{lemma} The modified Laplacian $\OO$  in Boyer-Lindquist coordinates reads
\bea\label{operator-O-explicit}
\OO &=& \frac{1}{\sin\th} \pr_\th(\sin\th\pr_\th)+\frac{1}{\sin^2\th} \pr^2_\vphi +2a \partial_t\partial_\vphi+a^2\sin^2\th\pr^2_t\\
&=& \lap_{\SSS^2} +2a \partial_t\partial_\vphi+a^2\sin^2\th\pr^2_t
\eea
where $ \lap_{\SSS^2} =\frac{1}{\sin\th} \pr_\th(\sin\th\pr_\th)+\frac{1}{\sin^2\th} \pr^2_\vphi $ is the (unit) spherical Laplacian on $\SSS^2$. 
\end{lemma}
\begin{proof}
We recall that the Laplacian $\lap$ of a scalar function is given by, see \cite{GKS}
\beaa
\lap\psi&=& \de^{ab} \nab_a\nab_b\psi= \nab_1\nab_1\psi+\nab_2\nab_2\psi= e_1e_1(\psi )-\nab_{\nab_1 e_1} \psi+e_2e_2(\psi)-\nab_{\nab_2 e_2} \psi\\
&=& e_1e_1(\psi )+e_2e_2(\psi )+\La e_1(\psi).
\eeaa
where we used \eqref{eq:expressions-nab11}.
Using the values \eqref{eq:Out.PGdirections-Kerr-12} of $e_1$ and $e_2$, we compute
\beaa
e_1e_1(\psi)&=&\frac{1}{|q|}\pr_\th(\frac{1}{|q|}\pr_\th \psi)=\frac{1}{|q|^2}\pr^2_\th \psi+\frac{1}{|q|}\pr_\th(\frac{1}{|q|})\pr_\th \psi=\frac{1}{|q|^2}\pr^2_\th \psi+\frac{a^2\cos\th\sin\th}{|q|^4}\pr_\th \psi\\
e_2e_2(\psi)&=&(\frac{a\sin\th}{|q|}\pr_t+\frac{1}{|q|\sin\th}\pr_\vphi)(\frac{a\sin\th}{|q|}\pr_t\psi+\frac{1}{|q|\sin\th}\pr_\vphi \psi)\\
&=&\frac{a^2\sin^2\th}{|q|^2}\pr^2_t\psi+\frac{2a}{|q|^2}\pr_t \pr_\vphi\psi +\frac{1}{|q|^2\sin^2\th}\pr^2_\vphi \psi
\eeaa
which gives
\beaa
\lap\psi&=&\frac{1}{|q|^2}\pr^2_\th \psi+\frac{a^2\cos\th\sin\th}{|q|^4}\pr_\th \psi+\frac{a^2\sin^2\th}{|q|^2}\pr^2_t\psi+\frac{2a}{|q|^2}\pr_t \pr_\vphi\psi +\frac{1}{|q|^2\sin^2\th}\pr^2_\vphi \psi+\frac{r^2+a^2}{|q|^4}\cot\th \pr_\th\psi\\
&=&\frac{1}{|q|^2}\pr^2_\th \psi+\frac{r^2+a^2+a^2\sin^2\th}{|q|^4} \cot\th\pr_\th \psi+\frac{a^2\sin^2\th}{|q|^2}\pr^2_t\psi+\frac{2a}{|q|^2}\pr_t \pr_\vphi\psi +\frac{1}{|q|^2\sin^2\th}\pr^2_\vphi \psi.
\eeaa
We also compute
\beaa
(\eta+\etab) \c \nab \psi&=&(\eta_1+\etab_1)e_1(\psi)+(\eta_2+\etab_2)e_2(\psi)=-2\frac{a^2\cos\th\sin\th}{|q|^3}\frac{1}{|q|}\pr_\th(\psi)=-2\frac{a^2\sin^2\th}{|q|^4}\cot\th\pr_\th(\psi)
\eeaa
which gives
\beaa
\OO(\psi)&=& |q|^2( \lap \psi+ (\eta+\etab) \c \nab \psi)\\
&=&\pr^2_\th \psi+\frac{r^2+a^2+a^2\sin^2\th}{|q|^2} \cot\th\pr_\th \psi+a^2\sin^2\th\pr^2_t\psi+2a\pr_t \pr_\vphi\psi +\frac{1}{\sin^2\th}\pr^2_\vphi \psi-2\frac{a^2\sin^2\th}{|q|^4}\cot\th\pr_\th\psi\\
&=&\pr^2_\th \psi+ \cot\th\pr_\th \psi+\frac{1}{\sin^2\th}\pr^2_\vphi \psi+a^2\sin^2\th\pr^2_t\psi+2a\pr_t \pr_\vphi\psi\\
&=& \frac{1}{\sin\th} \pr_\th(\sin\th\pr_\th)+\frac{1}{\sin^2\th}\pr^2_\vphi \psi+a^2\sin^2\th\pr^2_t\psi+2a\pr_t \pr_\vphi\psi,
\eeaa
as stated.

\end{proof}

\section{The physical-space analysis of the wave equation}\label{subsection-preliminaries-wave}\label{section:slowly}

In this section we prove energy and local decay estimates for solutions to the scalar wave equation
\bea\label{wave-eq-gen}\label{wave-equation-general}
\square_{\g_{M, a, Q}} \psi=0
\eea
in Kerr-Newman spacetime for $|a| \ll M$ \textbf{entirely in physical space}, i.e. without recurring into frequency space or decomposition in modes. In order to do that, we will make use the approach first developed in \cite{And-Mor} in Kerr to commute the wave equation with the Carter tensor and the modified Laplacian $\OO$. 
From the fundamental commutation property obtained in Theorem \ref{theo:Carter-operator-commutes-KN}, we can extend the procedure of \cite{And-Mor} to the case of Kerr-Newman, where a resolution in physical space is crucial to tackle the problem of stability of the solution to electromagnetic-gravitational perturbations for the Einstein-Maxwell equation, as we will see in Section \ref{Einstein-Maxwell-equations}.

As in \cite{And-Mor}, define
\bea
|\psi|^2_{\SS}:= |\psi|^2+ |\pr_t \psi|^2+ |\pr_\vphi \psi|^2+ \sum_{\aund=1}^4 |\SS_\aund \psi|^2,
\eea
where $\SS_\aund$, for $\aund=1,2,3,4$, denote the set of the second order operators  given by, see Definition \ref{definition-SS},
\bea
\SS_1=\partial_t^2, \qquad \qquad \SS_2=\partial_{t} \partial_{\vphi}, \qquad \qquad \SS_3=\partial_\vphi^2, \qquad \qquad \SS_4=\OO.
\eea

We obtain the following.

\begin{theorem}\label{theorem:general} Let $\psi$ be a sufficiently regular solution to the wave equation in the slowly rotating Kerr-Newman spacetime $\g_{M, a, Q}$ with $|a| \ll M$, with initial data on $\Sigma_0$ which decays sufficiently fast.

Then the following energy-Morawetz estimates, for $\tau \geq 0$, can be obtained through \textbf{a physical-space analysis}:
\bea
E_{\tau, \SS}[\psi] +\mbox{Mor}_{(0, \tau), \SS}[\psi] \les E_{0, \SS}[\psi]
\eea
where 
\bea
E_{\tau, \SS}[\psi]&:=&\int_{\Sigma_\tau} |\partial_t \psi|_{\SS}^2+ |\partial_r \psi|_{\SS}^2+  |\nab \psi|_{\SS}^2 \\
\mbox{Mor}_{(\tau_1, \tau_2)}[\psi]&:=& \int_{\MM(\tau_1, \tau_2)} \frac{M^2}{r^3} |\pr_r\psi|_{\SS}^2 +\frac{M}{r^4}  |\psi|_{\SS}^2+ \widetilde{\mathbbm{1}}_{\{r \neq r^{RN}_{trap}\}} \big( r^{-1} |\nab \psi|_{\SS}^2+\frac{M}{r^2} |\pr_t \psi|_{\SS}^2\big) 
\eea
where $r^{RN}_{trap}=\frac{3M+\sqrt{9M^2-8Q^2}}{2}$ is the photon sphere of Reissner-Nordstr\"om $\g_{M, Q}$ and $\widetilde{\mathbbm{1}}_{\{r \neq r^{RN}_{trap}\}}$ is a function that is identically 1 for $\{ |r-r^{RN}_{trap}| > \delta\}$ for some $\delta >0$ and zero otherwise,  and $|\nab \psi|^2=|\nab_1 \psi|^2+|\nab_2\psi|^2$ with respect to the orthonormal basis in \eqref{eq:Out.PGdirections-Kerr-12}. 
\end{theorem}

Observe that the above estimates is not optimal in terms of decay in $r$, and those weights can be improved  by applying the $r^p$ hierarchy of estimates introduced by Dafermos-Rodnianski in \cite{DR09a} in a standard fashion.

\subsection{Preliminaries}

From now on, we denote $\g=\g_{M, a, Q}$ for $a^2+Q^2 < M^2$.
From the form of the Kerr-Newman metric in Boyer-Lindquist coordinates $(t, r,  \th, \vphi)$ as given by \eqref{metric-KN}, one can deduce, see \cite{And-Mor}, that its conformal inverse $|q|^2 \g^{-1}$ can be written as 
\bea\label{inverse-metric-Kerr}
|q|^2 \g^{\a\b}&=& \Delta \partial_r^\a \partial_r^\b+\frac{1}{\Delta} \RR^{\a\b}
\eea
where
\bea
\RR^{\a\b}&=&  -(r^2+a^2)^2 \partial_t^\a \partial_t^\b-2a(r^2+a^2)\partial_t^{(\a} \partial_\vphi^{\b)}-a^2  \partial_\vphi^\a \partial_\vphi^\b+ \Delta O^{\a\b}, \label{definition-RR-tensor}\\
 O^{\a\b}&=& \partial_\th^\a  \partial_\th^\b  +\frac{1}{\sin^2\th} \partial_{\vphi}^\a \partial_{\vphi}^\b+2a\partial_t^{(\a} \partial_\vphi^{\b)}+a^2 \sin^2\th \partial_t^\a \partial_t^\b.
\eea
Observe that $O^{\a\b}$ is the same tensor as defined in \eqref{eq:first-definition-O}, but written here in Boyer-Lindquist coordinates using the orthonormal frame in \eqref{eq:Out.PGdirections-Kerr-12}.

The scalar wave equation on a Lorentzian manifold can be written in coordinates as
\beaa
\square_\g \psi=\frac{1}{\sqrt{-\det \g}}\partial_\a ((\sqrt{-\det\g}) \g^{\a\b} \partial_\b \psi)=0,
\eeaa
and one can deduce from \eqref{inverse-metric-Kerr} that the wave operator for the Kerr-Newman metric in Boyer-Lindquist coordinates is given by
\begin{equation}\label{square-GKS}
\begin{split}
|q|^2\square_\g&=\pr_r(\Delta \pr_r) +\frac{1}{\Delta} \Big(-(r^2+a^2)^2 \pr^2_t-2a(r^2+a^2)\pr_t \pr_\vphi-a^2\pr_\vphi^2\Big)\\
&+\frac{1}{\sin\th} \pr_\th(\sin\th\pr_\th)+\frac{1}{\sin^2\th} \pr^2_\vphi +2a \partial_t\partial_\vphi+a^2\sin^2\th\pr^2_t\\
&=\pr_r(\Delta \pr_r) +\frac{1}{\Delta} \Big(-(r^2+a^2)^2 \pr^2_t-2a(r^2+a^2)\pr_t \pr_\vphi-a^2\pr_\vphi^2\Big)+\OO,
\end{split}
\end{equation}
where $\OO$ is the modified Laplacian  defined in \eqref{operator-O-explicit}.

We now briefly formulate the initial value problem for the wave equation. 
We prescribe initial data on the (axisymmetric) hypersurface $\Sigma_0=\{ t^*=0\}$ where $t^*$ is the Kerr star coordinate.  We are interested in the behavior of the solution in the future Cauchy development of $\Sigma_0$ which is given by $\{ t^{*} \geq 0\}$. Denote $\phi_\tau$ the 1-parameter family of diffeomorphisms generated by the vector field $T$, and define the spacelike hypersurfaces $\Sigma_\tau=\phi_{\tau}(\Sigma_0)=\{ t=t^{*}\}$. Each leaf of this foliation terminates at the horizon and at spatial infinity $i^0$ (for more details see \cite{DR11}, \cite{DRSR}). For $\tau_2 >\tau_1$, the leaf $\Sigma_{\tau_2}$ lies in the future of $\Sigma_{\tau_1}$, and we denote the region bounded by $\Sigma_{\tau_1}$, $\Sigma_{\tau_2}$ and $\mathcal{H}^+$ by $\MM(\tau_1, \tau_2)=\cup_{\tau_1 \leq \tau \leq \tau_2} \Sigma_\tau$. We also denote $\mathcal{H}^+(\tau_1, \tau_2)=\mathcal{H}^+\cap \MM(\tau_1, \tau_2)$.

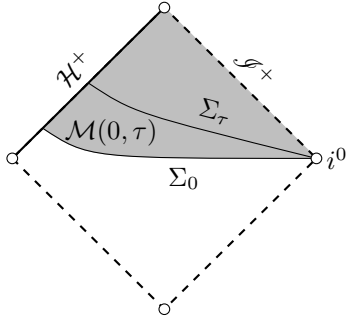
\begin{center}
\begin{tikzpicture}
 \filldraw[lightgray] (4,0) .. controls (1,0) .. (0.4,0.4) -- (2,2) -- (4,0);
 \draw[dashed,black, thick] (0,0) -- (2,-2) -- (4,0) node[anchor=west]{$i^0$} -- node[sloped, above]{$\mathscr{I}^+$} (2,2) -- node[sloped, above]{$\mathcal{H}^+$} (0,0);
 \draw[black] (4,0) .. controls (1,0) .. (0.4,0.4) node[near start, below]{$\Sigma_0$} node[midway, above]{$\MM(0, \tau)$};
  \draw[black] (4,0) .. controls (1.6,0.6) .. (1,1) node[near start, sloped, above]{$\Sigma_{\tau}$};
  \draw[black, thick] (0,0) --(2,2);
\filldraw[fill=white, draw=black] (0,0) circle (2pt);
\filldraw[fill=white, draw=black] (2,-2) circle (2pt);
\filldraw[fill=white, draw=black] (4,0) circle (2pt);
\filldraw[fill=white, draw=black] (2,2) circle (2pt);
\end{tikzpicture}
\end{center}

The wave equation is well posed in $\MM(0, \tau)$ with initial data $(\psi_0, \psi_1)$ defined on $\Sigma_0$ in $H^j_{loc}(\Sigma_0) \times H^{j-1}_{loc}(\Sigma_0)$, $j \geq 1$, see for example \cite{DR10}. Furthermore, the solutions depends smoothly on the parameters $a$ and $Q$, see \cite{DRSR}\cite{Civin}. By symmetry we can always assume the positivity of $a$.

\subsection{The vectorfield method}

We recall the main definitions in applying the vectorfield method to derive energy estimates for the wave equation. The vectorfield method is based on applying the divergence theorem in a causal domain, like $\MM(\tau_1, \tau_2)$, to certain energy currents, which are constructed from the energy momentum tensor.

\begin{itemize}
\item  The energy-momentum tensor associated to the wave equation \eqref{wave-equation-general} is given by
\bea\label{definition-energy-momentum-tensor}
\QQ[\psi]_{\mu\nu}&=& \pr_\mu\psi \pr_\nu \psi -\frac 1 2 \g_{\mu\nu} \pr_\lambda \psi \pr^\lambda \psi.
\eea
The wave equation \eqref{wave-equation-general} is satisfied if and only if the divergence of the energy-momentum tensor $\QQ[\psi]$ vanishes.

\item Let $X$ be a vectorfield and $w$ be a function. The current associated to $(X, w)$ is defined as 
 \bea\label{definition-of-P}
 \PP_\mu^{(X, w)}[\psi]&=&\QQ[\psi]_{\mu\nu} X^\nu +\frac 1 2  w  \psi \pr_\mu \psi   -\frac 1 4(\pr_\mu w )\psi^2.
  \eea
  
  \item The energy associated to  $(X, w)$ on the hypersurface $\Sigma_\tau$ is
  \beaa
E^{(X, w)}[\psi](\tau)&=& \int_{\Sigma_\tau} \PP^{(X, w)}_\mu[\psi] n_{\Sigma_\tau}^\mu 
\eeaa
where $n_{\Sigma_\tau}$ denotes the future directed timelike unit normal to $\Sigma_\tau$.
\end{itemize}

A standard computation, see for example \cite{KS}, then implies for the divergence of $\PP$:
  \bea
  \label{le:divergPP-gen}
  \D^\mu \PP_\mu^{(X, w)}[\psi]= \frac 1 2 \QQ[\psi]  \c\piX-\frac 1 4 \square_\g w |\psi|^2+\frac 12  w (\pr_\lambda \psi \pr^\lambda \psi),
 \eea
 where $\piX_{\mu\nu}=\D_{(\mu} X_{\nu)}$ is the deformation tensor of the vectorfield $X$. Recall that if $X$ is a Killing vectorfield, then $\piX=0$.

For convenience    we  introduce the notation,
 \bea
 \label{eq:modified-div}
 \EE^{(X, w)}[\psi]&:=&   \D^\mu  \PP_\mu^{(X, w)}[\psi].
  \eea

By applying the divergence theorem to $\PP_\mu^{(X, w)}$ within a region such as $\MM(\tau_1, \tau_2)$ for carefully chosen $(X, w)$ one obtains the associated energy identity:
\bea
E^{(X, w)}[\psi](\tau_2) + \int_{\mathcal{H}^+(\tau_1, \tau_2)} \PP^{(X, w)}_\mu[\psi] n_{\mathcal{H}^+}^\mu+\int_{\MM(\tau_1, \tau_2)}\EE^{(X, w)}[\psi] = E^{(X, w)}[\psi](\tau_1),
\eea
where the induced volume forms are to be understood. By convention, along the event horizon $\mathcal{H}^+$ we choose $n_{\mathcal{H}^+}=\That_\HH$.

 For two positive quantities $F$ and $G$, in what follows we write $F \les G$ to signify that there exists a universal constant $C$, depending only on $M, a, Q$, such that $F \leq C G$.

\subsection{Symmetry operators and the generalized vectorfield method}\label{section:generalized-vectorfield-method}

Since the Kerr-Newman metric possesses only two Killing vectorfields, $T=\partial_t$ and $Z=\partial_{\vphi}$, which commute with the D'Alembertian operator associated to the metric, the control on those first derivatives is not sufficient to control all the first derivatives of a solution to the wave equation.

In order to control all the derivatives, one needs to make use of the Carter tensor, through the form of the modified Laplacian $\OO$ as defined in Section \ref{section:modified-laplacian}. Since such operator is second order, we will make use of second order differential operators as obtained from the two Killing vectorfields and the modified Laplacian.

Following \cite{And-Mor} and making crucial use of our extension of the commutation in the case of the non-Ricci flat Kerr-Newman spacetime in Theorem \ref{theo:Carter-operator-commutes-KN} and \eqref{eq:OO-commutes-}, we define the following second order symmetry operators. 

\begin{definition}\label{definition-SS} The set of the second order operators $\SS_\aund$, for $\aund=1,2,3,4$, given by
\bea
\SS_1=\partial_t^2, \qquad \qquad \SS_2=\partial_{t} \partial_{\vphi}, \qquad \qquad \SS_3=\partial_\vphi^2, \qquad \qquad \SS_4=\OO
\eea
are denoted \textup{conformal symmetry operator}\footnote{Observe that $\SS_\aund$ for $\aund=1,2,3$ are symmetry operators as they commute with $\square_\g$, while $\SS_4=\OO$ is only a conformal symmetry operator.}, as for a scalar function $\psi$ we have
\bea\label{eq-conf-sym}
[\SS_\aund, |q|^2\square_{\g_{M, a, Q}}]\psi=0, \qquad \aund=1,2,3,4.
\eea

\end{definition}
Observe that the conformal symmetry operators commute with each other, i.e.
\bea\label{eq:commutator-SS-aund-bund}
[\SS_\aund, \SS_\bund]=0, \qquad \aund, \bund=1,2,3,4.
\eea

The modified Laplacian $\OO$ in \eqref{operator-O-explicit} differs from the spherical Laplacian by terms depending on $\partial_t^2$ and $\pr_t \pr_\vphi$, so the second order operators $\SS_\aund$ for $\aund=1,2,3,4$ together provide the spherical Laplacian, which has the elliptic properties necessary to control the angular derivatives. More precisely, the following estimates allow us to obtain pointwise boundedness and decay bounds for $\psi$:
\bea
|\psi|^2 \leq C \int_{\mathbb{S}^2} |\psi|^2+| \lap_{\SSS^2}\psi|^2 \leq C \int_{\mathbb{S}^2} |\psi|^2+\sum_{\aund=1}^4| \SS_\aund \psi|^2,
\eea
which follows from the spherical Sobolev inequality and \eqref{operator-O-explicit}.

In addition to the conformal symmetry operators, we also define their tensorial versions.

\begin{definition}\label{definition-S} We define the following symmetric tensors
\bea
S_1^{\a\b}=T^\a T^\b, \qquad S_2^{\a\b} =T^{(\a}Z^{\b)}, \qquad S_3^{\a\b}=Z^\a Z^\b, \qquad S_4^{\a\b}=O^{\a\b}
\eea
\end{definition}
With the above definition, from \eqref{definition-RR-tensor}, one can write as in \cite{And-Mor}
\bea\label{eq:RR-ab-RR-aund}
\RR^{\a\b}=  -(r^2+a^2)^2 S_1^{\a\b}-2a(r^2+a^2)S_2^{\a\b}-a^2  S_3^{\a\b}+ \Delta O^{\a\b}=:  \RR^\aund S_\aund^{\a\b}, 
\eea
with 
\bea
\RR^1= -(r^2+a^2)^2, \qquad \RR^2=-2a(r^2+a^2), \qquad \RR^3=-a^2, \qquad \RR^4=\Delta.
\eea

Observe that the symmetric tensors $S_\aund$ are easily related to the conformal symmetry operators $\SS_\aund$. 

\begin{lemma}\label{lemma:D-S-SS} The symmetric tensors defined in Definition \ref{definition-S} and the conformal symmetry operators defined in Definition \ref{definition-SS} are related by the following:
\bea
\SS_\aund=|q|^2\D_\a(|q|^{-2}S_\aund^{\a\b}  \D_\b) , \qquad \aund=1,2,3,4.
\eea
\end{lemma}
\begin{proof} Observe that for $\aund=1,2,3$, we have $|q|^2\D_\a(|q|^{-2}S_\aund^{\a\b}  \D_\b)=\D_\a(S_\aund^{\a\b}  \D_\b)$, and for example for $\aund=1$ we have
\beaa
\D_\a(S_1^{\a\b}  \D_\b) =\D_\a(T^\a T^{\b}  \D_\b)=(\D_\a T^\a) T^\b \D_\b + (T^\a \D_\a) (T^\b \D_\b)=\partial^\a_t \partial^\a_t
\eeaa
since $\D_\a T^\a=\tr ^{(T)}\pi=0$. Similarly for $\aund= 2,3$. 
Using \eqref{eq:expression-K-O} and \eqref{operator-KK}, we obtain
\beaa
|q|^2\D_\a(|q|^{-2}O^{\a\b}  \D_\b)&=& \D_\a( O^{\a\b} \D_\b)+|q|^2O^{\a\b}\D_\a(|q|^{-2})\D_\b\\
&=&\D_\a( K^{\a\b} \D_\b)+\D_\a( (a^2\cos^2\th) \g^{\a\b} \D_\b)+|q|^4(e_1^\a e_1^\b +e_2^\a e_2^\b)\D_\a(|q|^{-2})\D_\b\\
&=&\KK+(a^2 \cos^2\th) \ \square_{\g}+\D_\a (a^2\cos^2\th) \D^\a-\D_\a(|q|^{2})\D_\a\\
&=&\KK+(a^2 \cos^2\th) \ \square_{\g}=\OO
\eeaa
as stated.
\end{proof}

Define
\bea\label{eq:def-psia}
\psia:= \SS_\aund(\psi), \qquad \aund=1,2,3,4.
\eea
Then if $\psi$ is a solution to the wave equation, by \eqref{eq-conf-sym}  $\psia$ is also a solution,
\bea\label{eq:psia-solution}
\square_{\g_{M, a, Q}} \psia=0, \qquad \aund=1,2,3,4.
\eea

Following \cite{And-Mor}, we recall here a generalized vectorfield method which incorporates the commutation of the wave equation by the conformal symmetry operators.

\begin{itemize}
\item The generalized energy-momentum tensor for $\psi$ solutions to \eqref{wave-equation-general} is defined as
\beaa
\QQ[\psi]_{\aund \bund \mu\nu}&=& \pr_\mu\psia \pr_\nu \psib -\frac 1 2 \g_{\mu\nu} \pr_\lambda \psia \pr^\lambda \psib
\eeaa
  for $\aund=1,2,3,4$. 
  \begin{remark} Since from \eqref{eq:psia-solution} each $\psia$ is a solution to the wave equation, the above energy-momentum tensor can be interpreted as a symmetrization with respect to $\psia$ and $\psib$ of their respective energy-momentum tensors. The underlined indices of $\QQ$ denote the conformal symmetry operators applied to $\psi$ while the Greek indices denote the spacetime indices as in the standard energy-momentum tensor.
  \end{remark}

\item Let $\mathbf{X}$ be a symmetric double-indexed collection of vector fields $\mathbf{X}=\{ X^{\underline{a} \underline{b}}\}$ and $\mathbf{w}$ be a symmetric double-indexed collection of functions $\mathbf{w}=\{ w^{\underline{ab}} \}$.  The generalized current associated to $(\mathbf{X}, \mathbf{w})$ is defined as 
 \bea\label{definition-of-gener-P}
 \PP_\mu^{(\mathbf{X}, \mathbf{w})}[\psi]&=&\QQ[\psi]_{\aund \bund \mu\nu} X^{\aund\bund\nu} +\frac 1 2  w^{\aund\bund}  \psia \c \pr_\mu \psib   -\frac 1 4(\pr_\mu w^{\aund \bund} )\psia  \c \psib
  \eea
  for $\aund=1,2,3,4$.

  \item The energy associated to  $(\mathbf{X}, \mathbf{w})$ on the hypersurface $\Sigma_\tau$ is
  \beaa
E^{(\mathbf{X}, \mathbf{w})}[\psi](\tau)&=&\int_{\Sigma_\tau} \PP^{(\mathbf{X}, \mathbf{w})}_\mu[\psi] n_{\Sigma_\tau}^\mu.
\eeaa
\end{itemize}

As in \eqref{le:divergPP-gen}, we obtain  for the divergence of the generalized $\PP$:
  \bea
  \label{le:divergPP-gen-gen}
  \D^\mu \PP_\mu^{(\mathbf{X}, \mathbf{w})}[\psi]&=& \frac 1 2 \QQ[\psi]_{\aund\bund}  \c \D_{(\mu} X^{\aund\bund}_{\nu)}-\frac 1 4 \square_\g w^{\aund\bund} \psia \c \psib+\frac 12  w^{\aund\bund} (\pr_\lambda \psia \pr^\lambda \psib),
 \eea
and we denote 
 \bea
 \label{eq:modified-div-gen}
 \EE^{(\mathbf{X}, \mathbf{w})}[\psi]&:=&   \D^\mu  \PP_\mu^{(\mathbf{X}, \mathbf{w})}[\psi].
  \eea

By applying the divergence theorem to $\PP_\mu^{(\mathbf{X}, \mathbf{w})}$ within a region such as $\MM(\tau_1, \tau_2)$ for carefully chosen $(\mathbf{X}, \mathbf{w})$ one obtains the associated energy identity:
\bea
E^{(\mathbf{X}, \mathbf{w})}[\psi](\tau_2) + \int_{\mathcal{H}^+(\tau_1, \tau_2)} \PP^{(\mathbf{X}, \mathbf{w})}_\mu[\psi] n_{\mathcal{H}^+}^\mu +\int_{\MM(\tau_1, \tau_2)}\EE^{(\mathbf{X}, \mathbf{w})}[\psi] = E^{(\mathbf{X}, \mathbf{w})}[\psi](\tau_1),
\eea
where the induced volume forms are to be understood.

\subsection{The relevant vectorfields and the spacetime current identities}\label{section:relevant-vectorfields}

In applying the vectorfield method to derive non-degenerate energy-Morawetz estimates for the wave equation, we make use of  the following vectorfields:
\begin{itemize}
\item  A radial vectorfield  $X=\FF(r) \pr_r$, for a well chosen function $\FF$, to obtain \textbf{Morawetz estimates},
\item A timelike vectorfield $\That_{\chi}=T + \chi \om_{\mathcal{H}}Z$, where $\chi=1$ for $\{r < r_{1}\}$ for some $r_1>r_{+}$, and $\chi=0$ for $\{r> r_2\}$, with a smooth decrease in $[r_1, r_2]$. In particular, $\That_{\chi}=\That_{\mathcal{H}}$ close to the horizon, and $\That_{\chi}=T$ for $r>r_2$. Also, $\That_\chi$ is not Killing only for $r \in [r_1, r_2]$. This is used to obtain \textbf{energy estimates}.
\end{itemize}

We collect here some relevant computations involving the deformation tensors of the above vectors which will be used in the next sections.

\begin{lemma} Let $\psi$ be a solution to the wave equation \eqref{wave-equation-general}, then the following relations hold true.
\begin{itemize}
\item For $X=\FF(r) \pr_r$, we have
\bea
\piX^{\a\b}=|q|^{-2} \Big( 2\De^{3/2}\pr_r \big(\frac{\FF}{\De^{1/2}} \big)\pr_r^\a\pr_r^\b- \FF\pr_r\big(\frac 1 \De\RR^{\a\b}\big) \Big) +|q|^{-2} X\big(|q|^2\big) \g^{\a\b},
\eea
and therefore
         \begin{equation}\lab{eq:PPpi(X)}
         \begin{split}
  |q|^2   \QQ[\psi]  \c\piX     &= 2\De^{3/2}\pr_r \big(\frac{\FF}{\De^{1/2}} \big)|\pr_r \psi|^2- \FF\pr_r\big(\frac 1 \De\RR^{\a\b}\big) \pr_\a\psi \pr_\b \psi  \\
  &+\Big( X\big(|q|^2\big) - |q|^2(\div X) \Big) \pr_\lambda \psi \pr^\lambda \psi.
  \end{split}
    \end{equation}
   \item For $\That_{\chi}=T + \chi \om_{\mathcal{H}}Z$, we have
   \bea
\, ^{(\That_\chi)} \pi^{\a\b}&=&2|q|^{-2}\De  \om_{\mathcal{H}}(\pr_r\chi) \pr_\vphi^{(\a} \pr_r^{\b)}  , 
\eea 
and therefore
       \bea
    \lab{eq:QQpi(That_de)}
   |q|^2  \QQ[\psi]  \c \,^{(\That_\chi)} \pi&=&2\De  \om_{\mathcal{H}}(\pr_r\chi)  \pr_\vphi\psi \pr_r \psi.
  \eea

\end{itemize}
\end{lemma}
\begin{proof} Using the expression for the inverse metric \eqref{inverse-metric-Kerr}, we compute
\beaa
 \LL_X(|q|^2 \g^{\a\b})&=&\LL_X\big(\De \pr_r^\a \pr_r^\b \big)+\LL_X\big(\frac 1 \De\RR^{\a\b} \big)= X(\De) \pr_r^\a\pr_r^\b+ \De [X, \pr_r]^\a \pr_r^\b +\De \pr_r^\a [X, \pr_r]^\b  +\LL_X\big(\frac 1 \De\RR^{\a\b} \big).
\eeaa
Consequently, for $X=\FF \partial_r$, we obtain
\beaa
 \LL_X(|q|^2 \g^{\a\b})&=& \FF (\pr_r\De) \pr_r^\a\pr_r^\b+ \De [\FF \pr_r, \pr_r]^\a \pr_r^\b+ \De \pr_r^\a [\FF \pr_r, \pr_r]^\b  +\FF \LL_{\pr_r} \big(\frac 1 \De\RR^{\a\b} \big)\\
 &=& \FF (\pr_r\De) \pr_r^\a\pr_r^\b-2 \De (\pr_r \FF)\pr_r^\a \pr_r^\b +\FF \partial_r \big(\frac 1 \De\RR^{\a\b} \big)\\
  &=&-2\De^{3/2}\pr_r \big(\frac{\FF}{\De^{1/2}} \big) \pr_r^\a\pr_r^\b +\FF \partial_r \big(\frac 1 \De\RR^{\a\b} \big).
\eeaa
For $\That_{\chi}=T + \chi \om_{\mathcal{H}}Z$ we obtain
\beaa
 \LL_{\That_{\chi}}(|q|^2 \g^{\a\b})&=& \De [\That_{\chi}, \pr_r]^\a \pr_r^\b +\De \pr_r^\a [\That_{\chi}, \pr_r]^\b=-2\De  \om_{\mathcal{H}}(\pr_r\chi) \pr_\vphi^{(\a} \pr_r^{\b)}  .
\eeaa
In particular observe that $ ^{(\That_\chi)} \pi^{\mu\nu}\g_{\mu\nu}=0$. By writing
\beaa
 \piX^{\a\b}&=&-\LL_X\big( |q|^{-2}  |q|^2 \g^{\a\b}\big)=-|q|^{-2} \LL_X\big(  |q|^2 \g^{\a\b}\big)- |q|^2\LL_X\big(|q|^{-2}\big)\g^{\a\b}
 \eeaa
 we obtain the stated expressions for the deformation tensors.

Finally, for any vectorfield $X$, we write
  \beaa
     \QQ[\psi]  \c\piX&=& \QQ[\psi]_{\a\b} \piX^{\a\b}= \big( \pr_\a\psi \pr_\b \psi -\frac 1 2 \g_{\a\b} \pr_\lambda \psi \pr^\lambda \psi\big) \piX^{\a\b} \\
     &=&\piX^{\a\b} \pr_\a\psi \pr_\b \psi - (\div X) \pr_\lambda \psi \pr^\lambda \psi
    \eeaa
    since $\g_{\mu\nu} \piX^{\mu\nu}=\g_{\mu\nu} \D^{(\mu}X^{\nu)}=2\div X$. Using the above expressions for $\piX$ we obtain the stated identities.

\end{proof}

We make use of the above computations to derive the Morawetz current for $\EE^{(X, w)}[\psi]$ with $X=\FF \partial_r$ and its generalized version $X^{\aund\bund} =\FF^{\aund\bund} \pr_r$. We mostly follow notations in \cite{And-Mor}. See also \cite{GKS2}.

  \begin{proposition}   
        \lab{proposition:Morawetz1}
       Let $z$ be a given function of $r$. The following identities hold true.
        \begin{enumerate}
        \item
Let $u$ be a given function of $r$. Then for
\bea
\lab{def-w-red-in-fun-FF-00-wave}\label{definition-w-}
X= \FF \partial_r, \qquad \quad \FF=z u,  \qquad \quad  w = z \pr_r u ,  \lab{Equation:w0}
\eea
      the current $\EE^{(X, w)}[\psi]$ satisfies
  \bea
  \lab{identity:prop.Morawetz1}
   |q|^2\EE^{(X, w)}[\psi] &=\AA |\pr_r\psi|^2 + \UU^{\a\b}(\pr_\a \psi )(\pr_\b \psi )+\VV |\psi|^2,
   \eea
   where
\bea
 \AA&=&z^{1/2}\Delta^{3/2} \partial_r\left( \frac{ z^{1/2}  u }{\Delta^{1/2}}  \right), \lab{eq:coeeficientsUUAAVV}  \\
  \UU^{\a\b}&=& -  \frac{ 1}{2}  u \pr_r\left( \frac z \De\RR^{\a\b}\right), \lab{eq:coeeficientsUUAAVV2}\\
\VV&=&-\frac 1 4 \pr_r \big(\De \pr_r w \big)= -\frac 1 4 \pr_r\big(\De \pr_r \big(
 z \pr_ru  \big)  \big) .\lab{eq:coeeficientsUUAAVV3}
\eea
\item  Let  $u^{\aund\bund}$  be a given double-indexed function of $r$. Then for
\bea\lab{choice-FF-w-operator}
X^{\aund\bund} =\FF^{\aund\bund} \pr_r, \qquad \quad \FF^{\aund\bund}= z u^{\aund\bund}, \qquad \quad w^{\aund\bund}=z \partial_r u^{\aund\bund} ,
\eea
the generalized current $\EE^{(\mathbf{X}, \mathbf{w})}[\psi] $ satisfies
\bea\label{generalized-current-operator}
 |q|^2 \EE^{(\mathbf{X}, \mathbf{w})}[\psi]   &=&\AA^{\aund\bund} \, \pr_r\psia \pr_r\psib + \UU^{\a\b\aund\bund} \, \pr_\a \psia \, \pr_\b \psib  +\VV^{\aund\bund} \, \psia\psib 
   \eea
   where
\bea
  \AA^{\aund\bund}&=&z^{1/2}\Delta^{3/2} \partial_r\left( \frac{ z^{1/2}  u^{\aund\bund} }{\Delta^{1/2}}  \right)  \label{eq:coefficients-commuted-eq-AA} ,
 \\
  \UU^{\a\b\aund\bund}&=& -  \frac{ 1}{2}  u^{\aund\bund} \pr_r\left( \frac z \De\RR^{\a\b}\right), \label{eq:coefficients-commuted-eq-UU}\\
   \VV^{\aund\bund}&=&-\frac 1 4 \pr_r \Big(\De \pr_r w^{\aund\bund} \Big)= -\frac 1 4 \pr_r\big(\De \pr_r \big(
 z \pr_ru^{\aund\bund}  \big)  \big) \label{eq:coefficients-commuted-eq-VV}.
\eea

    \end{enumerate}
        \end{proposition}

\begin{proof} We prove the first part of the Proposition, as the second part in the case of generalized current follows in the same way.
    Using \eqref{eq:modified-div}, \eqref{le:divergPP-gen} and \eqref{eq:PPpi(X)} we compute for $(X=\FF\partial_r, w)$,
\beaa
|q|^2\EE^{(\FF \pr_r, w)}[\psi]&=& \frac 1 2 |q|^2\QQ[\psi]  \c\piX-\frac 1 4 |q|^2\square_\g w |\psi|^2+\frac 12  |q|^2w (\pr_\lambda \psi \pr^\lambda \psi)\\
&=& \De^{3/2}\pr_r \big(\frac{\FF}{\De^{1/2}} \big)|\pr_r \psi|^2- \frac 1 2 \FF\pr_r\big(\frac 1 \De\RR^{\a\b}\big) \pr_\a\psi \pr_\b \psi  -\frac 1 4|q|^2 \square_\g w |\psi|^2\\
  &&+\frac 1 2 \Big( X\big(|q|^2\big) - |q|^2(\div X)+ |q|^2 w \Big) \pr_\lambda \psi \pr^\lambda \psi.
\eeaa
By defining an intermediate function $w_{int}$ as
     \beaa
   \frac 1 2 \Big(    X\big( |q|^2\big)- |q|^2 \div X+ |q|^2  w\Big)=\frac 1 2 |q|^2 \Big( |q|^{-2} X\big( |q|^2\big)-\div X +w\Big)=:-\frac 1 2 |q|^2w_{int},
       \eeaa
       and using \eqref{inverse-metric-Kerr} to write
       \beaa
      |q|^2 \pr_\lambda \psi \pr^\lambda \psi&=& |q|^2\g^{\lambda \mu} \pr_\lambda \psi \pr_\nu \psi=\Delta |\pr_r \psi|^2+\frac{1}{\De} \RR^{\a\b}\pr_\a \psi \pr_\b \psi,
       \eeaa
       we simplify the above to
       \beaa
       |q|^2\EE^{(\FF \pr_r, w)}[\psi]&=& \De^{3/2}\pr_r \big(\frac{\FF}{\De^{1/2}} \big)|\pr_r \psi|^2- \frac 1 2 \FF\pr_r\big(\frac 1 \De\RR^{\a\b}\big) \pr_\a\psi \pr_\b \psi  -\frac 1 4|q|^2 \square_\g w |\psi|^2\\
  &&-\frac 1 2w_{red}\Big(\Delta |\pr_r \psi|^2+\frac{1}{\De} \RR^{\a\b}\pr_\a \psi \pr_\b \psi \Big)\\
  &=&\Big( \De^{3/2}\pr_r \big(\frac{\FF}{\De^{1/2}} \big)-\frac 1 2w_{int}\Delta \Big) |\pr_r \psi|^2- \frac 1 2\Big( \FF\pr_r\big(\frac 1 \De\RR^{\a\b}\big) + w_{int}\frac{1}{\De} \RR^{\a\b}\Big) \pr_\a \psi \pr_\b \psi\\
  && -\frac 1 4|q|^2 \square_\g w |\psi|^2.
       \eeaa
To summarize, we write the above as
  \beaa
     |q|^2\EE^{(\FF \pr_r, w)}[\psi]&=\AA |\pr_r\psi|^2 + \UU^{\a\b}(\pr_\a \psi )(\pr_\b \psi )+\VV |\psi|^2,
   \eeaa
   where
   \beaa
   \AA&=& \De^{3/2}\pr_r \big(\frac{\FF}{\De^{1/2}} \big)-\frac 1 2w_{int}\Delta\\
   \UU^{\a\b}&=&  -\frac 1 2  \FF\pr_r \left(\frac 1 \De\RR^{\a\b}\right)-\frac 1 2   w_{int}\frac 1 \De \RR^{\a\b}\\
   \VV&=&-\frac 1 4|q|^2 \square_\g  w .
   \eeaa
   with $w=  |q|^2 \div \big( |q|^{-2}  X \big)-w_{int}$.
       Observe that for a function $z$ we can write
 \beaa
\UU^{\a\b} &=&  -\frac 1 2  \FF\pr_r \left(\frac 1 \De\RR^{\a\b}\right)-\frac 1 2   w_{int}\frac 1 \De \RR^{\a\b}\\
&=&  - \frac{ 1}{2}\FF z^{-1}  \partial_r\left( \frac z \De\RR^{\a\b}\right)+ \frac{ 1}{2}\left(\FF z^{-1}\partial_r z  -w_{int}\right) \frac{ \RR^{\a\b}}{\Delta}
 \eeaa
Setting $\FF=zu$ for a function $u$, and choosing  $w_{int}=  \FF z^{-1}\partial_r z=u \pr_r z $,  the coefficient of $\frac{ \RR^{\a\b}}{\Delta}$ cancels out, and we deduce the stated expression for $\UU^{\a\b}$ in \eqref{eq:coeeficientsUUAAVV2}. 

With these choices for $\FF$ and $w_{int}$, we compute the function $w$:
\beaa
w &=&  |q|^2 \D_\a\big( |q|^{-2}   \FF\pr_r ^\a \big)-w _{int}=   |q|^2\pr_r\big(|q|^{-2}  \FF\big)+\FF( \D_\a\pr_r^\a) -u \pr_r z
\eeaa
Observe that
\beaa
\D_\a \pr_r^\a=\frac{1}{\sqrt{|\g|}} \pr_\a\big( \sqrt{|\g|} \pr_r^\a \big)= \frac{1}{\sqrt{|\g|}} \pr_r \big( \sqrt{|\g|} \big)=\frac{1}{|q|^2} \pr_r \big(|q|^2\big),
\eeaa
and therefore
\beaa
w &=&   |q|^2\pr_r\big(|q|^{-2}  zu \big)+zu |q|^{-2} \pr_r \big(|q|^2\big) -u \pr_r z= \pr_r\big(  zu \big) -u \pr_r z= z \pr_r u.
\eeaa
We also compute
\beaa
\AA&=&   \partial_r\left(\frac{\FF}{\Delta^{1/2}}  \right) \Delta^{3/2}-\frac 12\Delta   w_{int} =\partial_r\left(\frac{zu}{\Delta^{1/2}}  \right) \Delta^{3/2}-\frac 12\Delta   ( \partial_r z) u\\
&=& \frac 1 2 \partial_rz  \frac{  u}{\Delta^{1/2}}   \Delta^{3/2}+z^{1/2} \partial_r\left(\frac{ z^{1/2}  u}{\Delta^{1/2}}  \right) \Delta^{3/2}-\frac 12\Delta   ( \partial_r z)u =z^{1/2}\Delta^{3/2} \partial_r\left( \frac{ z^{1/2}  u}{\Delta^{1/2}}  \right).
\eeaa
Finally, as for a function $H=H(r)$, 
\beaa
 \square_\g H  =\frac{1}{\sqrt{|\g|} } \pr_\a \big(\sqrt{|\g|} \g^{\a\b} \pr_\b\big)  H= \frac{1}{\sqrt{|g|} } \pr_r  \big(\sqrt{|\g|} \g^{rr} \pr_r \big)  H=\frac{1}{|q|^2} \pr_r \big( \De\pr_r H\big)
\eeaa
we compute
\beaa
 |q|^2 \square_\g  w=\pr_r \big(\De \pr_r w \big)= \pr_r\big(\De \pr_r \big(
 z \pr_ru  \big)  \big),
\eeaa
as stated.

\end{proof}

 \subsection{Trapped null geodesics}\label{section:trapped-null-geodesics}
 
 We now dedicate this section to the derivation of the equation satisfied by the trapped null geodesics in Kerr-Newman spacetime, see also \cite{GKS2}. Trapped null geodesics are relevant for the analysis of the wave equation, as in the high frequency limit waves behave like null geodesics and the integrated local energy estimates necessarily have to degenerate at the trapping region of the black hole.

Let $\ga(\la)$ be a null geodesic in Kerr-Newman spacetime. Using the expression for the inverse of the metric given by \eqref{inverse-metric-Kerr}, along $\ga(\la)$, since $\g(\gadot, \gadot)=0$ we have, with $\gadot_r=\pr_r^\a \gadot_\a$, $\gadot_t=\pr_t^\a \gadot_\a$, $\gadot_\phi=\pr_\vphi^\a \gadot_\a$
\beaa
0= |q|^2 \g^{\a\b} \gadot_\a \gadot_\b=\big( \De \pr_r^\a \pr_r^\b +\frac{1}{\De}\RR^{\a\b}\big)\gadot_\a \gadot_\b=
\De\gadot_r \gadot_r+ \frac{1}{\De} \RR^{\a\b}\gadot_\a \gadot_\b
\eeaa
with
\bea\label{eq:RR-ab-geodesics}
\RR^{\a\b}\gadot_\a \gadot_\b&=& -(r^2+a^2) ^2\gadot_t \gadot_t- 2a(r^2+a^2) \gadot_t \gadot_\phi- a^2  \gadot_\phi \gadot_\phi  +\De O^{\a\b} \gadot_\a \gadot_\b
\eea

Since $\pr_t=T$ and $\pr_\phi=Z$ are Killing vectorfields we deduce that  $\gadot_t=\g(\gadot, T)$ and $\gadot_\phi= \g(\gadot, Z)$ are constants of the motion i.e. constants along  $\ga$, and respectively called the energy and the azimuthal angular momentum. We write,
\beaa
\e&:=&-\g(\gadot, T) , \qquad \lz=-\g(\gadot, Z).
\eeaa
We also define\footnote{Observe that $\k^2$ is a positive constant of motion by definition of $K$. }
\beaa
\k^2&:=&K^{\a\b} \gadot_\a\gadot_\b
\eeaa
for the Carter tensor $K$ in Kerr-Newman. Since $K$ is Killing, $\k^2$ is also a constant of motion.

With these constants from \eqref{eq:RR-ab-geodesics}  we have
\beaa
\RR(r;   \e,\lz, \k^2):=\RR^{\a\b}\gadot_\a \gadot_\b &=& -(r^2+a^2) ^2 \e^2- 2a(r^2+a^2) \e \c \lz - a^2  \lz^2  +\De \k^2\\
&=&-\big(( r^2+a^2)\e + a \lz\big)^2+\De\k^2
\eeaa
which is only a function of $r$ along  any fixed  null geodesic. 
Going back to the equation for null geodesics  we  infer that
\beaa
\De\Big(\frac{dr}{d\la} \Big)^2  =-\RR(r;  \e,\lz, \k^2),
\eeaa
which is the equation for a null geodesic with constant of motions $\e$, $\lz$ $\k^2$.

There exist null geodesics along which  $\RR(r;  \e,\lz, \k^2)=0$  i.e. $r$ remains constant. These are called 
orbital null geodesics, or trapped null geodesics.  The $r$ values for which such solutions are possible  must then verify the equations
\beaa
\RR(r;  \e,\lz, \k^2)=\pr_r\RR(r;  \e,\lz, \k^2)=0.
\eeaa

\begin{lemma}\label{lemma:trapped-geodesics} All orbital null geodesics in Kerr-Newman spacetime are given by the equation
\bea\label{eq:trapped-region-KN}
\TT_{\e, \lz}:= \big( r^3-3Mr^2 + (a^2+2Q^2)r+Ma^2\big)\e-  (r-M) a\lz=0.
\eea
\end{lemma}
\begin{proof} 
We solve for 
\beaa
\RR(r; \e,\lz, \k^2)=-\big(( r^2+a^2)\e + a \lz\big)^2+\De\k^2=0\\
\pr_r\RR(r;  \e,\lz, \k^2)=-4r \e \c \big(( r^2+a^2)\e + a \lz\big)+2\big(r-M\big)\k^2=0
\eeaa
Writing from the second equation $\k^2=\frac{2r}{(r-M)} \e \c \big(( r^2+a^2)\e + a \lz\big)$, and substituting in the first equation, we obtain
\beaa
0&=&\big(( r^2+a^2)\e + a \lz\big) \c \Big( -\big(( r^2+a^2)\e + a \lz\big)+\De\frac{2r}{(r-M)} \e \Big)\\
&=&\frac{\big(( r^2+a^2)\e + a \lz\big)}{(r-M)} \c \Big( \big(-( r^2+a^2)(r-M)+2r\De \big) \e -(r-M) a \lz \Big)\\
&=&\frac{\big(( r^2+a^2)\e + a \lz\big)}{(r-M)} \c \Big( \big(-( r^3-Mr^2+a^2r-Ma^2)+2r(r^2-2Mr+a^2+Q^2) \big) \e -(r-M) a \lz \Big)\\
&=&\frac{\big(( r^2+a^2)\e + a \lz\big)}{(r-M)} \c \Big( \big( r^3-3Mr^2+(a^2+2Q^2)r+Ma^2 \big) \e -(r-M) a \lz \Big).
\eeaa
Observe that the case of $( r^2+a^2)\e + a \lz=0$ implies $\k^2=0$, and it yields the vanishing of all the constants of motion. For non-trivial trapped null geodesics, we then obtain the vanishing of the factor on the right, denoted by $\TT_{\e, \lz}$. 
\end{proof}

For $a=0$, the above orbital null geodesic equation gives
\beaa
\TT_{\e}:=r \big( r^2-3Mr +2Q^2\big)\e=0,
\eeaa
and therefore the trapped null geodesics lie on the hypersurface, called \textit{photon sphere}, defined by the polynomial $r^2-3Mr  +2Q^2=0$, which in Reissner-Nordstr\"om spacetime is at $\{r^{RN}_{trap}=\frac{3M+\sqrt{9M^2-8Q^2}}{2}\}$ and in Schwarzschild for $Q=0$ is at $\{r^S_{trap}=3M\}$.

For $a \neq 0$, as a consequence of Lemma \ref{lemma:trapped-geodesics} the values of $r$ for which trapped null geodesics exist  depends on  the ratio $\lz/ \e$. More precisely, at the trapped null geodesics we have
\beaa
\frac{r^3-3Mr^2 + ( a^2+2Q^2)r+Ma^2}{r-M}=\frac{ a\lz}{\e}.
\eeaa
In particular, unlike Schwarzschild and Reissner-Nordstr\"om, which have a single radius where all trapped geodesics occur, Kerr or Kerr-Newman spacetimes have an entire radial interval where trapped geodesics can occur. However, for a fixed geodesic angular momentum $\lz$, there is again only a single trapping radius for null geodesics with that angular momentum. 
If $\lz=0$ the trapped region defined by \eqref{eq:trapped-region-KN} reduces to a unique hypersurface defined by
\bea\label{definition-TT}
 \TT:= r^3-3Mr^2 + ( a^2+2Q^2)r+Ma^2 =0.
\eea
Geometrically, this reflects the fact that trapped null geodesics orthogonal to the axial Killing vectorfield $Z$ must necessarily approach the root of $\TT$.
Observe that the polynomial $\TT$ has a unique root in the exterior of the black hole region, and we denote that root by $r_{trap}$. Moreover $2M < r_{trap} < 3M$, where the lower bound is reached in the extremal Reissner-Nordstr\"om case and the upper bound in the Schwarzschild case. 
At the extremal Kerr-Newman, for $a^2+Q^2=M^2$, the trapping hypersurface becomes
\beaa
\TT=r^3-3Mr^2 +  (2M^2-a^2)r+Ma^2=(r-M)(r^2-2Mr-a^2)
\eeaa
which vanishes at $r=M+\sqrt{M^2+a^2}$. 

\begin{remark}\label{remark-derivative-TT}
From the computations in Lemma \ref{lemma:trapped-geodesics}, one can see that the polynomial $\TT$ can be obtained as
\bea
\pr_r\left(\frac{\De}{(r^2+a^2)^2} \right)&=& \frac{2(r-M)(r^2+a^2)-4r\De}{(r^2+a^2)^3}=-\frac{2\TT}{(r^2+a^2)^3}.
\eea
In particular, the trapping radius $r_{trap}$ in axial symmetry maximizes the geodesic potential $\frac{\De}{(r^2+a^2)^2} $, which coincides with the familiar potential $r^{-2} \big( 1-\frac{2M}{r}\big) $ in Schwarzschild.
\end{remark}

A crucial property of the trapped null geodesics in Kerr-Newman spacetime is that they are unstable, i.e. one can show \cite{GKS2} that  at $\TT_{\e, \lz}=0$, we have $\pr_r^2 \RR(r;  \e,\lz, \k^2) \leq 0$ in the subextremal range $a^2+Q^2 \leq M^2$.

\subsection{The choice of the function $z$ in the Morawetz estimates}

Here we motivate the choice of the function $z$ done in \cite{And-Mor}, which we use here. Recall the expression for the Morawetz current in \eqref{identity:prop.Morawetz1}
\beaa
   |q|^2\EE^{(X, w)}[\psi] &=\AA |\pr_r\psi|^2 + \UU^{\a\b}(\pr_\a \psi )(\pr_\b \psi )+\VV |\psi|^2,
   \eeaa
   where the principal term $ \UU^{\a\b}(\pr_\a \psi )(\pr_\b \psi )$ is given by
\beaa
  \UU^{\a\b}&=& -  \frac{ 1}{2}  u \pr_r\left( \frac z \De\RR^{\a\b}\right).
\eeaa
We denote 
\beaa
\RRtp'^{\a\b}:= \pr_r\left( \frac z \De\RR^{\a\b}\right), \qquad \RRtp'^{\aund}:=  \pr_r\left( \frac z \De\RR^{\aund}\right) 
\eeaa
and thus write from \eqref{eq:RR-ab-RR-aund},  i.e. $\RR^{\a\b}= \RR^\aund S_\aund^{\a\b}$,
\beaa
\RRtp'^{\a\b}= \pr_r\left( \frac z \De\RR^{\a\b}\right)= \pr_r\left( \frac z \De \RR^\aund  S_{\aund}^{\a\b}\right)=  \pr_r \Big(\frac{z}{\De} \RR^\aund \Big) S_{\aund}^{\a\b} =\RRtp'^{\aund} S_{\aund}^{\a\b}.
\eeaa
Explicitly, from  \eqref{definition-RR-tensor}, we deduce
\beaa
\RRtp'^{\a\b}&=& -\pr_r\left(\frac{z}{\De} ( r^2+a^2)^2 \right) \partial_t^\a \partial_t^\b - 2a \pr_r\left(\frac{z}{\De} (r^2+a^2)\right) \partial_t^{(\a} \partial_\vphi^{\b)} - a^2\pr_r\left(\frac{z}{\De} \right)\partial_\vphi^\a \partial_\vphi^\b+(\pr_rz)  O^{\a\b},
\eeaa
and
\beaa
\RRtp'^1=   -\pr_r\left(\frac{z}{\De} ( r^2+a^2)^2 \right), \qquad \RRtp'^2= - 2a \pr_r\left(\frac{z}{\De} (r^2+a^2)\right), \qquad \RRtp'^3= - a^2\pr_r\left(\frac{z}{\De} \right), \qquad \RRtp'^4 =\pr_rz.
\eeaa

The principal term $ \UU^{\a\b}= -  \frac{ 1}{2}  u\RRtp'^{\a\b}$   contains the angular and time derivatives of the solution $\psi$, and therefore we expect it to be degenerate at the trapping region. The choice of $z$ has to reflect this property.
From Remark \ref{remark-derivative-TT}, following \cite{And-Mor} we notice that if 
\bea\label{choice-z}
z_0&=& \frac{\De}{(r^2+a^2)^2}
\eea 
we simultaneously obtain a degenerate coefficient at trapping for the term $O^{\a\b}$ and the vanishing of the term in $\pr_t^\a \pr_t^\b$, i.e.
\beaa
\RRtp'^{\a\b}[z_0]=\pr_r\left( \frac {z_0}{ \De}\RR^{\a\b}\right)&=&- 2a \pr_r\left( \frac{1}{(r^2+a^2)} \right) \partial_t^{(\a} \partial_\vphi^{\b)} - a^2\pr_r\left(\frac{1}{(r^2+a^2)^2} \right)\partial_\vphi^\a \partial_\vphi^\b-\frac{2\TT}{(r^2+a^2)^3}  O^{\a\b}\\
&=&     \frac{4ar}{(r^2+a^2)^2}  \partial_t^{(\a} \partial_\vphi^{\b)} +\frac{4a^2r}{(r^2+a^2)^3}  \partial_\vphi^\a \partial_\vphi^\b-\frac{2\TT}{(r^2+a^2)^3}  O^{\a\b}, 
\eeaa
or also
\bea
\RRtp'^1[z_0]=   0, \qquad \RRtp'^2[z_0]= \frac{4ar}{(r^2+a^2)^2} , \qquad \RRtp'^3[z_0]= - \frac{4a^2r}{(r^2+a^2)^3}  , \qquad \RRtp'^4[z_0] =-\frac{2\TT}{(r^2+a^2)^3} .
\eea

Recalling the definition \eqref{define:That} of
$\That=\pr_t+\frac{a}{r^2+a^2} \pr_\vphi$, for the choice of $z_0= \frac{\De}{(r^2+a^2)^2}$ we obtain
\bea\label{eq:term-UU-z0}
\RRtp'^{\a\b}[z_0]&=&   -\frac{2\TT}{(r^2+a^2)^3}  O^{\a\b}+  \frac{4ar}{(r^2+a^2)^2}  \That^{(\a} Z^{\b)}.
\eea

Observe that if $z=\big( \frac{\De}{(r^2+a^2)^2}\big)^2$, then the coefficient of $\partial_t^\a \partial_t^\b$ in $\pr_r\left( \frac z \De\RR^{\a\b}\right)$ will reduce to
\beaa
 -\pr_r\left(\frac{z}{\De} ( r^2+a^2)^2 \right)&=&  -\pr_r\left( \frac{\De}{(r^2+a^2)^2} \right)=-\pr_r z_0=\frac{2\TT}{(r^2+a^2)^3},
\eeaa
which is also trapped. To combine the above choice of $z_0$ with one which allows for a trapped coefficient in the time derivative as well,  following \cite{And-Mor} we define for a sufficiently small $\ep >0$,
\bea\label{eq:definition-z1}
z_1=z_0-\ep \, z_0^2.
\eea
With this choice we obtain
\begin{equation}\label{eq:expression:z1RR}
\begin{split}
\RRtp'^{\a\b}[z_1]=\pr_r\left(  \frac {z_1}{ \De}\RR^{\a\b}\right)&=  -\ep \frac{2\TT}{ (r^2+a^2)^3}\partial_t^\a \partial_t^\b  -\frac{2\TT}{(r^2+a^2)^3}  \big(1+O(\ep r^{-2}) \big)  O^{\a\b}\\
&+  \frac{4ar}{(r^2+a^2)^2} \big(1+O(\ep r^{-2}) \big)   \That^{(\a} Z^{\b)},
\end{split}
\end{equation}
or also
\begin{equation}\label{eq:values-RRtp'z1}
\begin{split}
\RRtp'^1[z_1]&=  -\ep \frac{2\TT}{ (r^2+a^2)^3}, \qquad \RRtp'^4[z_1] =   - \frac{2\TT}{ (r^2+a^2)^3}\big( 1+O(\ep r^{-2})\big) \\
\RRtp'^2[z_1]& =\frac{4ar}{(r^2+a^2)^2} \big( 1+O(\ep r^{-2})\big), \qquad \RRtp'^3[z_1]  =\frac{4a^2r}{(r^2+a^2)^3} \big( 1+O(\ep r^{-2})\big).
\end{split}
\end{equation}

\subsection{The case of axial symmetry}\label{section:axysimmetry}

From Proposition   \ref{proposition:Morawetz1}, for $z=z_0= \frac{\De}{(r^2+a^2)^2}$ as in \eqref{choice-z} and $u$ a function of $r$, with
\bea
X= \FF \partial_r, \qquad \quad \FF=z_0 u,  \qquad \quad  w = z_0 \pr_r u ,  \lab{Equation:w}
\eea
      the current $\EE^{(X, w)}[\psi]$ satisfies
  \bea
   |q|^2\EE^{(X, w)}[\psi] &=\AA |\pr_r\psi|^2 + \UU^{\a\b}(\pr_\a \psi )(\pr_\b \psi )+\VV |\psi|^2,
   \eea
   where, recall \eqref{eq:term-UU-z0}, 
\begin{equation}\label{eq:expressions-AA-VV}
\begin{split}
 \AA&=z_0^{1/2}\Delta^{3/2} \partial_r\left( \frac{ z_0^{1/2}  u }{\Delta^{1/2}}  \right)=\frac{\Delta^{2}}{r^2+a^2} \partial_r\left( \frac{ u }{r^2+a^2}  \right),  \\
  \UU^{\a\b}&= -  \frac{ 1}{2}  u \pr_r\left( \frac {z_0}{ \De}\RR^{\a\b}\right)=u \Big( \frac{\TT}{(r^2+a^2)^3}  O^{\a\b}-  \frac{2ar}{(r^2+a^2)^2}  \That^{(\a} Z^{\b)}\Big), \\
\VV&=-\frac 1 4 \pr_r \big(\De \pr_r w \big) .
\end{split}
\end{equation}
For axially symmetric solutions, the term involving $Z$ in $\UU^{\a\b}$ vanishes, and we can write
 \bea\label{current-mor}
   |q|^2\EE^{(X, w)}[\psi] &=&\AA |\pr_r\psi|^2 +  \frac{u \TT}{(r^2+a^2)^3}  |q|^2 |\nab \psi|^2+\VV |\psi|^2,
   \eea
   where we used \eqref{eq:first-definition-O} to write 
   \beaa
   O^{\a\b}(\pr_\a \psi )(\pr_\b \psi )= |q|^2 ( e_1^\a e_1^\b + e_2^\a e_2^\b)(\pr_\a \psi )(\pr_\b \psi )=|q|^2 ((\nab_1\psi)^2+(\nab_2\psi)^2)=:|q|^2 |\nab \psi|^2.
   \eeaa
   Here $|\nab \psi|^2=\frac{1}{|q|^2} |\nabb_{\mathbb{S}^2} \psi|^2+\frac{a^2\sin^2\th}{|q|^2} |\partial_t \psi|^2$, where $ |\nabb_{\mathbb{S}^2} \psi|^2$ is the norm of the gradient of $\psi$ on the unit round sphere.

In order to obtain a positive definite Morawetz current in \eqref{current-mor}, we make use of the following construction first due to Stogin, see Lemma 5.2.6 in \cite{Stogin}.

\begin{lemma}\label{lemma:construction-functions} It is possible to choose functions $u$ and $w$ such that in Kerr-Newman for the full sub-extremal $a^2+Q^2< M^2$, we have
\bea\label{eq:conditions-positivity}
\AA \geq 0, \qquad u\TT \geq 0, \qquad \VV \geq 0,
\eea
and therefore the current $\EE^{(X, w)}[\psi] $ is positive definite and satisfies for a uniform constant $c_0 >0$,
\bea\label{eq:first-bound-EE-X-w}
\EE^{(X, w)}[\psi] & \geq &  c_0 \Big(  \frac{\Delta^2}{r^7} |\pr_r\psi|^2 +  r^{-1} \big( 1-\frac{r_{trap}}{r}\big)^2 |\nab \psi|^2+\frac{M}{r^4} \mathbbm{1}_{\{r \geq r_{*}\}} |\psi|^2\Big)
\eea
where $r_{trap}$ denotes the root of $\TT$ in the exterior region and $r_{*}> r_{trap}$ is defined in the construction.

\end{lemma}
\begin{proof}
In order to obtain the positivity conditions in \eqref{eq:conditions-positivity}, from the expressions of $\AA$ and $\VV$ in \eqref{eq:expressions-AA-VV} we need to have
   \bea
       \partial_r\left( \frac{ u }{r^2+a^2}  \right) \geq 0, \label{condition-partial-u-AA}\\
    u  \TT \geq 0,  \label{condition-u-trapped}\\
     \pr_r\big(\De \pr_r w  \big) \leq 0, \label{condition-VV-positive} 
     \eea
     with the compatibility condition
     \bea
     w=z_0 \pr_r u \label{condition-w-pr-u}.
   \eea

We start by imposing condition \eqref{condition-u-trapped}, i.e. $u \TT\geq 0$. Recall that $r_{trap}$ is the root of the polynomial $\TT$ in the exterior region, with $2M < r_{trap} < 3M$. 
  We define $u$ by the following:
   \bea
   u(r_{trap})=0, \qquad \pr_r u=\frac{1}{z_0} w=\frac{(r^2+a^2)^2}{\De} w, \qquad w \geq 0.
   \eea    
   This automatically implies that $u$ vanishes at $r_{trap}$ and is increasing, therefore satisfying \eqref{condition-u-trapped}. Also, by definition, the compatibility condition \eqref{condition-w-pr-u} is satisfied.
      
We now rewrite condition \eqref{condition-partial-u-AA} in terms of $w$, i.e.
\bea\label{relation-AA-}
 \partial_r\left( \frac{ u }{r^2+a^2}  \right) &=& \frac{1}{r^2+a^2} \pr_r u -\frac{2r}{(r^2+a^2)^2}u=   \frac{r^2+a^2}{\De}w -\frac{2r}{(r^2+a^2)^2}u
\eea
To eliminate the dependence on $u$, we multiply \eqref{relation-AA-} by $\frac{(r^2+a^2)^2}{r}$, and take another derivative in $r$.
By defining $\KK:= \frac{(r^2+a^2)^2}{r} \partial_r\left( \frac{ u }{r^2+a^2}  \right)$, we obtain
\beaa
\KK&=&\frac{(r^2+a^2)^2}{r} \partial_r\left( \frac{ u }{r^2+a^2}  \right)=\frac{(r^2+a^2)^3}{r \De} w  -2u\\
\pr_r \KK&=&\pr_r\big(\frac{(r^2+a^2)^3}{r \De} w\big)  -2\frac{(r^2+a^2)^2}{\De}w \\
&=&2r \big(\frac{(r^2+a^2)^2}{r \De} w\big) +(r^2+a^2)\pr_r\big(\frac{(r^2+a^2)^2}{r \De} w\big)  -2\frac{(r^2+a^2)^2}{\De}w\\
&=& (r^2+a^2)\pr_r\big(\frac{(r^2+a^2)^2}{r \De} w\big) =(r^2+a^2)\pr_r\big(\frac{1}{r z_0} w\big) .
\eeaa
Since $\KK$ has the same sign as $ \partial_r\left( \frac{ u }{r^2+a^2}  \right)$, condition \eqref{condition-partial-u-AA} is satisfied if and only if $\KK\geq 0$. 

We now impose that $\partial_r \KK=0$ for sufficiently large $r$, which, together with the condition that $\KK \geq 0$ up to that large $r$, implies positivity for $\KK$. From the above we have
\bea
 w=rz_0=\frac{r \De}{(r^2+a^2)^2} \qquad \text{for sufficiently large $r$}
\eea
We then compute for such choice:
\beaa
\pr_r w&=&\pr_r\big( \frac{r \De}{(r^2+a^2)^2} \big)=\frac{ \De}{(r^2+a^2)^2}  -\frac{2r\TT}{ (r^2+a^2)^3} \\
&=& \frac{ (r^2-2Mr+a^2+Q^2) (r^2+a^2)-2r( r^3-3Mr^2 + ( a^2+2Q^2)r+Ma^2 )}{(r^2+a^2)^3}\\
&=&- \frac{ r^4-4Mr^3+3Q^2r^2+4Ma^2r-a^4-a^2Q^2}{(r^2+a^2)^3}.
\eeaa
The function $w$ then has a maximum at the root of the above polynomial. Denote $r_{*}$ the unique root in the exterior region. Observe that in Schwarzschild $r_*=4M$, in Reissner-Nordstr\"om $r_{*}=2M+\sqrt{4M^2-3Q^2}$, and in Kerr $r_{*}=2M+\sqrt{4M^2-a^2}$. In the subextremal Kerr-Newman we have $3M < r_{*} < 4M$, and therefore in the full sub-extremal range we have $r_{*} >r_{trap}$. 

 We then impose 
\bea
 w=rz=\frac{r \De}{(r^2+a^2)^2} \qquad \text{for $r \geq r_{*}$}.
\eea
Then for $r \geq r_{*}$ condition \eqref{condition-VV-positive} reduces to
\beaa
- \pr_r\big(\De  \pr_r w  \big)&=& \pr_r\big(\frac{\De}{r^2+a^2}  \frac{ r^4-4Mr^3+3Q^2r^2+4Ma^2r-a^4-a^2Q^2}{(r^2+a^2)^2} \big) \geq 0.
 \eeaa
 The above is a derivative of a product of two functions that are positive for $r > r_{*}$ and increasing to $1$, and therefore necessarily positive. 

 Consider the function $\De \pr_r w$. This function vanishes at $r=r_{+}$ and at $r=r_{*}$ with the above choice of $w$ for $r \geq r_{*}$. By the mean value theorem, in order to have condition  \eqref{condition-VV-positive}, satisfied for all $r\geq r_{+}$, we need to have the vanishing of the function in the interval $[r_{+}, r_{*}]$, and therefore $\partial_r w=0$ there. Define
\bea
w(r)=w(r_{*})>0 \qquad \text{for $r_{+} \leq r \leq r_{*}$}
\eea
At this stage, a non-negative $w$ has been chosen for all $r$, and condition \eqref{condition-VV-positive}  has been proved to be satisfied everywhere.

 We are now left with proving that with the above choice $\KK \geq 0$ for $r \leq r_{*}$. In this region, we have
\beaa
\pr_r \KK&=& (r^2+a^2)\pr_r\big(\frac{1}{r z_0} w\big) =w(r_{*}) (r^2+a^2)\pr_r\big(\frac{1}{r z_0} \big) =-w(r_{*}) (r^2+a^2)\frac{1}{r^2z_0^2}\pr_r\big(rz_0\big).
\eeaa
Recall that the function $rz_0$ had a maximum at $r_{*}$ and therefore $\KK$ decreases to a minimum at $r=r_{*}$. It is therefore enough to check that $\KK(r=r_{*})> 0$, or equivalently that $\partial_r\left( \frac{ u }{r^2+a^2}  \right)\big|_{r=r_{*}}>0$. We have
\beaa
\partial_r\left( \frac{ u }{r^2+a^2}  \right)&=& \frac{1}{r^2+a^2}\partial_r u-\frac{2r}{(r^2+a^2)^2}u\\
&=& \frac{1}{r^2+a^2}\frac{w}{z_0}-\frac{2r}{(r^2+a^2)^2}\int_{r_{trap}}^{\cdot} \partial_ru\\
&=& \frac{1}{r^2+a^2}\frac{w}{z_0}-\frac{2r}{(r^2+a^2)^2}\int_{r_{trap}}^{\cdot} \frac{w}{z_0}
\eeaa
We evaluate the above for $r$ with $r_{trap} \leq r \leq r_{*}$, where $w(r)=w(r_{*})$ is constant, and $\partial_r(\frac{1}{z_0})=-\frac{1}{z_0^2} \partial_r z_0=\frac{1}{z_0^2} \frac{2\TT}{(r^2+a^2)^3} \geq 0$. Since the  function $\frac{1}{z_0}$ is increasing, we can bound its integral  from above by the product of the interval times the value of the function at the right end of the interval. We therefore obtain
\beaa
\partial_r\left( \frac{ u }{r^2+a^2}  \right)\big|_{r=r_{*}}&=& \frac{1}{r_{*}^2+a^2}\frac{w(r_{*})}{z_0(r_{*})}-\frac{2r_{*} w(r_{*})}{(r_{*}^2+a^2)^2}\int_{r_{trap}}^{r_{*}} \frac{1}{z_0}\\
&\geq & \frac{1}{r_{*}^2+a^2}\frac{w(r_{*})}{z_0(r_{*})}-\frac{2r_{*} w(r_{*})}{(r_{*}^2+a^2)^2} \frac{1}{z_0(r_{*})}(r_{*}-r_{trap})\\
&=& \frac{1}{(r_{*}^2+a^2)^2}\frac{w(r_{*})}{z_0(r_{*})} \Big(r_{*}^2+a^2- 2r_{*}(r_{*}-r_{trap})\Big)
\eeaa
Observe that
\beaa
r_{*}^2+a^2- 2r_{*}(r_{*}-r_{trap})=r_{*}(2r_{trap}-r_{*})+a^2>0
\eeaa
as can be seen by comparing the range of $r_{trap}$ and $r_{*}$. 

To summarize, the choice for the functions  $u$ and $w$ given by
\[
w= \left\{
     \begin{array}{@{}l@{\thinspace}l}
       r_{*}z_0(r_{*}) &\qquad  r \leq r_{*}\\
       rz_0(r) &\qquad r > r_{*} \\
     \end{array}
   \right.
\]
and 
\beaa
 u(r_{trap})=0, \qquad \pr_r u=\frac{1}{z_0} w
\eeaa
satisfy in the whole exterior region of the full subextremal Kerr-Newman spacetime the conditions \eqref{eq:conditions-positivity}. From the analysis of the asymptotics of the functions $u$ and $w$ we can easily deduce the bound \eqref{eq:first-bound-EE-X-w}. 

\end{proof}

As in \cite{Stogin}\cite{KS}, the above construction has to be corrected to fix some remaining issue in the bound \eqref{eq:first-bound-EE-X-w}. We briefly describe them here, and we remind to \cite{Stogin}\cite{KS} for more details, as they do not depend on the Kerr-Newman spacetime specifically.

The function $u$ as defined above, and consequently the vectorfield $X$, presents a logarithmic divergence near the horizon. This can be fixed by tempering it in that region through a small deviation from the original ones. More precisely one can define $(X_{\de}, w_{\de})$ which agree with $(X, w)$ outside a neighborhood of the horizon, and that are regular up to the event horizon. This will create a negative contribution in the lower order term $|\psi|^2$. 

The coefficient of $(\partial_r \psi)^2$ in \eqref{eq:first-bound-EE-X-w} vanishes at the horizon. This can be fixed by the standard procedure of making use of the redshift vectorfield, which gives a positive contribution for any spacetime in the sub-extremal range. 

The coefficient of $|\psi|^2$ vanishes in the interval $r_{+} \leq r \leq r_{*}$. In addition, because of the tempering of the vectorfield close to the horizon, we have created a negative term of size $\de$ in the lower order term. This can be fixed by borrowing extra-positivity from the coefficient of $(\partial_r \psi)^2$ through a local Hardy inequality.

In  \eqref{eq:first-bound-EE-X-w}, there is no control on the $(\partial_t \psi)^2$ derivative. This can be fixed through a standard procedure of making use of the Lagrangian of the wave equation, i.e. applying the vectorfield method with $X=0$ and $w=\epsilon w_{\epsilon}$. (Equivalently, one can use the function $z_1$ in \eqref{eq:definition-z1} to insert a degenerate term at trapping for $(\partial_t \psi)^2$, as done in Section \ref{section:slowly}.)

After applying these standard arguments, one arrives to an improved bound on the current $\EE^{(X, w)}[\psi] $ given by
\bea\label{eq:second-bound-EE-X-w}
\EE^{(X, w)}[\psi] & \geq &  c_0 \Bigg(  \frac{M^2}{r^3} |\pr_r\psi|^2 +   \big( 1-\frac{r_{trap}}{r}\big)^2 \Big( r^{-1} |\nab \psi|^2+\frac{M}{r^2} (\partial_t \psi)^2 \Big)+\frac{M}{r^4}  |\psi|^2\Bigg).
\eea
Observe that the boundary term of the Morawetz current $\PP^{(X, w)}[\psi]$ will be combined in the next subsection with the boundary terms from the energy estimates multiplied by a large constant, so that the overall boundary terms are positive definite.

\subsection{The Morawetz estimate}\label{section:Morawetz-estimates}

We derive here the Morawetz estimates for general solutions to the wave equation \eqref{wave-eq-gen} for $|a| \ll M$.  Observe that the main issue is that the principal term in \eqref{eq:term-UU-z0},  i.e.
\beaa
\UU^{\a\b}= - \frac 1 2 u \pr_r\left(  \frac {z_0}{ \De}\RR^{\a\b}\right)&=&   \frac{u\TT}{(r^2+a^2)^3}  O^{\a\b}-  u\frac{2a  r}{(r^2+a^2)^2}  \That^{(\a} Z^{\b)},
\eeaa
cannot be made to be positive definite, as the function $u$ has to change sign in the trapped region. 
On the other hand, if one considers the generalized vectorfield method applied to $\psia$, as obtained in \cite{And-Mor}, it is possible to extract a positive definite contribution which is degenerate at the trapped region.

In the following, we will show the main steps in the proof of the physical-space analysis in \cite{And-Mor}, while referring to \cite{And-Mor} for more details.

\subsubsection{Choice of double-indexed function $u^{\aund\bund}$}

Following \cite{And-Mor}, see also \cite{GKS2}, we perform an integration by parts  in the principal term $\UU^{\a\b\aund\bund} \, \pr_\a \psia \, \pr_\b \psib$, which allows to create a positive definite term for a trapped combination of $\psia$, which we denote $\Psi$ here.

\begin{lemma}\label{lemma:integration-parts}
Let $u^{\aund\bund}$ the double-indexed function of $r$ as described in the second part of Proposition \ref{proposition:Morawetz1}  be given by\footnote{No assumption on $z$ is needed in this Lemma.} 
      \bea
      u^{\aund\bund}=-h \RRtp'^{\aund} \LL^{\bund}, \qquad  \RRtp'^{\aund}=  \pr_r\left( \frac z \De\RR^{\aund}\right),
      \eea
where $h$ is a positive function, and  $\LL^{\aund}$ are the coefficient of a constant symmetric tensor $L^{\a\b}:=\LL^\aund S_\aund^{\a\b}$. Then, defining 
   \bea\label{eq:definition-Psi}
      \Psi:= \RRtp'^\aund \psia,
      \eea
 the Morawetz identity \eqref{generalized-current-operator} is given by
\bea\label{eq:obunds-divergence-L}
 |q|^2 \D^\a \Big( \PP_\a^{(\mathbf{X}, \mathbf{w})}[\psi]+  \BB_\a[\psi]\Big)&=&\AA^\aund  \LL^{\bund}  \pr_r\psia \pr_r\psib + \frac 1 2 h   L^{\a\b} \pr_\a     \Psi \,  \pr_\b      \Psi  +\VV^\aund \LL^{\bund} \psia\psib,
 \eea
      where
\bea
\AA^\aund:=-z^{1/2}\Delta^{3/2} \partial_r\left( \frac{ z^{1/2}  h \RRtp'^{\aund} }{\Delta^{1/2}}  \right), \qquad  \VV^\aund:=\frac 1 4 \pr_r\big(\De \pr_r \big(
 z \pr_r(h \RRtp'^{\aund} ) \big)  \big),
\eea
      and
 \bea\label{eq:def-BB-a}
 \BB_\a[\psi]&:=& \frac{ 1}{2}|q|^{-2} h  \LL^{\bund} \RRtp'^{\cund}      \Psi \, \big( S_\bund^{\a\b}  \pr_\b \psic - S_{\cund}^{\a\b} \, \pr_\b \psib\big)
 \eea
 denotes a boundary term. 
\end{lemma}
\begin{proof} From Proposition \ref{proposition:Morawetz1}, we write the principal term as
\beaa
  \UU^{\a\b\aund\bund}\pr_\a \psia \, \pr_\b \psib= -  \frac{ 1}{2}  u^{\aund\bund} \RRtp'^{\a\b}\pr_\a \psia \, \pr_\b \psib=-  \frac{ 1}{2}  u^{\aund\bund} \RRtp'^{\cund} S_{\cund}^{\a\b}\pr_\a \psia \, \pr_\b \psib
  \eeaa
  Integrating by parts in $\partial_\a$, we can write the above as
  \beaa
  |q|^{-2}  \UU^{\a\b\aund\bund}\pr_\a \psia \, \pr_\b \psib&=&-  \frac{ 1}{2} |q|^{-2} u^{\aund\bund} \RRtp'^{\cund} S_{\cund}^{\a\b}\pr_\a \psia \, \pr_\b \psib\\
    &=&-  \frac{ 1}{2} \D_\a\big( |q|^{-2} u^{\aund\bund} \RRtp'^{\cund} S_{\cund}^{\a\b} \psia \, \pr_\b \psib\big)+ \frac{ 1}{2}  \pr_\a \big(u^{\aund\bund} \RRtp'^{\cund} \big) |q|^{-2}S_{\cund}^{\a\b}\psia \, \pr_\b \psib\\
    &&+ \frac{ 1}{2}  u^{\aund\bund} \RRtp'^{\cund}  \psia \, \pr_\a\big( |q|^{-2}S_{\cund}^{\a\b} \pr_\b \psib\big).
  \eeaa
  Since $u^{\aund\bund}$ and $\RRtp'^{\cund}$ only depend on $r$, while the derivatives in $S_\cund^{\a\b}$ do not contain derivatives in $r$, the second term on the right hand side above vanishes. 
  Using Lemma \ref{lemma:D-S-SS}  and the commutation property \eqref{eq:commutator-SS-aund-bund}, we write the last term as
  \beaa
   \pr_\a\big( |q|^{-2} S_{\cund}^{\a\b} \pr_\b \psib\big)&=&  \pr_\a\big( |q|^{-2} S_{\cund}^{\a\b} \pr_\b  \SS_\bund \psi\big)= \SS_\cund  \SS_\bund \psi = \SS_\bund  \SS_\cund \psi=\pr_\a( |q|^{-2} S_\bund^{\a\b}  \pr_\b \psic) .
  \eeaa
We then obtain
  \beaa
 |q|^{-2}   \UU^{\a\b\aund\bund}\pr_\a \psia \, \pr_\b \psib    &=& \frac{ 1}{2}  u^{\aund\bund} \RRtp'^{\cund}  \psia \, \pr_\a(|q|^{-2}S_\bund^{\a\b}  \pr_\b \psic)-  \frac{ 1}{2} \D_\a\big( |q|^{-2} u^{\aund\bund} \RRtp'^{\cund} S_{\cund}^{\a\b} \psia \, \pr_\b \psib\big)
      \eeaa
      Repeating the integration by parts procedure to the first term as above, we obtain
    \beaa
  |q|^{-2}  \UU^{\a\b\aund\bund}\pr_\a \psia \, \pr_\b \psib    &=& \frac{ 1}{2}  \D_\a\big( |q|^{-2} u^{\aund\bund} \RRtp'^{\cund}  \psia \, S_\bund^{\a\b}  \pr_\b \psic \big)- \frac{ 1}{2} \pr_\a\big( u^{\aund\bund} \RRtp'^{\cund} \big) \psia \,  |q|^{-2} S_\bund^{\a\b}  \pr_\b \psic\\
  &&- \frac{ 1}{2}   |q|^{-2} u^{\aund\bund} \RRtp'^{\cund}  \pr_\a \psia \, S_\bund^{\a\b}  \pr_\b \psic-  \frac{ 1}{2} \D_\a\big(  |q|^{-2}u^{\aund\bund} \RRtp'^{\cund} S_{\cund}^{\a\b} \psia \, \pr_\b \psib\big)\\
     &=&- \frac{ 1}{2}   |q|^{-2} u^{\aund\bund}S_\bund^{\a\b}  \RRtp'^{\cund}  \pr_\a \psia \,  \pr_\b \psic+ \frac{ 1}{2}  \D_\a\big( |q|^{-2}u^{\aund\bund} \RRtp'^{\cund}  \psia \, \big( S_\bund^{\a\b}  \pr_\b \psic - S_{\cund}^{\a\b} \, \pr_\b \psib\big) \big)
      \eeaa
      Notice that in this way one creates the first term given by $ u^{\aund\bund}S_\bund^{\a\b}$. In order to separate such dependence and create a quadratic expression in the above, we define $u^{\aund\bund}$ in the following way:
      \beaa
      u^{\aund\bund}=-h \RRtp'^{\aund} \LL^{\bund},
      \eeaa
      for some positive function $h$ and some constant symmetric tensor $L^{\a\b}=\LL^\aund S_\aund^{\a\b}$.
      
      Then the first term of the above relation becomes
      \beaa
      - \frac{ 1}{2}  u^{\aund\bund}S_\bund^{\a\b}  \RRtp'^{\cund}  \pr_\a \psia \,  \pr_\b \psic&=& \frac 1 2 h \RRtp'^\aund \LL^\bund S_\bund^{\a\b}  \RRtp'^{\cund}  \pr_\a \psia \,  \pr_\b \psic= \frac 1 2 h \RRtp'^\aund  \RRtp'^{\cund}  L^{\a\b} \pr_\a \psia \,  \pr_\b \psic\\
      &=& \frac 1 2 h   L^{\a\b} \pr_\a\big(  \RRtp'^\aund \psia\big) \,  \pr_\b \big(  \RRtp'^{\cund}\psic\big)
      \eeaa
      where the terms in $\RRtp'^\aund$ and $\RRtp'^\bund$ can be inserted inside the derivatives as they only depend on $r$, where the same indices means summation over those. By denoting $\Psi:= \RRtp'^\aund \psia$ we obtain 
          \beaa
    |q|^{-2} \UU^{\a\b\aund\bund}\pr_\a \psia \, \pr_\b \psib        &=& \frac 1 2 |q|^{-2} h   L^{\a\b} \pr_\a\Psi \,  \pr_\b \Psi- \frac{ 1}{2}  \D_\a\big( |q|^{-2} h  \LL^{\bund} \RRtp'^{\cund} \Psi \, \big( S_\bund^{\a\b}  \pr_\b \psic - S_{\cund}^{\a\b} \, \pr_\b \psib\big) \big).
      \eeaa
which gives
                \beaa
 \UU^{\a\b\aund\bund}\pr_\a \psia \, \pr_\b \psib        &=& \frac 1 2  h   L^{\a\b} \pr_\a\Psi \,  \pr_\b \Psi-   |q|^2\D_\a \BB_\a[\psi],
      \eeaa
      for $\BB_\a[\psi]$ defined in \eqref{eq:def-BB-a}. 
Finally, with the choice of $u^{\aund\bund}=-h \RRtp'^{\aund} \LL^{\bund}$, the expression of $\AA$ and $\UU$ in \eqref{eq:coefficients-commuted-eq-AA} and \eqref{eq:coefficients-commuted-eq-VV} become
 \beaa
   \AA^{\aund\bund}&=&  \AA^\aund  \LL^{\bund}, \qquad \qquad \AA^\aund:=-z^{1/2}\Delta^{3/2} \partial_r\left( \frac{ z^{1/2}  h \RRtp'^{\aund} }{\Delta^{1/2}}  \right)
 \\
   \VV^{\aund\bund}&=&  \VV^\aund \LL^{\bund} , \qquad \qquad  \VV^\aund:=\frac 1 4 \pr_r\big(\De \pr_r \big(
 z \pr_r(h \RRtp'^{\aund} ) \big)  \big).
 \eeaa
Recalling that $|q|^2 \EE^{(\mathbf{X}, \mathbf{w})}[\psi]=  |q|^2 \D^\mu  \PP_\mu^{(\mathbf{X}, \mathbf{w})}[\psi]$, we obtain the relation \eqref{eq:obunds-divergence-L}, and prove the lemma.  

\end{proof}

We now put into effect the choice of function $z_1$ as given in \eqref{eq:definition-z1}, i.e.
\beaa
z_1=z_0-\ep \, z_0^2.
\eeaa
From \eqref{eq:expression:z1RR} and \eqref{eq:values-RRtp'z1}, we have
\beaa
\RRtp'^{\a\b}&=&  -\ep \frac{2\TT}{ (r^2+a^2)^3}\partial_t^\a \partial_t^\b  -\frac{2\TT}{(r^2+a^2)^3}  \big(1+O(\ep r^{-2}) \big)  O^{\a\b}+  \frac{4ar}{(r^2+a^2)^2} \big(1+O(\ep r^{-2}) \big)   \That^{(\a} Z^{\b)},
\eeaa
and therefore for $\Psi=\RRtp'^{\aund}\psia$, 
\bea\label{eq:Psi-choice-z}
 \Psi&=& -\ep \frac{2\TT}{ (r^2+a^2)^3}\pr_t^2 \psi  -\frac{2\TT}{(r^2+a^2)^3}  \big(1+O(\ep r^{-2}) \big)  \OO(\psi)+  \frac{4ar}{(r^2+a^2)^2} \big(1+O(\ep r^{-2}) \big)   \That( \pr_\vphi\psi).
\eea
By considering as in \cite{And-Mor} the constant symmetric tensor given by
\bea\label{eq:choice-L}
L^{\a\b}=\ep \pr_t^\a \pr_t^\a + O^{\a\b}
\eea
then for a positive function $h$, the principal term $ \frac 1 2 h   L^{\a\b} \pr_\a     \Psi \,  \pr_\b      \Psi $ is positive and given by
\beaa
\frac 1 2 h   L^{\a\b} \pr_\a     \Psi \,  \pr_\b      \Psi &=& \frac 1 2 h \big(\ep | \pr_t     \Psi|^2+ O^{\a\b}\pr_\a \Psi \pr_\b \Psi\big)= \frac 1 2 h \big(\ep | \pr_t     \Psi|^2+ |q|^2 |\nab \Psi|^2\big)
\eeaa
where we used \eqref{eq:first-definition-O} to write 
   \beaa
   O^{\a\b}(\pr_\a \psi )(\pr_\b \psi )= |q|^2 ( e_1^\a e_1^\b + e_2^\a e_2^\b)(\pr_\a \psi )(\pr_\b \psi )=|q|^2 ((\nab_1\psi)^2+(\nab_2\psi)^2)=:|q|^2 |\nab \psi|^2.
   \eeaa
We can summarize the above and write $\frac 1 2 h   L^{\a\b} \pr_\a     \Psi \,  \pr_\b      \Psi=\frac 1 2 |D \Psi|^2$
where we denote $D=\ep^{1/2}\pr_t, |q|\nab$.
To express it in terms of the $\psia$, we use \eqref{eq:Psi-choice-z} and we bound the above by (see Lemma 3.13 in \cite{And-Mor})
\beaa
\frac 1 2 h   L^{\a\b} \pr_\a     \Psi \,  \pr_\b      \Psi &\geq&\frac 1 2 h \widetilde{ \mathbbm{1}}_{\{r \neq r^{RN}_{trap}\}}   |D \Psi|^2\\
&\geq&\frac 1 2 h \widetilde{ \mathbbm{1}}_{\{r \neq r^{RN}_{trap}\}}  \Big(\ep^2| D   \psi_1|^2+\sum_{\aund=2}^4|D \psia|^2 -\ep (|a|+ \epsilon)\sum_{\aund\neq\bund=1}^4 \sum_{D=\pr_t, |q|\nab} D \psia D \psib  \Big)\\
&&- a^2\frac 1 2 h \widetilde{ \mathbbm{1}}_{\{r \neq r^{RN}_{trap}\}} \sum_{\aund=1}^4|\psia|^2,
\eeaa
where $\widetilde{\mathbbm{1}}_{\{r \neq r^{RN}_{trap}\}}$ is a function that is identically 1 for $\{ |r-r^{RN}_{trap}| > \delta\}$ for some $\delta >0$ and zero otherwise.  Observe that once the commutation with the conformal symmetries allows to create a positive term with respect to $\Psi$, we can insert the trapped function $\widetilde{ \mathbbm{1}}_{\{r \neq r^{RN}_{trap}\}}$ to express it in terms of the $\psia$.

Choosing $a$ sufficiently small with respect to $\epsilon$ and fixing $\epsilon$ sufficiently small with respect to $1$, we obtain 
\bea\label{bound-h-L-Psi}
\frac 1 2 h   L^{\a\b} \pr_\a     \Psi \,  \pr_\b      \Psi &\geq& \sum_{\aund=1}^4\frac 1 2 h \widetilde{ \mathbbm{1}}_{\{r \neq r^{RN}_{trap}\}}  \big(| \pr_t  \psia |^2+ |q|^2 |\nab \psia|^2\big)- a^2\frac 1 2 h \widetilde{ \mathbbm{1}}_{\{r \neq r^{RN}_{trap}\}} \sum_{\aund=1}^4|\psia|^2.
\eea



\subsubsection{Choice of function $h$}

We now look at the terms $\AA^\aund  \LL^{\bund}  \pr_r\psia \pr_r\psib$ and $ \VV^\aund \LL^{\bund} \psia\psib$. For $z_1=z_0-\ep z_0^2$, we compute
\beaa
\AA^\aund &=&-z_1^{1/2}\Delta^{3/2} \partial_r\left( \frac{ z_1^{1/2}  h \RRtp'^{\aund} }{\Delta^{1/2}}  \right)\\
&=&-z_0^{1/2}\Delta^{3/2} \partial_r\left( \frac{( z_0^{1/2}+O(\ep r^{-3}) ) h \RRtp'^{\aund} }{\Delta^{1/2}}  \right)+O(\ep)\partial_r\left( \frac{ z_0^{1/2}  h \RRtp'^{\aund} }{\Delta^{1/2}}  \right)\\
&=&-\frac{\De^2}{r^2+a^2}  \partial_r\left( \frac{  h \RRtp'^{\aund} }{r^2+a^2}  \right)+O(\ep)\partial_r\left( \frac{ h \RRtp'^{\aund} }{r^2+a^2}  \right)
\eeaa

To obtain positivity of the those terms, we can for example make use of the construction in Lemma \ref{lemma:construction-functions}. In particular, we denote here by $u_{ax}$ and $w_{ax}$ the functions $u$ and $w=z_0 \pr_r u_{ax}$ constructed in Lemma \ref{lemma:construction-functions} in the axially symmetric case. We define\footnote{This choice differs from the one in \cite{And-Mor} or \cite{GKS2}, but the following procedure is identical, as it only makes use of the positivity of $\AA^1$ and $\AA^4$, and the fact that $\AA^2, \AA^3=O(a)$.}
\bea
h&:=&\frac{(r^2+a^2)^3 }{2\TT} u_{ax}
\eea
Observe that this function is positive and smooth everywhere, as both $u_{ax}$ and $\TT$ vanish at $r=r_{trap}$ of order 1. With this choice we have $h \RRtp'^{1}=-\ep u_{ax}$ and $h \RRtp'^4=-u_{ax}$, and therefore obtain
\beaa
\AA^1 &=&-\frac{\De^2}{r^2+a^2}  \partial_r\left( \frac{  h \RRtp'^{1} }{r^2+a^2}  \right)+O(\ep)\partial_r\left( \frac{ h \RRtp'^{1} }{r^2+a^2}  \right)=\ep  \frac{\De^2}{r^2+a^2} \partial_r\left( \frac{  u_{ax}   }{r^2+a^2}  \right)\big( 1+O(\ep r^{-2})\big),\\
\AA^4 &=&  \frac{\De^2}{r^2+a^2} \partial_r\left( \frac{  u_{ax}   }{r^2+a^2}  \right)\big( 1+O(\ep r^{-2})\big),
\eeaa
which are positive definite for $\ep$ sufficiently small by Lemma  \ref{lemma:construction-functions}. We also have
\beaa
\AA^2&=&- a \frac{\De^2}{r^2+a^2}  \partial_r\left( \frac{2r u_{ax}}{\TT} \big( 1+O(\ep r^{-2})\big) \right)+O(\ep r^{-2})=O(a r^{-2}) \frac{\De^2}{r^2+a^2} \big( 1+O(\ep r^{-2})\big)\\
\AA^3&=& -a^2\frac{\De^2}{r^2+a^2}  \partial_r\left( \frac{2r u_{ax}}{\TT(r^2+a^2)}\big( 1+O(\ep r^{-2})\big) \right)+O(\ep r^{-3})=O(a^2 r^{-3}) \frac{\De^2}{r^2+a^2} \big( 1+O(\ep r^{-2})\big).
\eeaa
We can then bound the term $\AA^\aund  \LL^{\bund}  \pr_r\psia \pr_r\psib$ by (see Lemma 3.9 in \cite{And-Mor})
\beaa
\AA^\aund  \LL^{\bund}  \pr_r\psia \pr_r\psib &=&( \AA^\aund    \pr_r\psia) \pr_r(\LL^{\bund}\psib )= (\AA^\aund    \pr_r\psia) \pr_r(\ep\pr_t^2 \psi + \OO(\psi))\\
&=& (\AA^1    \pr_r\pr_t^2 \psi+\AA^4    \pr_r\OO(\psi)+\AA^2    \pr_r\pr_t\pr_\vphi\psi+\AA^3    \pr_r\pr_\vphi^2\psi) \pr_r(\ep \pr_t^2 \psi + \OO(\psi))\\
&\geq&c_0\frac{\De^2}{r^5} \Big(\ep^2| \pr_r  \psi_1|^2+\sum_{\aund=2}^4|\pr_r \psia|^2  -|a| (|a|+ \epsilon)\sum_{\aund\neq\bund=1}^4 \pr_r \psia \pr_r \psib \Big)
\eeaa
where the term $ \pr_r \psi_1 \c \pr_r \OO\psi $ can be shown to be positive, up to boundary terms, by integration by parts, and being comparable to $|\pr_r(\pr_t \nab)\psi|^2$ (see Lemma 2.4 in \cite{And-Mor}).
As above, choosing $a$ sufficiently small with respect to $\epsilon$ and fixing $\epsilon$ sufficiently small with respect to $1$, we obtain positivity for the above term.

We similarly compute 
\beaa
\VV^\aund&=&\frac 1 4 \pr_r\big(\De \pr_r \big( z_1 \pr_r(h \RRtp'^{\aund} ) \big)  \big)=\frac 1 4 \pr_r\big(\De \pr_r \big( z_0 \pr_r(h \RRtp'^{\aund} ) \big)  \big)\big( 1+O(\ep r^{-2})\big).
\eeaa
With the given choice of $h$ we then obtain
\beaa
\VV^1&=&-\ep  \frac 1 4 \pr_r\big(\De \pr_r \big( z_0 \pr_ru_{ax} \big)  \big)\big( 1+O(\ep r^{-2})\big)=-\ep  \frac 1 4 \pr_r\big(\De \pr_r w_{ax}  \big)\big( 1+O(\ep r^{-2})\big)\\
\VV^4&=& -  \frac 1 4 \pr_r\big(\De \pr_r w_{ax}  \big)\big( 1+O(\ep r^{-2})\big),
\eeaa
which are non-negative for $\ep$ sufficiently small by Lemma \ref{lemma:construction-functions}. We also have 
\beaa
\VV^2&=&\frac 1 4 a \pr_r\big(\De \pr_r \big( z_0 \pr_r(\frac{2r u_{ax}}{\TT} ) \big)  \big)\big( 1+O(\ep r^{-2})\big)=O(ar^{-4} )\big( 1+O(\ep r^{-2})\big)\\
\VV^3&=&\frac 1 4 a^2 \pr_r\big(\De \pr_r \big( z_0 \pr_r\big(\frac{2r u_{ax}}{\TT(r^2+a^2)} \big) \big)  \big)\big( 1+O(\ep r^{-2})\big)=O(a^2r^{-5} )\big( 1+O(\ep r^{-2})\big).
\eeaa
We can then bound the term $ \VV^\aund \LL^{\bund} \psia\psib$ by (see Lemma 3.9 in \cite{And-Mor})
\beaa
 \VV^\aund \LL^{\bund} \psia\psib&=& ( \VV^\aund  \psia)( \LL^{\bund}\psib)=(\VV^1 \pr_t^2\psi+ \VV^4 \OO(\psi)+\VV^2\pr_t\pr_\vphi \psi + \VV^3 \pr_\vphi^2 \psi)(\ep \pr_t^2\psi+\OO(\psi))\\
 &\geq&c_0\frac{M}{r^4} \Big(\mathbbm{1}_{\{r \geq r_{*}\}} \big( \ep^2|  \psi_1|^2+\sum_{\aund=2}^4| \psia|^2 \big)  -|a| (|a|+ \epsilon)\sum_{\aund\neq\bund=1}^4  \psia \psib \Big).
\eeaa
Upon applying the Hardy inequality as in the axially symmetric case, we can upgrade the bound of the first terms to be valid everywhere in the exterior region. By combining the above bound with the one obtained in \eqref{bound-h-L-Psi}, and choosing $a$ sufficiently small with respect to $\epsilon$ and fixing $\epsilon$ sufficiently small with respect to $1$, we obtain positivity for the overall zero-th order term in the Morawetz bulk, which finally gives
 \bea\label{eq:final-bound-mor}
 \D^\a \Big( \PP_\a^{(\mathbf{X}, \mathbf{w})}[\psi]+ \widetilde{ \BB}_\a[\psi]\Big)\geq c_0 \Big(\frac{M^2}{r^5} |\pr_r \psi|_{\SS}^2+ \frac{M}{r^4}| \psi|_{\SS}^2+ \widetilde{ \mathbbm{1}}_{\{r \neq r^{RN}_{trap}\}}  \big( r^{-1}|\nab \psi|_{\SS}^2+\frac{M}{r^2} |\pr_t \psi|_{\SS}^2\big) \Big)
 \eea
where $ \widetilde{ \BB}_\a[\psi]$ incorporates the boundary term in \eqref{eq:def-BB-a}, together with the ones obtained in the above integration by parts.

\subsection{The energy estimate}\label{section:energy-estimates}

 The energy estimates are obtained from the current associated to the vectorfield $\That_{\chi}=T + \chi \om_{\mathcal{H}}Z$ as defined in Section \ref{section:relevant-vectorfields}. Observe that for $|a|/M$ sufficiently small, $\That_{\chi}$ is timelike everywhere in the exterior region and Killing outside the region $[r_1, r_2]$. Moreover, for $|a|/M$ sufficiently small, the trapped null geodesics remain close to the hypersurface $\{r=r_{trap}^{RN}\}$. In particular for $|a|/M \ll 1$, the region with $r \in[r_1, r_2]$, for $r_1=\frac 76 r_{+}$ and $r_2=\frac 8 9 r_{trap}^{RN}$, does not contain any trapped null geodesics. In particular, 
 $\That_{\chi}$ is Killing in the entire region where trapped null geodesics appear.
 From \eqref{eq:QQpi(That_de)}, we obtain
 \bea\label{eq:energy-identity-higher}
E^{(\That_{\chi}, 0)}[\psi](\tau)   +\int_{\MM(0, \tau) }|q|^{-2} \De  \om_{\mathcal{H}}(\pr_r\chi)  \pr_\vphi\psi \pr_r \psi\leq  E^{(\That_{\chi}, 0)}[\psi](0),
\eea
where 
\beaa
E^{(\That_{\chi}, 0)}[\psi](\tau) &\sim&  \int_{\Sigma_\tau}\frac{ \Delta}{r^2+a^2} (\pr_r \psi)^2+(\partial_t\psi)^2+|\nab \psi|^2.
\eeaa
The above energy can be made to be non-degenerate at the horizon by making use of the red-shift vectorfield.

By applying the energy estimates to each commuted equation \eqref{eq:psia-solution}, we obtain the higher order version:
\beaa
E^{(\That_{\chi}, 0)}[\psia](\tau)  +\int_{\MM(0, \tau) }|q|^{-2} \De  \om_{\mathcal{H}}(\pr_r\chi)  \pr_\vphi\psia \pr_r \psia\leq E^{(\That_{\chi}, 0)}[\psia](0)
\eeaa
Observe that the spacetime integral in \eqref{eq:energy-identity-higher} is different from zero in the support of $\pr_r \chi$, i.e. in $[r_1, r_2]$ outside the trapping region. Also, as $\om_{\HH}=O(|a|)$, we obtain
\beaa
\int_{\MM(0, \tau) }|q|^{-2} \De  \om_{\mathcal{H}}(\pr_r\chi)  \pr_\vphi\psia \pr_r \psia &\les& O(|a|)  \int_{\MM(0, \tau)} \mathbbm{1}_{[r_1, r_2]}\big( |\pr_r\psia|^2+|\pr_\vphi \psia|^2\big)\\
&\les& O(|a|) \mbox{Mor}_{(0, \tau)}[\psi]
\eeaa
By combining the bound obtained in \eqref{eq:final-bound-mor} with the energy estimates multiplied by a large constant,  one can choose a constant $\Lambda$ large enough in order to have the boundary terms of the Morawetz estimates absorbed by the positive ones from the non-degenerate energy estimates (see Lemma 3.11 in \cite{And-Mor}), and $|a|$ small enough to absorb the above on the right hand side. 
 This concludes the proof of Theorem \ref{theorem:general}.

 \section{Application to the Einstein-Maxwell equations}\label{Einstein-Maxwell-equations}
 
 In this section we show how the physical-space analysis relying on the commutation with the Carter differential operator can be adapted to the system of coupled Regge-Wheeler equations describing the coupled electromagnetic-gravitational perturbations of Kerr-Newman spacetime, as obtained in \cite{Giorgi7}. We first recall the system of generalized Regge-Wheeler equations, and then show how the commutations with adapted symmetry operators maintain the same structure of the system.  
 
 \subsection{The generalized Regge-Wheeler (gRW) system}

We recall the main theorem in \cite{Giorgi7}.

\begin{theorem}\label{main-theorem-RW}[Theorem 7.3. in \cite{Giorgi7}] Consider a linear electromagnetic-gravitational perturbation of Kerr-Newman spacetime $\g_{M, a, Q}$. Then we can define a complex horizontal 1-tensor $\pf \in \sk_1(\mathbb{C})$ and a symmetric traceless 2-tensor $\qf^\F \in \sk_2(\mathbb{C})$ that, as a consequence of the Einstein-Maxwell equations,
 satisfy
 the following coupled system of wave equations:
\bea
 \squared_1\pf-i  \frac{2a\cos\th}{|q|^2}\nab_T \pf  -V_1  \pf &=&4Q^2 \frac{\ov{q}^3 }{|q|^5} \left(  \ov{\DD} \c  \qf^\F  \right) + L_\pf[\Bfr, \Ffr] \label{final-eq-1}\\
\squared_2\qf^\F-i  \frac{4a\cos\th}{|q|^2}\nab_T \qf^\F -V_2  \qf^\F &=&-   \frac 1 2\frac{q^3}{|q|^5} \left(  \DD \hot  \pf  -\frac 3 2 \left( H - \Hb\right)  \hot \pf \right) + L_{\qf^\F}[\Bfr, \Ffr]\label{final-eq-2}
 \eea
where
\begin{itemize}
\item $\squared_1=\g^{\a\b}\Ddot_\a\Ddot_\b$ and $\squared_2=\g^{\a\b}\Ddot_\a\Ddot_\b$ denote the wave operators for horizontal 1-tensors and 2-tensors respectively, 
\item the potentials $V_1$ and $V_2$ are \textbf{real} positive scalar functions, which for $a=0$ coincide with the potentials of the Regge-Wheeler system of equations in Reissner-Nordstr\"om \cite{Giorgi7}, i.e.
\beaa
V_1&=&- \frac {1}{ 4} \trch\trchb +5 \rhoF^2+ O\big(\frac{|a|}{r^4}\big), \qquad V_2=- \trch \trchb  +2\rhoF^2+ O\big(\frac{|a|}{r^4}\big),
\eeaa
\item $L_\pf[\Bfr, \Ffr] $ and $L_{\qf^\F}[\Bfr, \Ffr]$ are linear first order operators in $\Bfr$ and $\Ffr$, which are lower order in terms of differentiability with respect to $\pf$ and $\qf^\F$.
\end{itemize}

We call the system of equations \eqref{final-eq-1}-\eqref{final-eq-2} a system of  \textbf{generalized Regge-Wheeler equations}.
\end{theorem}

Observe that on the right hand side of the equations, the coupling terms on the right hand side involving $ \ov{\DD} \c  \qf^\F $ and $ \DD \hot  \pf $ are proper of the Einstein-Maxwell case, while the left hand side of
 equation \eqref{final-eq-2} has the same structure as the generalized Regge-Wheeler equation in Kerr as obtained in \cite{GKS}, where the coupling term in $\pf$ does not appear.

 \subsubsection{Decomposition in spheroidal harmonics and non-commutativity with the system}\label{section:spheroidal}

We now recall why the decomposition in modes fails for the gRW system in perturbations of Kerr-Newman spacetime.

Following the physics literature and standard decomposition in modes for scalar functions, we consider the scalar projection of equations \eqref{final-eq-1} and \eqref{final-eq-2} to the first component of the tensors $\pf$ and $\qf^\F$, i.e. for $\psi^{[1]}=\pf(e_1)$ and $\psi^{[2]}=\qf^\F(e_1, e_1)$, where $\psi^{[s]}$ is a complex scalar of spin $s$. Then the projection of the above equations gives, see Appendix E of \cite{GKS}, 
\bea
 \squared_\g\psi^{[1]}+i \frac{2}{|q|^2}\frac{\cos\th}{\sin^2\th}  \pr_\vphi\psi^{[1]}-i  \frac{2a\cos\th}{|q|^2}\pr_t\psi^{[1]} -V_1 \psi^{[1]} &=&-4Q^2 \frac{\ov{q}^3 }{|q|^5}\Edth'_{2}\psi^{[2]}+ \lot \label{eq:1-proj}\\
\squared_\g\psi^{[2]}+i \frac{4}{|q|^2}\frac{\cos\th}{\sin^2\th}  \pr_\vphi\psi^{[2]}-i  \frac{4a\cos\th}{|q|^2}\pr_t \psi^{[2]} -V_2  \psi^{[2]}&=&   \frac 1 2\frac{q^3}{|q|^5} \Edth_{1}\psi^{[1]} +\lot\label{eq:2-proj}
 \eea
 where the operators $\Edth_s$ and $\Edth'_s$ are respectively rising and lowering-spin operators, given by
	\bea
	\Edth_{s} \psi^{[s]}&:=& \big(-\partial_\theta-\frac{i}{\sin\theta}\partial_\phi + s \cot\theta -  ia \sin\theta\partial_t  \big) \psi^{[s]}\label{eq:Edth}\\
	\Edth'_{s} \psi^{[s]}&:=& \big(-\partial_\theta+\frac{i}{\sin\theta}\partial_\phi - s \cot\theta + ia \sin\theta   \partial_t\big) \psi^{[s]}.\label{eq:Edth'}
	\eea
 The mode decomposition of the scalar complex functions 
\beaa
\psi^{[s[}(t, r, \th, \vphi) &=&e^{-i \omega t} e^{i m \vphi} R^{[s]}(r)S^{[s]}_{m\ell} ( a \om, \cos\th) 
\eeaa
involves the \textit{spin $s$-weighted spheroidal harmonics} $ S^{[s]}_{m\ell} ( a \om, \cos\th) $ which are eigenfunctions of the spin $s$-weighted Laplacian
	\beaa
	\begin{split}
	\lap^{[s]} &:= \frac{1}{\sin\th} \partial_\th (\sin\th \partial_\th) +\frac{1}{\sin^2\theta}\partial_\vphi^2+2is \frac{\cos\theta}{\sin^2\theta} \partial_\vphi + (s-s^2\cot^2\theta)+a^2\om^2\cos^2\theta -2a\om s \cos\theta  
	\end{split}
	\eeaa
	with parameter $_{s}\lambda_{\ell m}(a\om)$, i.e. $\lap^{[s]}( S^{[s]}_{m\ell} )=_{s}\lambda_{\ell m}(a\omega) S^{[s]}_{m\ell} $.

\medskip

When $a=0$, the above spin-weighted Laplacian reduce to the spherical spin-weighted Laplacian and  the spheroidal harmonics  reduce to the standard spherical harmonics $ S^{[s]}_{m\ell} ( 0, \cos\th) = Y^{[s]}_{m\ell} (\cos\th) $. Crucially, in spherical symmetry the operators \eqref{eq:Edth} and \eqref{eq:Edth'} precisely relate spherical harmonics of different spin. In fact in this case one can check \cite{Breuer} that the spherical Laplacian can be written as
\beaa
\Edth'_{s+1} \Edth_s=\lap^{[s]} , \qquad \Edth_{s-1} \Edth'_s=\lap^{[s]} -2s.
\eeaa
As a consequence, the operators $\Edth$ and $\Edth'$ simply relate spherical harmonics of different spins. More precisely, \cite{Breuer}
\beaa
\Edth_sY^{[s]}_{m\ell}= \big( (\ell-s) (\ell+s+1)\big)^{1/2} Y^{[s+1]}_{m\ell}, \qquad \Edth'_sY^{[s]}_{m\ell}=- \big( (\ell+s) (\ell-s+1)\big)^{1/2} Y^{[s-1]}_{m\ell}.
\eeaa
Consequently, the  operators $\Edth$ and $\Edth'$
appearing on the right hand side of  \eqref{eq:1-proj} and \eqref{eq:2-proj} commute with the decomposition in spherical harmonics, and in spherical symmetry (i.e. for electromagnetic-gravitational perturbations of Reissner-Nordstr\"om) one is able to decompose the equations in modes.

\medskip
On the other hand, in the general axisymmetric case, as in Kerr or Kerr-Newman, the spin-weighted spheroidal harmonics of different spins are not simply related through the angular operators  $\Edth$ and $\Edth'$. In fact one can show that for separated solutions we have \cite{Breuer}
 	\beaa
\Edth'_{s+1} \Edth_s &=&\lap^{[s]}   +2a\om (2s+1) \cos\theta    -a^2\om^2   +2a\om  m \label{Edth'-Edth-s-nu}\\
\Edth_{s-1} \Edth'_s &=&\lap^{[s]}  + 2a\om(2s-1) \cos\theta  -2s -a^2\om^2  +2a\om m \label{Edth-Edth'-s-1-nu},
\eeaa
and as a consequence given a spheroidal harmonic $ S^{[s]}_{m\ell} $, the scalar functions $\Edth_s S^{[s]}_{m\ell} $ or $ \Edth'_sS^{[s]}_{m\ell} $ do not describe spheroidal harmonics of higher or lower spins. In particular, in Kerr-Newman the  operators $\Edth$ and $\Edth'$
appearing on the right hand side of  \eqref{eq:1-proj} and \eqref{eq:2-proj}  do not commute with the decomposition in spheroidal harmonics: the right hand side of the first equation cannot be written in terms of $S^{[1]}_{m\ell}$, and the right hand side of the second equation cannot be written in terms of $S^{[2]}_{m\ell}$.  In electromagnetic-gravitational perturbations of the axially symmetric Kerr-Newman, the interaction between the spin-2 and spin-1 prevents the separability in modes. For more details see Section 111 of \cite{Chandra} and the introduction of \cite{Giorgi7}

\subsection{The physical-space combined energy-momentum tensor for the system}\label{section:combined-energy-momentum-tensor}

To solve the issue of non-commutativity with the decomposition in modes, we propose instead to perform a physical-space analysis of the system by making use of a combined energy-momentum tensor for the system. A sketch of this procedure also appeared in \cite{Giorgi7}.

As in \eqref{definition-energy-momentum-tensor}, one can define the energy-momentum tensor for a complex horizontal tensor $\psi \in \sk_k(\mathbb{C})$ as
\beaa
\QQ[\psi]_{\mu\nu}:= \Re\big(\Db_\mu  \psi \c \Db _\nu \ov{\psi}\big)
          -\frac 12 \g_{\mu\nu} \LL[\psi],
\eeaa
where $\Re$ denotes the real part and\footnote{Recall that the potentials of the gRW system are real.}
\beaa
\LL[\psi]:= \Db_\la \psi\c\Db^\la \ov{\psi} + V\psi \c \ov{\psi}.
\eeaa
The divergence of $\QQ[\psi]$ is then given by, see \cite{GKS},
 \beaa
 \D^\nu\QQ[\psi]_{\mu\nu}
  &=&\Re\Big(  \Db_\mu  \overline{\psi} \c  \left(\squared_k \psi- V\psi\right)+ \Db^\nu  \psi ^A\R_{ A   B   \nu\mu}\ov{\psi}^B \Big) -\frac 1 2 \D_\mu V\psi \c \ov{\psi}.
 \eeaa
Let $X$ be a vectorfield and $w$ a scalar. As in \eqref{definition-of-P}, one can define the associated current as
 \beaa
 \PP_\mu^{(X, w)}[\psi]&:=&\QQ_{\mu\nu} X^\nu +\frac 1 2  w \Re\big(\psi \c \Db_\mu \overline{\psi} \big)-\frac 1 4 \pr_\mu w |\psi|^2.
  \eeaa
 Then,  its divergence is given by, see \cite{GKS},
  \beaa
  \D^\mu  \PP_\mu^{(X, w)}[\psi] &=& \frac 1 2 \QQ[\psi]  \c\piX - \frac 1 2 X( V ) |\psi|^2+\frac 12  w \LL[\psi] -\frac 1 4 \square_\g  w |\psi|^2  \\
  &+&  X^\mu \Db^\nu  \psi ^a\R_{ ab   \nu\mu}\overline{\psi}^b+ \Re\Big( \big(X( \overline{\psi} )+\frac 1 2   w \overline{\psi}\big)\c \left(\squared_k \psi- V\psi\right)\Big)\\
   &=&\EE^{(X, w)}[\psi] + \Re\Big( \big(X( \overline{\psi} )+\frac 1 2   w \overline{\psi}\big)\c \left(\squared_k \psi- V\psi\right)\Big)+\R[\psi],
 \eeaa
where we used that $\R_{ab 34}=-\dual\rho \in_{ab}$ to write
      \beaa
\EE^{(X, w)}[\psi]  &:=& \frac 1 2 \QQ[\psi]  \c\piX - \frac 1 2 X( V ) |\psi|^2+\frac 12  w \LL[\psi] -\frac 1 4 \square_\g  w |\psi|^2 , \\
\R^{(X)}[\psi]&:=& -   \in_{AB}\rhod\Re\big( \big(X^4  \Db_4   \psi ^A -X^3  \Db_3   \psi ^A \big) \overline{\psi}^B\big).
\eeaa

 \medskip
 
 Applying the above to the complex tensors  $\pf \in \sk_1(\mathbb{C})$ and $\qf^\F\in \sk_2 (\mathbb{C})$, solutions to \eqref{final-eq-1} and \eqref{final-eq-2}, we obtain respectively
 \bea
   \D^\mu  \PP_\mu^{(X, w)}[\pf]    &=&\EE^{(X, w)}[\pf]+\R^{(X)}[\pf]\label{eq:div-PP1}\\
   && +\Re\Big(  \big(X( \overline{\pf} )+\frac 1 2   w \overline{\pf}\big)\c \big(i  \frac{2a\cos\th}{|q|^2}\nab_T \pf+4Q^2 \frac{\ov{q}^3 }{|q|^5} \left(  \ov{\DD} \c  \qf^\F  \right) + L_\pf[\Bfr, \Ffr]\big)\Big),\nonumber\\
      \D^\mu  \PP_\mu^{(X, w)}[\qf^\F]    &=&\EE^{(X, w)}[\qf^\F] +\R^{(X)}[\qf^\F]\label{eq:div-PP2}\\
      &&+\Re\Big(  \big(X( \overline{\qf^\F} )+\frac 1 2   w \overline{\qf^\F}\big)\c \Big(i  \frac{4a\cos\th}{|q|^2}\nab_T \qf^\F-   \frac 1 2\frac{q^3}{|q|^5} \big(  \DD \hot  \pf  -\frac 3 2 ( H - \Hb)  \hot \pf \big) + L_{\qf^\F}[\Bfr, \Ffr]  \Big)\Big).\nonumber
 \eea
As one can observe, the divergence of each equation for one of the two tensors involves the coupling terms with the other tensor of the system. Nevertheless, there exists a combined energy-momentum tensor for the system, which can be obtained by summing the two above, modulo spacetime boundary terms. In particular, we consider
\bea
\QQ[\pf, \qf^\F]_{\mu\nu}&:=& \QQ[\pf]_{\mu\nu}+8Q^2 \QQ[\qf^\F]_{\mu\nu}\\
 \PP_\mu^{(X, w)}[\pf, \qf^\F]&:=& \PP_\mu^{(X, w)}[\pf]+8 Q^2 \PP_\mu^{(X, w)}[ \qf^\F], \nonumber\\
\GG^{(X, w)}[\pf, \qf^\F]& :=& \D^\mu  \PP_\mu^{(X, w)}[\pf]+ 8Q^2\D^\mu  \PP_\mu^{(X, w)}[\qf^\F] ,\label{eq:def-GG}
\eea
 and we show that with the above combination the highest order terms in the coupling satisfy a cancellation.
 
 \subsubsection{The cancellation of the highest-order coupling terms in physical-space}\label{section:cancellation}
 To show how the cancellation of the higher order coupling terms takes place at the level of the energy-momentum tensor, as opposed to non-commutativity of the decomposition in spheroidal harmonics discussed in Section \ref{section:spheroidal}, we concentrate on the coupling terms in the divergence above. More precisely we look at
 \beaa
 E:=\frac{q^3 }{|q|^5} X( \pf ) \c \big(  \DD \c  \overline{\qf^\F}  \big)-X(\overline{\qf^\F} )\c \frac{q^3}{|q|^5} \big(  \DD \hot  \pf  -\frac 3 2 ( H - \Hb)  \hot \pf \big).
 \eeaa
 Using Lemma \ref{lemma:adjoint-operators}, we write
     \bea
 ( \DD \hot   \pf) \c  X( \ov{\qf^\F} ) &=&  -\pf \c (\DD \c X(\ov{\qf^\F})) -( (H+\Hb ) \hot \pf )\c X(\ov{\qf^\F} )+\D_\a (\pf \c X(\ov{\qf^\F}))^\a.
 \eea
Using \eqref{eq:derivatives-r-th} to deduce that 
\beaa
 \DD(\frac{ q^3}{|q|^5})  =\frac 1 2  \frac{ q^3}{|q|^5} (   \Hb -5  H) ,
 \eeaa
  we can write
\beaa
  \frac{ q^3}{|q|^5} ( \DD \hot \pf) \c X( \ov{\qf^\F} )  &=& -\frac{ q^3}{|q|^5}\pf \c (\DD \c X(\ov{\qf^\F})) -\frac{ q^3}{|q|^5}( (H+\Hb ) \hot \pf )\c X(\ov{\qf^\F} )- \frac 1 2  \frac{ q^3}{|q|^5}((   \Hb - 5   H)  \hot \pf) \c  X(\ov{\qf^\F} )\\
  &&+\D_\a \big(\frac{ q^3}{|q|^5}\pf \c X(\ov{\qf^\F} )\big)^\a\\
   &=& -\frac{ q^3}{|q|^5}\pf \c (\DD \c X(\ov{\qf^\F} )) +\frac 3 2\frac{ q^3}{|q|^5}(( H- \Hb ) \hot \pf )\c X(\ov{\qf^\F} ) +\D_\a \big(\frac{ q^3}{|q|^5}\pf \c X(\ov{\qf^\F} )\big)^\a.
   \eeaa
We therefore obtain the cancellation of the terms in $H$ and $\Hb$, as follows:
\beaa
 E&=&\frac{q^3 }{|q|^5} X( \pf ) \c \big(  \DD \c  \overline{\qf^\F}  \big)-X(\overline{\qf^\F} )\c \frac{q^3}{|q|^5} \big(  \DD \hot  \pf  -\frac 3 2 ( H - \Hb)  \hot \pf \big)\\
 &=&\frac{q^3 }{|q|^5} X( \pf ) \c \big(  \DD \c  \overline{\qf^\F}  \big) +\frac{ q^3}{|q|^5}\pf \c (\DD \c X(\ov{\qf^\F} )) -\frac 3 2\frac{ q^3}{|q|^5}(( H- \Hb ) \hot \pf )\c X(\ov{\qf^\F} ) -\D_\a \big(\frac{ q^3}{|q|^5}\pf \c X(\ov{\qf^\F} )\big)^\a\\
 &&+X(\overline{\qf^\F} )\c \frac{q^3}{|q|^5} \big( \frac 3 2 ( H - \Hb)  \hot \pf \big)\\
  &=&\frac{q^3 }{|q|^5} X( \pf ) \c \big(  \DD \c  \overline{\qf^\F}  \big) +\frac{ q^3}{|q|^5}\pf \c ( X(\DD \c\ov{\qf^\F} )) +\frac{ q^3}{|q|^5}\pf \c ([\DD \c, \nab_X]\ov{\qf^\F} ) -\D_\a \big(\frac{ q^3}{|q|^5}\pf \c X(\ov{\qf^\F} )\big)^\a
\eeaa
Integrating by parts in $X$ we finally obtain
\beaa
 E  &=&\frac{q^3 }{|q|^5} X( \pf ) \c \big(  \DD \c  \overline{\qf^\F}  \big)-X\big(\frac{ q^3}{|q|^5}\big)\pf \c (\DD \c\ov{\qf^\F} ) -\frac{ q^3}{|q|^5}X(\pf) \c (\DD \c\ov{\qf^\F} )+\frac{ q^3}{|q|^5}\pf \c ([\DD \c, \nab_X]\ov{\qf^\F} )\\
 && -\D_\a \big(\frac{ q^3}{|q|^5}\pf \c X(\ov{\qf^\F} )\big)^\a +\nab_X\big(\frac{ q^3}{|q|^5}\pf \c (\DD \c\ov{\qf^\F} )\big)\\
  &=&-X\big(\frac{ q^3}{|q|^5}\big)\pf \c (\DD \c\ov{\qf^\F} )+\frac{ q^3}{|q|^5}\pf \c ([\DD \c, \nab_X]\ov{\qf^\F} )+\mbox{Bdr},
\eeaa
where $\mbox{Bdr}$ denotes boundary terms. In particular, the highest order terms (those involving up two derivatives) got cancelled in the sum of the two terms. Similarly, we obtain
\beaa
F&:=& \frac{q^3 }{|q|^5} \frac 1 2 w \pf  \c \big(  \DD \c  \overline{\qf^\F}  \big)-\frac 1 2 w \overline{\qf^\F} \c \frac{q^3}{|q|^5} \big(  \DD \hot  \pf   -\frac 3 2 ( H - \Hb)  \hot \pf \big)\\
&=& \frac{q^3 }{|q|^5}  w \pf  \c \big(  \DD \c  \overline{\qf^\F}  \big)-\D_\a \big(\frac 1 2 w\frac{ q^3}{|q|^5}\pf \c \ov{\qf^\F} \big)^\a.
\eeaa
As $\Re(E)+\Re(F)$ precisely gives the coupling term of the combined current $\GG^{(X, w)}[\pf, \qf^\F] $, putting together \eqref{eq:div-PP1}, \eqref{eq:div-PP2} and \eqref{eq:def-GG} we finally obtain
\bea\label{eq:GG-intermediate}
\begin{split}
\GG^{(X, w)}[\pf, \qf^\F]& =\EE^{(X, w)}[\pf]+ 8Q^2\EE^{(X, w)}[\qf^\F] \\
   & - \frac{2a\cos\th}{|q|^2}\Im\Big(  \big(X( \overline{\pf} )+\frac 1 2   w \overline{\pf}\big)\c \nab_T \pf+ 16Q^2 \big(X( \overline{\qf^\F} )+\frac 1 2   w \overline{\qf^\F}\big)\c \nab_T \qf^\F \Big)\\
   &+\lot+\mbox{Bdr}
   \end{split}
\eea
where $\Im$ denotes the imaginary part, using that $\Re(iz)=-\Im(z)$, and we collected the lower order terms (i.e. those involving up to one derivatives of $\pf$ of $\qf^\F$) in $\lot$, which is given by
\beaa
\lot&:=& 4Q^2\Big(-X\big(\frac{ q^3}{|q|^5}\big)\pf \c (\DD \c\ov{\qf^\F} )+\frac{ q^3}{|q|^5}\pf \c ([\DD \c, \nab_X]\ov{\qf^\F} ) + \frac{q^3 }{|q|^5}  w \pf  \c \big(  \DD \c  \overline{\qf^\F}  \big)\Big) \\
   && +\Re\Big(  \big(X( \overline{\pf} )+\frac 1 2   w \overline{\pf}\big)\c  L_\pf[\Bfr, \Ffr]+ 8Q^2 \big(X( \overline{\qf^\F} )+\frac 1 2   w \overline{\qf^\F}\big)\c L_{\qf^\F}[\Bfr, \Ffr] \Big)+\R^{(X)}[\pf]+8Q^2\R^{(X)}[\qf^\F].
\eeaa
For general vectorfields, we still have terms involving two derivatives of $\pf$ and $\qf^\F$ in the second line of \eqref{eq:GG-intermediate}. Nevertheless, we now show that those get cancelled when $X=T$, used for the energy estimates in the trapping region, and can instead be absorbed by Cauchy-Schwarz, for small $|a|/M$, outside of the trapping region.

\subsubsection{The combined-energy momentum tensor associated to $\That_\chi$}

As obtained in Section \ref{section:energy-estimates}, the timelike vectorfield $\That_{\chi}=T + \chi \om_{\mathcal{H}}Z$ is used in the derivation of the energy estimates. Recall that $\That$ coincides with $T$ in the trapping region, i.e. $r \geq \frac89 r_{trap}^{RN}$.

We apply the combined energy-momentum tensor in \eqref{eq:GG-intermediate} to the case of $X=\That_{\chi}$, $w=0$ in the context of deriving the energy estimates for the gRW system.  In this case, the second line of \eqref{eq:GG-intermediate} becomes
\beaa
\Im\Big( \nab_{\That_\chi} \overline{\pf} \c \nab_T \pf+ 16Q^2 \nab_{\That_\chi} \overline{\qf^\F} \c \nab_T \qf^\F \Big)&=& \Im\Big( \nab_{T + \chi \om_{\mathcal{H}}Z} \overline{\pf} \c \nab_T \pf+ 16Q^2 \nab_{T + \chi \om_{\mathcal{H}}Z} \overline{\qf^\F} \c \nab_T \qf^\F \Big)\\
&=& \Im\Big( \nab_{T } \overline{\pf} \c \nab_T \pf+ 16Q^2 \nab_{T } \overline{\qf^\F} \c \nab_T \qf^\F \Big)\\
&&+\chi \om_{\mathcal{H}} \Im\Big( \nab_{ Z} \overline{\pf} \c \nab_T \pf+ 16Q^2 \nab_{Z} \overline{\qf^\F} \c \nab_T \qf^\F \Big)\\
&=&\chi \om_{\mathcal{H}} \Im\Big( \nab_{ Z} \overline{\pf} \c \nab_T \pf+ 16Q^2 \nab_{Z} \overline{\qf^\F} \c \nab_T \qf^\F \Big),
\eeaa
as the product involving $\nab_T$ gets cancelled. In particular, the above is only supported away from trapping, where $\chi=0$. Also, recall \eqref{eq:QQpi(That_de)}, which gives
\beaa
\EE^{(\That_\chi, 0)}[\pf]+ 8Q^2\EE^{(\That_\chi, 0)}[\qf^\F] &=&  \frac 1 2 \QQ[\pf]  \c ^{(\That_\chi)} \pi+8Q^2  \frac 1 2 \QQ[\qf^\F]  \c ^{(\That_\chi)} \pi\\
&=& \frac{\De  \om_{\mathcal{H}}(\pr_r\chi)}{|q|^2} \Re\big( \nab_Z\overline{\pf}\c \nab_r \pf+8Q^2 \nab_Z\overline{\qf^\F} \c\nab_r \qf^\F\big),
\eeaa
again only supported away from trapping, where $\pr_r \chi=0$.

From \eqref{eq:GG-intermediate}, we then obtain
\beaa
\GG^{(\That_\chi, 0)}[\pf, \qf^\F]& =& \frac{\De  \om_{\mathcal{H}}(\pr_r\chi)}{|q|^2} \Re\big( \nab_Z\overline{\pf}\c \nab_r \pf+8Q^2 \nab_Z\overline{\qf^\F} \c\nab_r \qf^\F\big)\\
&& - \frac{2a\cos\th}{|q|^2}\chi \om_{\mathcal{H}} \Im\Big( \nab_{ Z} \overline{\pf} \c \nab_T \pf+ 16Q^2 \nab_{Z} \overline{\qf^\F} \c \nab_T \qf^\F \Big)\\
   &&+ 4Q^2\frac{ q^3}{|q|^5}\pf \c ([\DD \c, \nab_{\That_\chi}]\ov{\qf^\F} )+\R^{(\That_{\chi})}[\pf]+8Q^2\R^{(\That_{\chi})}[\qf^\F]\\
   && +\Re\Big(  \nab_{\That_\chi} \overline{\pf} \c  L_\pf[\Bfr, \Ffr]+ 8Q^2 \nab_{\That_\chi} \overline{\qf^\F}\c L_{\qf^\F}[\Bfr, \Ffr] \Big)+\mbox{Bdr}
\eeaa
The first two lines of the above are supported away from the trapping region, and therefore can be bounded by Cauchy-Schwarz, and eventually absorbed for small $|a|$ by a Morawetz bulk where $\nab_r$, $\nab_Z$ and $\nab_T$ do not degenerate outside trapping.

Using that $[\nab_T, \nab]\psi=O(\frac{a}{r^4})\psi$ and, see    \cite{GKS2},
\beaa
\R^{(\That_\chi)}[\pf]&=& \Re\Big( - \rhod  \in_{AB} \Big(\That_\chi^4  \Db_4  \pf ^A \ov{\pf}^B -\That_\chi^3  \Db_3   \pf ^A \ov{\pf}^B\Big) \Big)=O(\frac{a}{r^4}) \Re( \nab_r \pf \c \ov{\pf}),
\eeaa
we see that the third line of the above can also be bounded by Cauchy-Schwarz in terms of $\nab_r$ derivatives or zero-th order terms of $\pf$ and $\qf^\F$, which are non-degenerate in the Morawetz norm.
 In particular, for very small $|a|/M$, those terms could be absorbed once combined with a Morawetz spacetime estimates. 
 
 For the analysis of the last line involving the terms $ L_\pf[\Bfr, \Ffr]$ and $L_{\qf^\F}[\Bfr, \Ffr] $ see \cite{Giorgi7}.

  The energy estimates would then have to be combined with the Morawetz estimates to obtain boundedness of the energy. In the case of Kerr-Newman, the non-separability in modes makes this procedure even more relevant, and in order to apply the Andersson-Blue's method \cite{And-Mor} described above for the scalar wave equation, we need to allow for a commutation with symmetry operators for the system.

 \subsection{Symmetry operators for the gRW system}
 
  In Section \ref{section:modified-laplacian}, we defined the modified Laplacian $\OO$ for scalar functions on Kerr-Newman spacetime. Nevertheless, the definition of $\OO$ as given in \eqref{definition-OO-lap}, i.e.
 \beaa
\OO(\psi) &=& |q|^2 \left(\lap \psi + (\eta+\etab) \c \nab \psi   \right),
\eeaa
can be applied to any tensor $\psi \in \sk_k$ where $\lap_k:= \de^{ab} \nab_a \nab_b$. When proving the corresponding of Proposition \ref{prop:KK-OO} for tensors, i.e. the commutator with the D'Alembertian operator $\squared_k$, one obtains additional lower order terms involving Riemann curvature. One can show, see \cite{GKS}, that for $\psi \in \sk_k$
\bea\label{eq-commutator-OO-q^2square}
\,[\OO, |q|^2 \squared ]\psi &=&O( a r^{-1}) \  \dk^{\leq 1} \psi.
\eea
 
 \subsubsection{The modified laplacian for the gRW system}

The operator $\OO$, even though is a conformal symmetry (up to the lower order terms above) for $\squared$, it is not a symmetry for gRW system of equations \eqref{final-eq-1} and \eqref{final-eq-2}, because of the presence of the coupling terms on the right hand side of the equations. We instead have to define a corrected pair of symmetry operators from it, and show how the commutation with such modified Carter differential operators allows to maintain the same structure of the equations. 
 
Following \cite{GKS}, we recall the following Gauss equation for horizontal structures:

\begin{proposition}[Proposition 2.52 in \cite{GKS}]\label{Gauss-equation-2-tensors} The following identity holds true for $\psi \in \sk_k$ for $k=0, 1,2$:
\bea
\lab{Gauss-eq-real-2-tensors}\lab{gauss-real-Ka}\label{Gauss-eq-real-tensors}
[ \nab_a, \nab_b] \psi &=&\Big( \frac 1 2 (\atrch\nab_3+\atrchb \nab_4) \psi +k  \, \Kh \dual \psi\Big)  \in_{ab}
 \eea
 where
 \bea\label{eq:definition-K-in}
\Kh&:=&- \frac 14  \trch \trchb-\frac 1 4 \atrch \atrchb+\frac 1 2 \chih \c \chibh-  \frac 1 4 \R_{3434}. 
\eea
      \end{proposition}
The scalar quantity $\Kh$ is denoted the \textit{modified Gauss curvature}\footnote{      In the integrable case, $\Kh=\frac{1}{r^2}$ is the Gauss curvature of the sphere orthogonal to the principal null frame. As in the non-integrable picture there are no spheres in the horizontal space to the null principal direction, this is not a real Gauss curvature.} of the horizontal structure. As a consequence of the above, we obtain, see \cite{GKS}, for $\xi \in \sk_1$, $u\in \sk_2$
\bea\label{eq:relation-lap1}
\DDd_2\DDs_2 \xi &=&-\frac12\triangle_1\xi +\frac 1 2 [\nab_1, \nab_2]\dual \xi=-\frac12\triangle_1\xi - \frac 1 2 \, \Kh  \xi+ \frac 1 4 (\atrch\nab_3+\atrchb \nab_4) \dual\xi 
\eea
and 
\bea\label{eq:relation-lap2}
\DDs_2 \DDd_2 u &=&-\frac12\triangle_2u -\frac 1 2 [ \nab_1, \nab_2] \dual u=-\frac12\triangle_2u+\Kh u  -\frac 1 4(\atrch\nab_3+\atrchb \nab_4) \dual u.
\eea

These relations are used to compute the commutator between $\OO$ and the angular operators $\ov{\DD} \c$ and $\DD \hot$.

\begin{lemma}\label{lemma:commutators-OO-DD} The following commutators hold true for $\phi \in \sk_1$ and $\psi \in \sk_2$:
\bea
 \,[  \OO , \DD\hot] \phi  &=&3 |q|^2 \Kh ( \DD\hot \phi) +|q|^2\DD(|q|^{-2})\hot \OO(\phi)+i  |q|^2(\atrch\nab_3+\atrchb \nab_4)( \DD\hot\phi )\nonumber \\
  &&+O(ar^{-1}) \dk^{\leq 1} \phi, \label{eq:commutator-OO-DDhot}\\
\,[\OO ,\ov{\DD}\c ]\psi&=&-3|q|^2\Kh (\ov{\DD}\c\psi )+|q|^2\DD (|q|^{-2}) \c  \OO(\psi) -i |q|^2 (\atrch\nab_3+\atrchb \nab_4) ( \ov{\DD}\c\psi)\nonumber\\
&&+ O(ar^{-1}) \dk^{\leq1}\psi. \label{eq:commutator-OO-DDc}
\eea
\end{lemma}
\begin{proof} 
Recall that $\DD\hot =\nab \hot +i\dual \nab \hot=-\DDs_2- i \dual \DDs_2$ and $\ov{\DD} \c =\DDd_2+ i \dual \DDd_2$, and therefore we can compute 
 the real part of the commutators, i.e. $[\OO, \DDs_2]$ and $[\OO, \DDd_2]$.
 We have using \eqref{eq:relation-lap1} and \eqref{eq:relation-lap2},
 \beaa
 \OO(\DDs_2\phi)&=&  |q|^2 (\lap_2 + (\eta+\etab) \c \nab) (\DDs_2 \phi)\\
 &=&  |q|^2 (-2\DDs_2 \DDd_2\DDs_2 \phi  +2\Kh \DDs_2 \phi  -\frac 1 2(\atrch\nab_3+\atrchb \nab_4) \dual\DDs_2 \phi + (\eta+\etab) \c \nab\DDs_2 \phi) \\
 &=&  |q|^2 \Big[-2\DDs_2 \big(-\frac12\lap_1\phi - \frac 1 2 \, \Kh  \phi+ \frac 1 4 (\atrch\nab_3+\atrchb \nab_4) \dual\phi  \big)  \\
 &&+2\Kh \DDs_2 \phi  -\frac 1 2(\atrch\nab_3+\atrchb \nab_4) \dual\DDs_2 \phi + (\eta+\etab) \c \nab\DDs_2 \phi \Big]\\
  &=&  |q|^2 \big(\DDs_2 \lap_1\phi +3 \, \Kh \DDs_2 \phi-  (\atrch\nab_3+\atrchb \nab_4) \dual\DDs_2\phi  + (\eta+\etab) \c \nab\DDs_2 \phi \big)+O(ar^{-1}) \dk^{\leq 1} \phi.
 \eeaa
 By writing $ \lap_1\phi=|q|^{-2}\OO(\phi)-(\eta+\etab) \c \nab \phi $, we obtain
 \beaa
  \OO(\DDs_2\phi)  &=&  |q|^2 \Big[\DDs_2\big(|q|^{-2}\OO(\phi)-(\eta+\etab) \c \nab \phi\big)+3 \, \Kh \DDs_2 \phi\\
  &&-  (\atrch\nab_3+\atrchb \nab_4) \dual\DDs_2\phi  + (\eta+\etab) \c \nab\DDs_2 \phi \Big]+O(ar^{-1}) \dk^{\leq 1} \phi\\
  &=& \DDs_2\OO(\phi)+3 |q|^2 \Kh \DDs_2 \phi-|q|^2\nab(|q|^{-2})\hot (\OO\phi)-  |q|^2(\atrch\nab_3+\atrchb \nab_4) \dual\DDs_2\phi  \\
  &&+O(ar^{-1}) \dk^{\leq 1} \phi.
 \eeaa
 By complexifying the above we obtain \eqref{eq:commutator-OO-DDhot}. Similarly, we compute
\beaa
\OO (\DDd_2 \psi)&=& |q|^2 (\lap_1 + (\eta+\etab) \c \nab) (\DDd_2 \psi)\\
&=& |q|^2 \big(-2\DDd_2\DDs_2 \DDd_2 \psi-\Kh (\DDd_2 \psi)+\frac 1 2 (\atrch\nab_3+\atrchb \nab_4)\dual \DDd_2 \psi+ (\eta+\etab) \c \nab \DDd_2 \psi\big)\\
&=& |q|^2 \Big[-2\DDd_2\big(-\frac12\lap_2\psi+\Kh \psi  -\frac 1 4(\atrch\nab_3+\atrchb \nab_4) \dual \psi  \big)\\
&&-\Kh (\DDd_2 \psi)+\frac 1 2 (\atrch\nab_3+\atrchb \nab_4)\dual \DDd_2 \psi+ (\eta+\etab) \c \nab \DDd_2 \psi\Big]\\
&=& |q|^2 \big(\DDd_2\lap_2\psi-3\Kh \DDd_2\psi  +(\atrch\nab_3+\atrchb \nab_4) \dual \DDd_2\psi  + (\eta+\etab) \c \nab \DDd_2 \psi\big)+ O(ar^{-1}) \dk^{\leq1}\psi.
\eeaa
By writing $ \lap_2\psi=|q|^{-2}\OO(\psi)-(\eta+\etab) \c \nab \psi $, we obtain
\beaa
\OO (\DDd_2 \psi)&=& |q|^2 \Big[\DDd_2\big(|q|^{-2}\OO(\psi)-(\eta+\etab) \c \nab \psi \big)-3\Kh \DDd_2\psi  \\
&&+(\atrch\nab_3+\atrchb \nab_4) \dual \DDd_2\psi  + (\eta+\etab) \c \nab \DDd_2 \psi\Big]+ O(ar^{-1}) \dk^{\leq1}\psi\\
&=&  \DDd_2 \OO(\psi)-3|q|^2\Kh \DDd_2\psi +|q|^2\nab (|q|^{-2}) \c  \OO(\psi) +|q|^2(\atrch\nab_3+\atrchb \nab_4) \dual \DDd_2\psi\\
&&+ O(ar^{-1}) \dk^{\leq1}\psi
\eeaa
By complexifing the above we obtain \eqref{eq:commutator-OO-DDc}.

\end{proof}

Observe the presence of the terms $3 |q|^2 \Kh  \DD\hot$ and $-3|q|^2\Kh \ov{\DD}\c$ in the above commutators. As they involve the modified Gauss curvature, these terms are not small even in Schwarzschild. 
 In order to obtain symmetry operators for the gRW system, we combine the modified Laplacian $\OO$ with lower order terms involving the modified Gauss curvature $\Kh$.
 
 \begin{proposition}\label{proposition:modified-OO-system} Let $\pf$ and $\qf^\F$ be solutions to the generalized Regge-Wheeler system in Theorem \ref{main-theorem-RW}. Then for any real number $c$, the complex tensors 
 \bea
\widehat{\pf}:=\big(  \OO +(c+3)|q|^2\Kh \big) \pf , \qquad \widehat{\qf^\F}:=\big(  \OO +c \ |q|^2\Kh \big) \qf^\F
 \eea
 satisfy the following system of wave equations:
 \bea
 \squared_1\widehat{\pf}-i  \frac{2a\cos\th}{|q|^2}\nab_T \widehat{\pf}  -V_1  \widehat{\pf}&=&4Q^2 \frac{\ov{q}^3 }{|q|^5} \left(  \ov{\DD} \c  \widehat{\qf^\F}  \right)+ N_1[\pf, \qf^\F]
\label{final-eq-1-commuted}\\
\squared_2\widehat{\qf^\F}-i  \frac{4a\cos\th}{|q|^2}\nab_T \widehat{\qf^\F} -V_2  \widehat{\qf^\F} &=&-   \frac 1 2\frac{q^3}{|q|^5} \left(  \DD \hot  \widehat{\pf} -\frac 3 2 \left( H - \Hb\right)  \hot \widehat{\pf} \right)+ N_2[\pf, \qf^\F]\label{final-eq-2-commuted}
 \eea
where the terms $N_1[\pf, \qf^\F]$ and $N_2[\pf, \qf^\F]$ are $O(|a|)$ lower order in differentiability with respect to $\pf$ and $\qf^\F$, explicitly given by 
\beaa
N_1[\pf, \qf^\F]&=& -4Q^2\Big[\frac{\ov{q}^3 }{|q|^3}i (\atrch\nab_3+\atrchb \nab_4) ( \ov{\DD}\c \qf^\F)-2 \nab(\frac{\ov{q}^3 }{|q|^3})\c \nab\left(  \ov{\DD} \c  \qf^\F  \right)- \frac{\ov{q}^3 }{|q|^3}\DD (|q|^{-2}) \c  \OO( \qf^\F)  \Big]\\
&& +i 4a    \nab(\cos\th) \c\nab(\nab_T \pf ) +|q|^{-2}\OO( |q|^2L_\pf[\Bfr, \Ffr])+(c+3)|q|^2\Kh L_\pf[\Bfr, \Ffr]\\
&&+O(ar^{-3}) \ \big( \dk^{\leq 1}\pf ,  \dk^{\leq 1} \qf^\F\big)\\
N_2[\pf, \qf^\F]&=&-   \frac 1 2  \Big[ \frac{q^3}{|q|^3}i(\atrch\nab_3+\atrchb \nab_4) (\DD\hot\pf)+2 \nab \Big(\frac{q^3}{|q|^3}\Big) \c \nab \left(  \DD \hot  \pf  \right)+\frac{q^3}{|q|^3}\DD(|q|^{-2})\hot (\OO\pf)  \Big] \\
&&+  i8a \nab(\cos\th) \c\nab(\nab_T  \qf^\F)+|q|^{-2} \OO(|q|^2L_{\qf^\F}[\Bfr, \Ffr])+ c|q|^2\Kh L_{\qf^\F}[\Bfr, \Ffr]\\
&&+O(ar^{-3}) \ \big( \dk^{\leq 1}\pf ,  \dk^{\leq 1} \qf^\F\big).
\eeaa
where $\dk=\nab_3, r\nab_4, r\nab$ denotes first order derivatives.

In particular, the higher order structure of equations \eqref{final-eq-1-commuted} and \eqref{final-eq-2-commuted} is the same as the gRW system of equations \eqref{final-eq-1} and \eqref{final-eq-2} for the un-commuted $\pf$ and $\qf^\F$.
 \end{proposition}
 \begin{proof} 
 We start by commuting the gRW system with the operator $\OO=|q|^2 \big( \triangle + (\eta+\etab) \c )$, where $\lap$ denotes $\lap_1$ or $\lap_2$ if it applies to $\pf$ or $\qf^\F$ respectively.
 By multiplying equations \eqref{final-eq-1} and \eqref{final-eq-2} by $|q|^2$ we obtain
\bea
|q|^2 \squared_1\pf-i  2a\cos\th\nab_T \pf  -(|q|^2V_1)  \pf &=&4Q^2 \frac{\ov{q}^3 }{|q|^3} \left(  \ov{\DD} \c  \qf^\F  \right) + |q|^2L_\pf[\Bfr, \Ffr] \label{eq:first-q^2}\\
|q|^2\squared_2\qf^\F-i  4a\cos\th \nab_T \qf^\F -(|q|^2V_2)  \qf^\F &=&-   \frac 1 2\frac{q^3}{|q|^3} \left(  \DD \hot  \pf  -\frac 3 2 \left( H - \Hb\right)  \hot \pf \right) + |q|^2L_{\qf^\F}[\Bfr, \Ffr].\label{eq:second-q^2}
 \eea
We now apply the operator $\OO$ and consider the left hand side of the two equations. Using that 
 \beaa
 \OO(f \psi )  &=& f\OO( \psi)+ \OO(f) \psi +2 |q|^2 \nab f \c \nab \psi 
 \eeaa
 we obtain for the first equation
\beaa
\OO(\text{LHS of }\eqref{eq:first-q^2})&=& |q|^2 \squared_1(\OO\pf)+ [\OO, |q|^2 \squared_1]\pf -i \OO( 2a\cos\th\nab_T \pf ) -\OO((|q|^2V_1)  \pf )\\
&=& |q|^2 \squared_1(\OO\pf)-i 2a\cos\th \OO( \nab_T \pf )-i 2a\OO(\cos\th) \nab_T \pf -2|q|^2i \nab( 2a\cos\th) \c\nab(\nab_T \pf )\\
&&   -(|q|^2V_1) \OO( \pf )-\OO(|q|^2V_1)  \pf -2|q|^2\nab (|q|^2V_1) \c \nab \pf + O( a r^{-1}) \  \dk^{\leq 1}\pf \\
&=& |q|^2 \squared_1(\OO\pf)-i 2a\cos\th  \nab_T\OO( \pf )-|q|^2V_1 \OO( \pf )-4a|q|^2i \nab(\cos\th) \c\nab(\nab_t \pf )\\
&&+ O( a r^{-1}) \  \dk^{\leq 1}\pf,
\eeaa
where we used that $[\OO, \nab_T]\psi= O( a r^{-3})  \dk^{\leq 1} \psi$, see  \cite{GKS2}. Similarly for the left hand side of \eqref{eq:second-q^2}.

 We now consider the right hand side of the equations. We make use of Lemma \ref{lemma:commutators-OO-DD} to obtain for the first equation:
 \beaa
 \OO(\text{RHS of }\eqref{eq:first-q^2}) &=&4Q^2 \frac{\ov{q}^3 }{|q|^3} \big( \ov{\DD}\c \OO( \qf^\F)-3|q|^2\Kh (\ov{\DD}\c \qf^\F)\big)  -4Q^2 \Big[ \frac{\ov{q}^3 }{|q|}i (\atrch\nab_3+\atrchb \nab_4) ( \ov{\DD}\c \qf^\F)\\
 &&-2|q|^2 \nab(\frac{\ov{q}^3 }{|q|^3})\c \nab\left(  \ov{\DD} \c  \qf^\F  \right)- \frac{\ov{q}^3 }{|q|}\DD (|q|^{-2}) \c  \OO( \qf^\F) \Big]+O(ar^{-1}) \  \dk^{\leq 1} \qf^\F+\OO( |q|^2L_\pf[\Bfr, \Ffr]),
 \eeaa
 and for the second equation
 \beaa
\OO(\text{RHS of }\eqref{eq:second-q^2})&=& -   \frac 1 2\frac{q^3}{|q|^3} \Big(  \DD\hot\OO(\pf) -\frac 3 2 \left( H - \Hb\right)  \hot  (\OO\pf) +3 |q|^2 \Kh (\DD\hot \pf)\Big)\\
&& -   \frac 1 2 \Big[ \frac{q^3}{|q|}i(\atrch\nab_3+\atrchb \nab_4) (\DD\hot\pf)+2 |q|^2\nab \Big(\frac{q^3}{|q|^3}\Big) \c \nab \left(  \DD \hot  \pf  \right)+\frac{q^3}{|q|}\DD(|q|^{-2})\hot (\OO\pf) \Big]  \\
&&+O(ar^{-1}) \  \dk^{\leq 1}\pf+ \OO(|q|^2L_{\qf^\F}[\Bfr, \Ffr]).
\eeaa
By combining the above computations, we obtain for $\OO\pf$ and $\OO\qf^\F$ respectively:
\beaa
\begin{split}
 \squared_1(\OO\pf)-i \frac{2a\cos\th}{|q|^2}  \nab_T(\OO \pf )-V_1( \OO \pf )&=4Q^2 \frac{\ov{q}^3 }{|q|^5} \big( \ov{\DD}\c (\OO \qf^\F)-3|q|^2\Kh (\ov{\DD}\c \qf^\F)\big) +\tilde{N}_1[\pf, \qf^\F]\\
 &+O(ar^{-3}) \ \big( \dk^{\leq 1}\pf ,  \dk^{\leq 1} \qf^\F\big)+|q|^{-2}\OO( |q|^2L_\pf[\Bfr, \Ffr])\\
\squared_1(\OO \qf^\F)-i \frac{4a\cos\th}{|q|^2}  \nab_T\OO(  \qf^\F)-V_2 \OO(  \qf^\F )&=  -   \frac 1 2\frac{q^3}{|q|^5} \Big(  \DD\hot (\OO\pf) -\frac 3 2 \left( H - \Hb\right)  \hot  (\OO\pf) +3 |q|^2 \Kh (\DD\hot \pf)\Big)\\
&+ \tilde{N}_2[\pf, \qf^\F] +O(ar^{-3}) \ \big( \dk^{\leq 1}\pf ,  \dk^{\leq 1} \qf^\F\big)+|q|^{-2} \OO(|q|^2L_{\qf^\F}[\Bfr, \Ffr])
\end{split}
  \eeaa
where 
\beaa
\tilde{N}_1[\pf, \qf^\F]&=&  -4Q^2\Big[\frac{\ov{q}^3 }{|q|^3}i (\atrch\nab_3+\atrchb \nab_4) ( \ov{\DD}\c \qf^\F)-2 \nab(\frac{\ov{q}^3 }{|q|^3})\c \nab\left(  \ov{\DD} \c  \qf^\F  \right)\\
&&- \frac{\ov{q}^3 }{|q|^3}\DD (|q|^{-2}) \c  \OO( \qf^\F)  \Big] +4a   i \nab(\cos\th) \c\nab(\nab_T \pf ), \\
\tilde{N}_2[\pf, \qf^\F]&=&-   \frac 1 2  \Big[ \frac{q^3}{|q|^3}i(\atrch\nab_3+\atrchb \nab_4) (\DD\hot\pf)+2 \nab \Big(\frac{q^3}{|q|^3}\Big) \c \nab \left(  \DD \hot  \pf  \right)\\
&&+\frac{q^3}{|q|^3}\DD(|q|^{-2})\hot (\OO\pf)  \Big] +  8a i \nab(\cos\th) \c\nab(\nab_T  \qf^\F).
\eeaa
Observe the terms involving $\Kh$ on the right hand side of both equations. In order to absorb them, we combine the above with the gRW equations commuted with $|q|^2 \Kh=1+O(a^2r^{-2})$, i.e. 
{\small{
\beaa
\squared_1(|q|^2\Kh\pf)-i\frac{  2a\cos\th}{|q|^2}\nab_T(|q|^2\Kh \pf)  -V_1 |q|^2\Kh \pf &=&4Q^2 \frac{\ov{q}^3 }{|q|^5} \left(|q|^2\Kh  \ov{\DD} \c ( \qf^\F ) \right) + |q|^2\Kh L_\pf[\Bfr, \Ffr]\\
&&+ O( a^2 r^{-5})    \ \big( \dk^{\leq 1}\pf ,  \dk^{\leq 1} \qf^\F\big) \\
\squared_2(|q|^2\Kh\qf^\F)-i  \frac{4a\cos\th}{|q|^2} \nab_T(|q|^2\Kh \qf^\F) -V_2 |q|^2\Kh \qf^\F &=&-   \frac 1 2\frac{q^3}{|q|^5} \left( |q|^2\Kh \DD \hot ( \pf ) -\frac 3 2 \left( H - \Hb\right)  \hot (|q|^2\Kh\pf )\right) \\
&&+ |q|^2\Kh L_{\qf^\F}[\Bfr, \Ffr]+ O( a^2 r^{-5})    \ \big( \dk^{\leq 1}\pf ,  \dk^{\leq 1} \qf^\F\big) .
 \eeaa}}
We therefore obtain, for $\widehat{\pf}:=\big(  \OO +(c+3)|q|^2\Kh \big) \pf$ and $\widehat{\qf^\F}:=\big(  \OO +(c)|q|^2\Kh \big) \qf^\F$, combining the above:
{\small{
\beaa
 \squared_1\widehat{\pf}-i  \frac{2a\cos\th}{|q|^2}\nab_T \widehat{\pf}  -V_1  \widehat{\pf}&=&4Q^2 \frac{\ov{q}^3 }{|q|^5} \big( \ov{\DD}\c (\OO \qf^\F+c|q|^2\Kh  \qf^\F)\big) +\tilde{N}_1[\pf, \qf^\F]+O(ar^{-3}) \ \big( \dk^{\leq 1}\pf ,  \dk^{\leq 1} \qf^\F\big)\\
 &&+|q|^{-2}\OO( |q|^2L_\pf[\Bfr, \Ffr])+(c+3)|q|^2\Kh L_\pf[\Bfr, \Ffr]\\
 \squared_2\widehat{\qf^\F}-i  \frac{4a\cos\th}{|q|^2}\nab_T \widehat{\qf^\F} -V_2  \widehat{\qf^\F} &=&  -   \frac 1 2\frac{q^3}{|q|^5} \Big(  \DD\hot \big(\OO\pf+(c+3)|q|^2\Kh\pf\big) -\frac 3 2 \left( H - \Hb\right)  \hot \big(\OO\pf+(c+3)|q|^2\Kh\pf\big) \Big)\\
&&+ \tilde{N}_2[\pf, \qf^\F] +O(ar^{-3}) \ \big( \dk^{\leq 1}\pf ,  \dk^{\leq 1} \qf^\F\big)\\
&&+|q|^{-2} \OO(|q|^2L_{\qf^\F}[\Bfr, \Ffr])+ c|q|^2\Kh L_{\qf^\F}[\Bfr, \Ffr]
\eeaa}}
which proves the proposition.
 \end{proof}
 
   \subsubsection{The other symmetry operators for the gRW system}
  
   The pairs of tensors obtained in Proposition \ref{proposition:modified-OO-system}, i.e.
 \beaa
\big(\widehat{\pf}, \widehat{\qf^\F}\big) =\big(  \OO\pf +(c+3)|q|^2\Kh\pf  ,  \ \ \OO \qf^\F +(c)|q|^2\Kh \qf^\F \big)
 \eeaa
 represents the symmetry operators associated to the Carter operator $\OO$ for the system. In addition to them, one can define the second order operator associated to $T$ and $Z$. 
 
 To maintain clear the difference between the operators applied to the 1-tensor $\pf$ and to the 2-tensor $\qf^\F$, we define the following couples of operators $(\SS, \WW)_{\aund}=(\SS_{\aund}, \WW_{\aund})$ for $\aund=1,2, 3, 4$:
\begin{equation}
  \begin{split}
    \SS_1 &= \nab_T \nab_T\\
    \SS_2 &=a \nab_T \nab_Z\\
    \SS_3&= a^2 \nab_Z \nab_Z\\
    \SS_4&=\OO +(c+3)|q|^2\Kh
  \end{split}
\quad\quad\quad
  \begin{split}
   \WW_1 &= \nab_T \nab_T\\
    \WW_2 &=a \nab_T \nab_Z\\
    \WW_3&= a^2 \nab_Z \nab_Z\\
    \WW_4&=\OO +c\ |q|^2\Kh
  \end{split}
\end{equation}
where $c$ is any number, the $\SS$ operators are applied to the $\pf$ and the $\WW$ operators are applied to $\qf^\F$.

 We also define the corresponding symmetric tensors\footnote{Observe the difference in the presence of $a$ and $a^2$ in the definition of $S_2$ and $S_3$ respectively. This is important in the case of tensors in order to obtain $O(a)$ lower order terms in the curvature.}
\bea
S_1^{\a\b}=T^\a T^\b, \qquad S_2^{\a\b} =aT^{(\a}Z^{\b)}, \qquad S_3^{\a\b}=a^2Z^\a Z^\b, \qquad S_4^{\a\b}=O^{\a\b}.
\eea

With the above definition, from \eqref{definition-RR-tensor}, one can write 
\bea\label{eq:RR-ab-RR-aund-new}
\RR^{\a\b}=  -(r^2+a^2)^2 S_1^{\a\b}-2(r^2+a^2)S_2^{\a\b}-  S_3^{\a\b}+ \Delta O^{\a\b}=:  \RR^\aund S_\aund^{\a\b}, 
\eea
with 
\bea
\RR^1= -(r^2+a^2)^2, \qquad \RR^2=-2(r^2+a^2), \qquad \RR^3=-1, \qquad \RR^4=\Delta.
\eea

We can then relate the symmetry operators $(\SS_\aund, \WW_\aund)$ to the symmetric tensors $S^{\a\b}$.

\begin{lemma}\label{lemma:SS-Db-system} We have for $\aund=1,2,3,4$:
\bea
\SS_\aund&=&|q|^2\Db_\a(|q|^{-2}S_\aund^{\a\b}  \Db_\b)+\de_{\aund4}(c+3)|q|^2\Kh , \\
 \WW_\aund&=&|q|^2\Db_\a(|q|^{-2}S_\aund^{\a\b}  \Db_\b)+ \de_{\aund4} c |q|^2 \Kh
\eea
where $\de_{\aund4}=0$ for $\aund=1,2,3$, and $\de_{44}=1$.
\end{lemma}
\begin{proof} For $\SS_\aund$, $\ZZ_\aund$ with $\aund=1,2,3$ it is proved as in Lemma \ref{lemma:D-S-SS}. Using that $|q|^2\Db_\a(|q|^{-2}O^{\a\b}  \Db_\b)=\OO$, we obtain the expression for $\SS_4$ and $\WW_4$ from their definition. 
\end{proof}

As in the case of $\OO$ in \eqref{eq-commutator-OO-q^2square}, in the case of tensors, the commutation between the operators $\SS_\aund$ and $\squared_k$ gives rise to terms involving the curvature, which are all $O(a)$ lower order terms, more precisely, see \cite{GKS2}:
  \beaa
\, [\SS_1, \squared_2]\psi&=&  O(a m r^{-4}) \dual (\nab_{(T}\nab_{1)})  \psi+O(a m r^{-4}) \dual (\nab_{(r}\nab_{T)}) \psi+O(a m r^{-5}) \dk^{\leq 1} \psi\\
\, [\SS_2, \squared_2]\psi&=&  O(am r^{-2}) \dual (\nab_{(T}\nab_{1)}) \psi+O(a^3 r^{-4}) \dual (\nab_{(T}\nab_{r)} ) \psi\\
&&+O(a^2 m r^{-4}) \dual (\nab_{(Z}\nab_{1)})  \psi+O(a^2 m r^{-4}) \dual (\nab_{(Z}\nab_{r)}) \psi+O(amr^{-3})\dk^{\leq 1} \psi.\\
\, [\SS_3, \squared_2]\psi&=&O(a^2m r^{-2}) \dual (\nab_Z\nab_1+\nab_1\nab_Z) \psi+O(a^4 r^{-4}) \dual(\nab_{(Z}\nab_{r)}) \psi+O(a^2mr^{-3})\dk^{\leq 1} \psi. 
\eeaa

As in \eqref{eq:def-psia}, we define
\bea
\pfa:= \SS_{\aund}(\pf), \qquad \qfa=\WW_{\aund}(\qf^\F), \qquad \aund=1,2, 3, 4.
\eea
Then, from Proposition \ref{proposition:modified-OO-system} and the above commutators, the couple of tensors $(\pfa, \qfa)$ for $\aund=1,2,3,4$ satisfies the gRW system, with possibly corrections in $O(a)$ lower order terms from the above commutators.

 \subsection{The positivity of the principal terms in the Morawetz estimates for the commuted system}
 
 We now describe how to apply the Andersson-Blue method to obtain the Morawetz estimates for the gRW system of equations for tensors. For its application to the gRW equation in Kerr see \cite{GKS2}.

From the combined energy-momentum tensor for the system given in Section \ref{section:combined-energy-momentum-tensor}, i.e. 
 \beaa
\QQ[\pf, \qf^\F]_{\mu\nu}&:=& \QQ[\pf]_{\mu\nu}+8Q^2 \QQ[\qf^\F]_{\mu\nu},
\eeaa
we define, as for the scalar wave equation in Section \ref{section:generalized-vectorfield-method}, the generalized energy-momentum tensor for the gRW system:
\bea
\QQ_{\aund\bund\mu\nu}[\pf, \qf^\F]&:=& \frac 1 4 \big(\QQ_{\mu\nu}[\pfa+\pfb, \qfa+\qfb]-\QQ_{\mu\nu}[\pfa-\pfb, \qfa-\qfb]\big)
\eea
and the generalized current:
\beaa
 \PP_\mu^{(\mathbf{X}, \mathbf{w})}[\pf, \qf^\F]&:=&\QQ_{\aund \bund \mu\nu}[\pf, \qf^\F] X^{\aund\bund\nu} +\frac 1 2  w^{\aund\bund} \big( \pfa \c \Db_\mu \pfb +8 \qfa \c \Db_\mu \qfb  \big) -\frac 1 4(\Db_\mu w^{\aund \bund} )\big( \pfa  \c \pfb +8\qfa  \c \qfb\big)\\
&=& \PP_\mu^{(\mathbf{X}, \mathbf{w})}[\pf]+8 Q^2 \PP_\mu^{(\mathbf{X}, \mathbf{w})}[ \qf^\F].
\eeaa
Its divergence \eqref{eq:def-GG} is given by, see \eqref{eq:GG-intermediate},
\beaa
\GG^{(\mathbf{X}, \mathbf{w})}[\pf, \qf^\F] &=& \D^\mu  \PP_\mu^{(\mathbf{X}, \mathbf{w})}[\pf]+ 8Q^2\D^\mu  \PP_\mu^{(\mathbf{X}, \mathbf{w})}[\qf^\F]\\
&=& \EE^{(\mathbf{X}, \mathbf{w})}[\pf]+ 8Q^2\EE^{(\mathbf{X}, \mathbf{w})}[\qf^\F] \\
   && - \frac{2a\cos\th}{|q|^2}\Im\Big(  \big(X^{\aund\bund}( \overline{\pfa} )+\frac 1 2   w \overline{\pfa}\big)\c \nab_T \pfb+ 16Q^2 \big(X^{\aund\bund}( \overline{\qfa} )+\frac 1 2   w \overline{\qfa}\big)\c \nab_T \qfb \Big)+\lot+\mbox{Bdr}
\eeaa
where
\beaa
  \EE^{(\mathbf{X}, \mathbf{w})}[\pf]&=& \frac 1 2 \QQ[\pf]_{\aund\bund}  \c \D_{(\mu} X^{\aund\bund}_{\nu)}- \frac 1 2 X^{\aund\bund}( V_1 ) \pfa \c \pfb -\frac 1 4 \square_\g w^{\aund\bund} \pfa \c \pfb+\frac 12  w^{\aund\bund} \LL_1[\pfa, \pfb],\\
 \EE^{(\mathbf{X}, \mathbf{w})}[\qf^\F]&=& \frac 1 2 \QQ[\qf^\F]_{\aund\bund}  \c \D_{(\mu} X^{\aund\bund}_{\nu)}- \frac 1 2 X^{\aund\bund}( V_2 ) \qfa \c \qfb -\frac 1 4 \square_\g w^{\aund\bund} \qfa \c \qfb+\frac 12  w^{\aund\bund} \LL_1[\qfa, \qfb]
\eeaa
with
\beaa
      \LL_1[\pfa, \pfb]&=&\g^{\a\b} \Db_\a \pfa \c \Db_\b  \pfb+ V_1\pfa \c \pfb\\
        \LL_2[\qfa, \qfb]&=&\g^{\a\b} \Db_\a \qfa \c \Db_\b  \qfb+ V_2\qfa \c \qfb.
\eeaa

Through a similar derivation, we obtain the corresponding of Proposition \ref{proposition:Morawetz1}.

  \begin{proposition}   
       Let $z$ be a given function of $r$.   Let  $u^{\aund\bund}$  be a given double-indexed function of $r$. Then for
\bea
X^{\aund\bund} =\FF^{\aund\bund} \pr_r, \qquad \quad \FF^{\aund\bund}= z u^{\aund\bund}, \qquad \quad w^{\aund\bund}=z \partial_r u^{\aund\bund} ,
\eea
we obtain
\beaa
  \EE^{(\mathbf{X}, \mathbf{w})}[\pf]&=&\AA^{\aund\bund} \, \nab_r\pfa \nab_r\pfb + \UU^{\a\b\aund\bund} \, \Db_\a \pfa \, \D_\b \pfb  +\VV_1^{\aund\bund} \, \pfa \c \pfb,\\
 \EE^{(\mathbf{X}, \mathbf{w})}[\qf^\F]&=&\AA^{\aund\bund} \, \nab_r\qfa \nab_r\qfb + \UU^{\a\b\aund\bund} \, \Db_\a \qfa \, \D_\b \qfb  +\VV_1^{\aund\bund} \, \qfa \c \qfb
\eeaa
   where
\bea
  \AA^{\aund\bund}&=&z^{1/2}\Delta^{3/2} \partial_r\left( \frac{ z^{1/2}  u^{\aund\bund} }{\Delta^{1/2}}  \right) ,
 \\
  \UU^{\a\b\aund\bund}&=& -  \frac{ 1}{2}  u^{\aund\bund} \pr_r\left( \frac z \De\RR^{\a\b}\right),\\
   \VV_{1,2}^{\aund\bund}&=&-\frac 1 4 \pr_r \Big(\De \pr_r w^{\aund\bund} \Big)-\frac 1 2  \Big( X^{\aund\bund}\big(|q|^2 V_{1,2}\big)    + |q|^2  w^{\aund\bund}_{red} V_{1,2} \Big)\\
   &=& -\frac 1 4 \pr_r\big(\De \pr_r \big(
 z \pr_ru^{\aund\bund}  \big)  \big) -\frac 1 2  \Big( X^{\aund\bund}\big(|q|^2 V_{1,2}\big)    + |q|^2  w^{\aund\bund}_{red} V_{1,2} \Big).
\eea
where $w^{\aund\bund}_{int}=u^{\aund\bund} \pr_r z $.

        \end{proposition}

Recall from Section \ref{section:Morawetz-estimates} that the crucial step in deriving Morawetz estimates for the commuted system was to perform an integration by parts  in the principal term $\UU^{\a\b\aund\bund} \, \pr_\a \psia \, \pr_\b \psib$, which allows to create a positive definite term for a trapped combination of $\psia$, denoted $\Psi$. We now show how this property can be extended to the gRW system and its combined energy-momentum tensor. 

We write the principal terms above as 
\beaa
\UU^{\a\b\aund\bund} \, \Db_\a \pfa \, \D_\b \pfb&=&  -  \frac{ 1}{2}  u^{\aund\bund} \RRtp'^{\a\b} \Db_\a \pfa \, \D_\b \pfb= -  \frac{ 1}{2}  u^{\aund\bund}\RRtp'^{\cund} S_{\cund}^{\a\b} \Db_\a \pfa \, \D_\b \pfb,\\
\UU^{\a\b\aund\bund} \, \Db_\a \qfa \, \D_\b \qfb&=&  -  \frac{ 1}{2}  u^{\aund\bund} \RRtp'^{\a\b} \Db_\a \qfa \, \D_\b \qfb= -  \frac{ 1}{2}  u^{\aund\bund}\RRtp'^{\cund} S_{\cund}^{\a\b} \Db_\a \qfa \, \D_\b \qfb.
\eeaa
By performing the integration by parts in $\Db_\a$, we obtain
  \beaa
  |q|^{-2} \UU^{\a\b\aund\bund} \, \Db_\a \pfa \, \D_\b \pfb    &=&-  \frac{ 1}{2} \D_\a\big( |q|^{-2} u^{\aund\bund} \RRtp'^{\cund} S_{\cund}^{\a\b} \pfa \, \Db_\b \pfb\big)+ \frac{ 1}{2}  u^{\aund\bund} \RRtp'^{\cund}  \pfa \, \Db_\a\big( |q|^{-2}S_{\cund}^{\a\b} \Db_\b \pfb\big).
  \eeaa
Now consider $\Db_\a\big( |q|^{-2}S_{\cund}^{\a\b} \Db_\b \pfb\big)$. Recall Lemma \ref{lemma:SS-Db-system}, i.e.
\beaa
\SS_\aund&=&|q|^2\Db_\a(|q|^{-2}S_\aund^{\a\b}  \Db_\b)+\de_{\aund4}(c+3)|q|^2\Kh
\eeaa
Therefore we can write
\beaa
 \Db_\a(|q|^{-2} \Sc^{\a\b}   \, \Db_\b \pfb)&=& \Db_\a(|q|^{-2} \Sc^{\a\b}   \, \Db_\b \SS_\bund\pf)\\
 &=&|q|^{-2}\SS_\cund \SS_\bund\pf- |q|^{-2}(\de_{\cund4}(c+3)|q|^2\Kh\SS_\bund \pf)\\
 &=&|q|^{-2}\SS_\bund \SS_\cund\pf+|q|^{-2}[\SS_{\cund}, \SS_{\bund}]\pf- |q|^{-2}(\de_{\cund4}(c+3)|q|^2\Kh\SS_\bund \pf)\\
 &=& \Db_\a(|q|^{-2} \Sb^{\a\b}   \, \Db_\b \pfc)+|q|^{-2}[\SS_{\cund}, \SS_{\bund}]\psi\\
 &&+ |q|^{-2}(\de_{\bund4}(c+3)|q|^2\Kh\SS_\cund \pf)- |q|^{-2}(\de_{\cund4}(c+3)|q|^2\Kh\SS_\bund \pf).
\eeaa
Thus, repeating the integration by parts procedure, and for  $u^{\aund\bund}=-h \RRtp'^{\aund} \LL^{\bund}$ and recalling that $  \LL^{\bund}   \Sb^{\a\b}=L^{\a\b}  $, we obtain
\beaa
  \UU^{\a\b\aund\bund} \, \Db_\a \pfa \, \D_\b \pfb&=&\frac 1 2  h  L^{\a\b}  \Db_\a (\RRtp'^{\aund} \pfa )    \Db_\b(\RRtp'^{\cund} \pfc) -\frac 1 2  h \RRtp'^{\cund}  \LL^{\bund}(\RRtp^{\aund}  \pfa )\c [\SS_{\cund}, \SS_{\bund}]\pf\\
&& +|q|^2\Db_\a \left(|q|^{-2} \frac 1 2  h  \RRtp'^{\cund}  \LL^{\bund} (\RRtp^{\aund} \pfa )\left(\Sc^{\a\b}   \, \Db_\b \pfb-  \Sb^{\a\b}   \Db_\b \pfc \right) \right ) \\
&&- \frac{ (c+3)}{2} h|q|^2\Kh (\RRtp'^{\aund} \pfa) \LL^{\bund} \RRtp'^{\cund}  \, \big( \de_{\bund4} \pfc-\de_{\cund4} \pfb \big), 
\eeaa
and similarly for $\qf^\F$.
By denoting $\Phi:=\RRtp'^{\aund} \pfa$ and $\Psi:=\RRtp'^{\aund} \qfa$, we respectively obtain
\bea
  \UU^{\a\b\aund\bund} \, \Db_\a \pfa \, \D_\b \pfb&=&\frac 1 2  h  L^{\a\b}  \Db_\a \Phi \Db_\b\Phi  -\frac 1 2  h\Phi \c \RRtp'^{\cund}  \LL^{\bund}[\SS_{\cund}, \SS_{\bund}]\pf -|q|^2\D_\a \BB^\a[\pf]\nonumber \\
&&- \frac{ (c+3)}{2}\big( h|q|^2\Kh\big) \ \Phi \c  \LL^{\bund} \RRtp'^{\cund}  \, \big( \de_{\bund4} \pfc-\de_{\cund4} \pfb \big), \label{eq:principal-term-p=}\\
  \UU^{\a\b\aund\bund} \, \Db_\a \qfa \, \D_\b \qfb&=&\frac 1 2  h  L^{\a\b}  \Db_\a \Psi \Db_\b\Psi  -\frac 1 2  h\Psi \c  \RRtp'^{\cund}  \LL^{\bund}[\WW_{\cund}, \WW_{\bund}]\qf^\F -|q|^2\D_\a \BB^\a[\qf^\F]\nonumber \\
&&- \frac{ c}{2} \big( h|q|^2\Kh\big) \ \Psi \c \LL^{\bund} \RRtp'^{\cund}  \, \big( \de_{\bund4} \qfc-\de_{\cund4} \qfb \big).\label{eq:principal-term-q=}
\eea
For the second line in each of the above expressions, we write 
\beaa
 \LL^{\bund} \RRtp'^{\cund}  \, \big( \de_{\bund4} \pfc-\de_{\cund4} \pfb \big)&=&     \LL^{\bund} \de_{\bund4} (\RRtp'^{\cund}\pfc)- \RRtp'^{\cund}\de_{\cund4} (\LL^{\bund}\pfb) =   \LL^{4}  \Phi- \RRtp'^{4}( \LL^{\bund}\pfb)\\
 \LL^{\bund} \RRtp'^{\cund}  \, \big( \de_{\bund4} \qfc-\de_{\cund4} \qfb \big)&=& \LL^{\bund}   \de_{\bund4}(\RRtp'^{\cund}  \qfc)- \RRtp'^{\cund} \de_{\cund4}(\LL^{\bund} \qfb)=\LL^{4}  \Psi- \RRtp'^{4}( \LL^{\bund}\qfb).
\eeaa

Now observe that with the choice $z=z_1$ as in \eqref{eq:definition-z1}, we have, see \eqref{eq:values-RRtp'z1},  $\RRtp'^4[z_1] =   - \frac{2\TT}{ (r^2+a^2)^3}\big( 1+O(\ep r^{-2})\big)$, and with the choice of $\LL^\bund$ as in \eqref{eq:choice-L}, we write, using \eqref{eq:Psi-choice-z}:
\beaa
 \LL^{\bund} \RRtp'^{\cund}  \, \big( \de_{\bund4} \pfc-\de_{\cund4} \pfb \big)&=&\LL^{4}  \Phi- \RRtp'^{4}( \LL^{\bund}\pfb)\\
 &=& \Phi - \frac{2\TT}{ (r^2+a^2)^3}\big( 1+O(\ep r^{-2})\big)(\ep \SS_1\psi +\OO\psi)\\
&=&\Phi - \big( \Phi -   \frac{4ar}{(r^2+a^2)^2} \big(1+O(\ep r^{-2}) \big)  \nab_{\That} \nab_Z\pf\big)=  \frac{4ar}{(r^2+a^2)^2} \big(1+O(\ep r^{-2}) \big)  \nab_{\That} \nab_Z\pf, \\
 \LL^{\bund} \RRtp'^{\cund}  \, \big( \de_{\bund4} \qfc-\de_{\cund4} \qfb \big)&=&    \frac{4ar}{(r^2+a^2)^2} \big(1+O(\ep r^{-2}) \big)  \nab_{\That} \nab_Z\qf^\F,
\eeaa
Also, using that, see \cite{GKS2},
\bea\label{eq:commutator-SS1-SS2-SS3}
\, [\SS_1, \SS_2]\psi=[\SS_1, \SS_3]\psi=[\SS_2, \SS_3]\psi=0
\eea
and 
\beaa
\, [\SS_4, \SS_1]\psi&=& O(ma r^{-2} )\nab_1 \nab_T \psi + O(ma r^{-4}) \dk^{\leq 1}\psi \\
\, [\SS_4, \SS_2]\psi&=& O(m a )\nab_1\nab_T \psi + O(m a^2 r^{-2}) \nab_1 \nab_Z \psi + O(m a^2 r^{-2}) \dk^{\leq 1} \psi \\
\, [\SS_4,\SS_3]\psi&=&  O(ma^2 ) \nab_1\nab_Z \psi+ O(ma^3 r^{-2}) \dk^{\leq 1}\psi.
\eeaa
and similarly for $\WW_{\aund}$, we can write for the commutator terms:
 \beaa
 -\frac 1 2  h\Phi \c \RRtp'^{\cund}  \LL^{\bund}[\SS_{\cund}, \SS_{\bund}]\pf &=&-\frac 1 2  h\Phi \Big( \big(\RRtp'^{1}  \LL^{4}- \RRtp'^{4}  \LL^{1}\big)[\SS_{1}, \SS_4]\pf+\RRtp'^{2}  \LL^{4}[\SS_{2}, \SS_4]\pf+ \big( \RRtp'^{3}  \LL^{4}-  \RRtp'^{4}  \LL^{3}\big) [\SS_{3}, \SS_4]\pf\Big)\\
    &=&-\frac 1 2  h\Phi \c \Big( O(a r^{-3} )\nab_1 \nab_T \pf+ O( a^2 r^{-3}) \nab_1 \nab_Z \pf + O( a^2 r^{-5}) \dk^{\leq 1} \pf \Big),\\
     -\frac 1 2  h\Psi \c \RRtp'^{\cund}  \LL^{\bund}[\SS_{\cund}, \SS_{\bund}]\qf^\F     &=&-\frac 1 2  h\Psi\c \Big( O(a r^{-3} )\nab_1 \nab_T \qf^\F+ O( a^2 r^{-3}) \nab_1 \nab_Z \qf^\F + O( a^2 r^{-5}) \dk^{\leq 1} \qf^\F \Big).
 \eeaa

We conclude from \eqref{eq:principal-term-p=} and \eqref{eq:principal-term-q=}, that the principal terms can be written as 
\beaa
  \UU^{\a\b\aund\bund} \, \Db_\a \pfa \, \D_\b \pfb&=&\frac 1 2  h  L^{\a\b}  \Db_\a \Phi \Db_\b\Phi   -|q|^2\D_\a \BB^\a[\pf]\\
&&-\frac 1 2  h\Phi \c \Big(O(ar^{-3}) \nab_{\That}\nab_Z \pf+ O(a r^{-3} )\nab_1 \nab_T \pf+ O( a^2 r^{-3}) \nab_1 \nab_Z \pf + O( a^2 r^{-5}) \dk^{\leq 1} \pf \Big), \\
  \UU^{\a\b\aund\bund} \, \Db_\a \qfa \, \D_\b \qfb&=&\frac 1 2  h  L^{\a\b}  \Db_\a \Psi \Db_\b\Psi   -|q|^2\D_\a \BB^\a[\qf^\F] \\
&&-\frac 1 2  h\Psi \c \Big(O(ar^{-3}) \nab_{\That}\nab_Z\qf^\F+ O(a r^{-3} )\nab_1 \nab_T \qf^\F+ O( a^2 r^{-3}) \nab_1 \nab_Z \qf^\F + O( a^2 r^{-5}) \dk^{\leq 1} \qf^\F \Big).
\eeaa
In both cases, the last line has the same structure as the one in Kerr, and can easily be bounded by Cauchy-Schwarz by $O(a) \big( |\Phi|^2+|\pf|^2_{\SS}+|\Psi|^2+|\qf^\F|^2_{\SS}\big)$, where all the terms appear without degeneracy in the trapping region in the Morawetz bulk. They therefore can be absorbed by the Morawetz bulk for very small angular momentum. 

The analysis of those lower order terms can be treated as in the case of Kerr, see \cite{GKS2}. From the analysis of the coupled system in Reissner-Nordstr\"om \cite{Giorgi7a}, where the quadratic form is shown to be positive for $a=0$, we deduce it to be positive for very small $|a|/M$. This ends the description of how to apply Theorem \ref{theo:Carter-operator-commutes-KN} to obtain physical-space analysis of the generalized Regge-Wheeler equations for electromagnetic-gravitational perturbations of Kerr-Newman. 

The full derivation of the Morawetz-Energy estimates for the gRW system in Kerr-Newman will follow the Morawetz-Energy estimates for the gRW equation in Kerr in \cite{GKS2} and will appear in a future work.
 
 \small
 
 \appendix
 
 \section{Proof of Proposition \ref{commutation-KK-square}}\label{proof-prop-commutators}

First of all, observe that as a consequence of the definition \eqref{Killing-equation}, contracting $\D_{\mu} K_{\nu\rho}+\D_{\nu} K_{\rho\mu}+\D_{\rho} K_{\mu\nu} = 0$ with $\g^{\mu\nu}$, one immediately derives
\bea\label{eq:divergence-K}
\D_\mu K^{\mu\nu}=-\frac 1 2 \D^\nu (\tr K),
\eea
where $\tr K=\g^{\mu\nu} K_{\mu\nu}$ is the trace of $K$.
By applying another derivative to the above and antisymmetrizing, one also obtains
\bea\label{eq:property-Killing-tensor}
\D^\mu\D_\a K^{\a\nu}-\D^\nu\D_\a K^{\a\mu}&=&-\frac 1 2 \D^\mu \D^\nu (\tr K)+\frac 1 2 \D^\nu \D^\mu (\tr K)=0.
\eea

 To obtain the commutator $[\KK, \square_\g]$, we first collect the following preliminary computations.

\begin{lemma}\label{general-commutator} The following commutation formulas hold for a scalar $\phi$, a 1-tensor $X$ and a 2-tensor $\Psi$ in $\MM$:
\beaa
\, [\D_\mu, \square_\g]\phi&=& - \R_{\mu\a} \D^\a \phi, \\
 \, [\D_\mu, \square_\g ]X^\mu &=&\R_{\mu\a} \D^\mu X^\a + \left(\D_\a \R -\D^\mu \R_{\mu\a} \right) X^\a,  \\
  \, [\D_\mu, \square_\g]{\Psi^{\mu}}_{\de}&=&2 \R_{\mu \a \de \ep} \D^\a \Psi^{\mu\ep} + \R_{\a\ep} \D^\a {\Psi^{\ep}}_{\de} + \D^\a{ \R_{\a\ep} } {\Psi^{\ep}}_{\de}+(\D_\epsilon \R_{\delta\mu}-\D_\delta \R_{\mu\epsilon})\Psi^{\mu\ep},
\eeaa
where $\R$ denotes the Riemann curvature, the Ricci curvature or the scalar curvature depending if it appears as a 4-tensor, a 2-tensor or a scalar respectively. 
\end{lemma}
\begin{proof} For a scalar $\phi$, we have by definition of Riemann curvature
\beaa
\, [\D_\a, \D_\b] \phi &=&0, \qquad  [\D_\a, \D_\b] \D_\gamma \phi = {\R_{\a\b\gamma}}^\delta \D_\delta \phi.
\eeaa
We therefore compute 
\beaa
\, [\D_\nu, \square_\g]\phi&=&  [\D_\nu, \D^\a \D_\a]\phi=[\D_\nu, \D^\a] \D_\a\phi + \D^\a[\D_\nu,  \D_\a]\phi=\g^{\a\mu}{\R_{\nu\mu \a}}^\de \D_\de\phi = - {\R_{\nu}}^\de \D_\de \phi.
\eeaa
For a 1-tensor $X$ we have
\beaa
\, [\D_\a, \D_\b] X_\gamma  &=& {\R_{\a\b\gamma\epsilon }} X^\epsilon, \qquad  [\D_\a, \D_\b] \D_\delta X_\gamma  = {\R_{\a\b\delta \epsilon }} \D^\epsilon X_\gamma+ {\R_{\a\b\gamma\epsilon }} \D_\delta X^\epsilon.
\eeaa
We then compute
\beaa
\, [\D_\mu, \square_\g ]X_\b &=&  [\D_\mu, \D^\a \D_\a]X_\b =[\D_\mu, \D^\a] \D_\a X_\b+ \D^\a ([\D_\mu, \D_\a]X_\b )\\
&=&\g^{\a\ze} ({\R_{\mu\ze\a}}^\ep \D_\ep X_\b+ {\R_{\mu\ze\b}}^\ep \D_\a  X_\ep )+ \D^\a ({\R_{\mu\a\b}}^\ep X_\ep )\\
&=& -{\R_{\mu}}^\ep \D_\ep X_\b +{{{\R_{\mu}}^\a}_{\b}}^\ep \D_\a  X_\ep + \D^\a {\R_{\mu\a\b}}^\ep X_\ep + {\R_{\mu\a\b}}^\ep \D^\a X_\ep .
\eeaa
Using the second Bianchi identity $\D^\a\R_{\a\mu\nu\gamma}=\D_\nu \R_{\mu\gamma}-\D_\gamma \R_{\mu\nu}$, we obtain
\beaa
\, [\D_\mu, \square_\g ]X_\b &=&2{\R_{\mu\a\b}}^\ep \D^\a X_\ep  -{\R_{\mu}}^\ep \D_\ep X_\b  + \left(\D^\ep \R_{\mu\b}-\D_\b {\R_\mu}^\ep \right) X_\ep. 
\eeaa
By contracting with $\g^{\mu\beta}$ we finally obtain
\beaa
\, [\D_\mu, \square_\g ]X^\mu &=&2{\R_{\a}}^\ep \D^\a X_\ep  -{\R_{\mu}}^\ep \D_\ep X^\mu  + \left(\D^\ep {\R_{\mu}}^{\mu}-\D^\mu {\R_\mu}^\ep \right) X_\ep \\
&=&{\R_{\mu}}^\ep \D^\mu X_\ep + \left(\D^\ep \R -\D^\mu {\R_\mu}^\ep \right) X_\ep .
\eeaa
For a 2-tensor $\Psi$ we have
\beaa
\,[ \D_\nu , \D_\a ] \Psi_{\ga\de} &=&{ \R_{\nu\a \ga} }^\ep \Psi_{\ep\de}+{ \R_{\nu\a\de }}^\ep\Psi_{\gamma\ep} \\
\, [\D_\nu, \D_\ze ] \D_\a\Psi_{\ga\de}&=& { \R_{\nu\ze \a} }^\ep \D_\ep \Psi_{\ga\de}+{ \R_{\nu\ze \gamma} }^\ep \D_\a \Psi_{\ep\de}+{ \R_{\nu\ze \de} }^\ep \D_\a \Psi_{\gamma\ep}.
\eeaa
We then compute
\beaa
\, [\D_\nu, \square_\g]\Psi_{\ga\de}&=& \, [\D_\nu, \D^\a \D_\a]\Psi_{\ga\de}= [\D_\nu, \D^\a ] \D_\a\Psi_{\ga\de}+ \D^\a [\D_\nu,  \D_\a]\Psi_{\ga\de}\\
&=&\g^{\a\ze} \left( { \R_{\nu\ze \a} }^\ep \D_\ep \Psi_{\ga\de}+{ \R_{\nu\ze \gamma} }^\ep \D_\a \Psi_{\ep\de}+{ \R_{\nu\ze \de} }^\ep \D_\a \Psi_{\gamma\ep} \right)+ \D^\a({ \R_{\nu\a \ga} }^\ep \Psi_{\ep\de}+{ \R_{\nu\a\de }}^\ep\Psi_{\gamma\ep})\\
&=& -{ \R_{\nu} }^\ep \D_\ep \Psi_{\ga\de}+{{{ \R_{\nu}}^\a}_{ \gamma} }^\ep \D_\a \Psi_{\ep\de}+{{{ \R_{\nu}}^\a}_{ \de} }^\ep \D_\a \Psi_{\gamma\ep} \\
&&+ \D^\a{ \R_{\nu\a \ga} }^\ep \Psi_{\ep\de}+{ \R_{\nu\a \ga} }^\ep \D^\a\Psi_{\ep\de}+\D^\a{ \R_{\nu\a\de }}^\ep\Psi_{\gamma\ep}+{ \R_{\nu\a\de }}^\ep \D^\a \Psi_{\gamma\ep}\\
&=&2{{{ \R_{\nu}}^\a}_{ \gamma} }^\ep \D_\a \Psi_{\ep\de}+2{{{ \R_{\nu}}^\a}_{ \de} }^\ep \D_\a \Psi_{\gamma\ep}  -{ \R_{\nu} }^\ep \D_\ep \Psi_{\ga\de}+ \D^\a{ \R_{\nu\a \ga} }^\ep \Psi_{\ep\de}+\D^\a{ \R_{\nu\a\de }}^\ep\Psi_{\gamma\ep}
\eeaa
By contracting with $\g^{\gamma\nu}$ we finally obtain
\beaa
\, [\D_\nu, \square_\g]{\Psi^\nu}_{\de}&=&2 {\R_{\nu}}^{\a\nu\ep} \D_\a \Psi_{\ep\de}+2{{{ \R_{\nu}}^\a}_{ \de} }^\ep \D_\a {\Psi^{\nu}}_{\ep}  -{ \R_{\nu} }^\ep \D_\ep {\Psi^\nu}_{\de}+ \D^\a{ \R_{\nu\a} }^{\nu\ep} \Psi_{\ep\de}+\D^\a{ \R_{\nu\a\de }}^\ep{\Psi^\nu}_{\ep}\\
&=&2{{{ \R_{\nu}}^\a}_{ \de} }^\ep \D_\a {\Psi^{\nu}}_{\ep} + \R^{\a\ep} \D_\a \Psi_{\ep\de} + \D^\a{ \R_{\a} }^{\ep} \Psi_{\ep\de}+\D^\a{ \R_{\nu\a\de }}^\ep{\Psi^\nu}_{\ep}
\eeaa
Again using the Bianchi identity to write ${\D^\a\R_{\nu\a\delta}}^{\epsilon}=\D^\epsilon \R_{\delta\nu}-\D_\delta {\R_{\nu}}^{\epsilon}$, we prove the Lemma.
\end{proof}

We use the above Lemma to compute the commutator between the Killing differential operator $\KK$ as defined by \eqref{definition-KK} and the D'Alembertian operator $\square_\g$.

 We compute for a scalar function $\phi$:
\beaa
[\KK, \square_\g]\phi &=& [ \D_{\mu} K^{\mu\nu} \D_\nu, \square_\g]\phi=\D_{\mu} [K^{\mu\nu} \D_\nu, \square_\g]\phi+  [ \D_{\mu}, \square_\g] K^{\mu\nu} \D_\nu\phi\\
&=&\D_{\mu}( K^{\mu\nu} \D_\nu \square_\g \phi- \square_\g (K^{\mu\nu} \D_\nu \phi))+  [ \D_{\mu}, \square_\g] K^{\mu\nu} \D_\nu\phi.
\eeaa
Writing $\square_\g(K^{\mu\nu} \D_\nu \phi)= (\square_\g K^{\mu\nu}) \D_\nu \phi+2 \D^\a K^{\mu\nu} \D_\a \D_\nu \phi+ K^{\mu\nu} \square_\g \D_\nu \phi $, we have
\beaa
[\KK, \square_\g]\phi &=&\D_{\mu}( K^{\mu\nu}[ \D_\nu,  \square_\g] \phi- \square_\g K^{\mu\nu} \D_\nu \phi-2 \D^\a K^{\mu\nu} \D_\a \D_\nu \phi)+  [ \D_{\mu}, \square_k] K^{\mu\nu} \D_\nu\phi.
\eeaa
Applying Lemma \ref{general-commutator} to $\phi$ and to $X^\mu=K^{\mu\nu} \D_\nu \phi$, we have
\beaa
\,  K^{\mu\nu}[ \D_\nu,  \square_\g] \phi&=&- K^{\mu\nu}  {\R_{\nu}}^\ep \D_\ep \phi, \\
 \, [\D_\mu, \square_\g ]K^{\mu\nu} \D_\nu \phi &=&{\R_{\mu}}^\ep \D^\mu (K_{\ep\nu} \D^\nu \phi) + \left(\D^\ep \R -\D^\mu {\R_\mu}^\ep \right) K_{\ep\nu} \D^\nu \phi.
 \eeaa
We therefore obtain
\beaa
[\KK, \square_\g]\phi &=&\D_{\mu}( - \square_\g K^{\mu\nu} \D_\nu \phi-2 \D^\a K^{\mu\nu} \D_\a \D_\nu \phi)-\D_{\mu}( K^{\mu\nu}  {\R_{\nu}}^\ep \D_\ep \phi)\\
&&+{\R_{\mu}}^\ep \D^\mu (K_{\ep\nu} \D^\nu \phi) + \left(\D^\ep \R -\D^\mu {\R_\mu}^\ep \right) K_{\ep\nu} \D^\nu \phi.
\eeaa
From \eqref{Killing-equation}, i.e. $\D_{\a} K^{\mu\nu}+\D^{\mu} {K^{\nu}}_{\a}+\D^{\nu} {K^{\mu} }_{\a}= 0$, 
we have
\beaa
(\D_\a K^{\mu\nu}) \D^\a \D_\nu \phi&=&  -\D^{\mu} {K^{\nu}}_{\a} \D^\a \D_\nu \phi-\D^{\nu} {K^{\mu} }_{\a} \D^\a \D_{\nu} \phi= -\D^{\mu} {K^{\nu}}_{\a} \D^\a \D_\nu \phi-\D_{\a} K^{\mu\nu}  \D^\a \D_\nu \phi,
\eeaa
which gives
\beaa
(\D_\a K^{\mu\nu}) \D^\a \D_\nu \phi&=& - \frac 1 2 \D^{\mu} {K^{\nu}}_{\a} \D^\a \D_\nu \phi.
\eeaa
This implies
\beaa
[\KK, \square_\g]\phi &=&\D_{\mu}( - \square_\g K^{\mu\nu} \D_\nu \phi+ \D^{\mu} {K^{\nu}}_{\a} \D^\a \D_\nu \phi)-\D_{\mu}( K^{\mu\nu}  {\R_{\nu}}^\ep \D_\ep \phi)\\
&&+{\R_{\mu}}^\ep \D^\mu (K_{\ep\nu} \D^\nu \phi) + \left(\D^\ep \R -\D^\mu {\R_\mu}^\ep \right) K_{\ep\nu} \D^\nu \phi.
\eeaa
Observe that in expanding the $\D_\mu$ derivative there is a cancellation of the term $(\square_\g K^{\mu\nu})\D_\mu\D_\nu \phi$:
\bea\label{eq-intermediate-3-proof}
[\KK, \square_\g]\phi &=& - \D_{\mu}(\square_\g K^{\mu\nu}) \D_\nu \phi- \square_\g K^{\mu\nu} \D_{\mu}\D_\nu \phi+\D_{\mu} \D^{\mu} {K^{\nu}}_{\a}  \D_\nu \D^\a\phi+ \D^{\mu} {K^{\nu}}_{\a} \D_{\mu}\D^\a \D_\nu \phi \nonumber\\
&&-\D_{\mu}( K^{\mu\nu}  {\R_{\nu}}^\ep \D_\ep \phi)+{\R_{\mu}}^\ep \D^\mu (K_{\ep\nu} \D^\nu \phi) + \left(\D^\ep \R -\D^\mu {\R_\mu}^\ep \right) K_{\ep\nu} \D^\nu \phi \nonumber\\
&=& - \D_{\mu}(\square_\g K^{\mu\nu}) \D_\nu \phi+ \D^{\mu} {K^{\nu}}_{\a} \D_{\mu}\D^\a \D_\nu \phi-\D_{\mu}( K^{\mu\nu}  {\R_{\nu}}^\ep \D_\ep \phi)\\
&&+{\R_{\mu}}^\ep \D^\mu (K_{\ep\nu} \D^\nu \phi) + \left(\D^\ep \R -\D^\mu {\R_\mu}^\ep \right) K_{\ep\nu} \D^\nu \phi. \nonumber
\eea
We now derive an expression for the two terms appearing in the first line above. For the second term we have, using again the Killing equation \eqref{Killing-equation},
\beaa
(\D_\a K^{\mu\nu}) \D^\a \D_{\mu}\D_\nu \phi&=& (-\D^{\mu} {K^{\nu}}_{\a}-\D^{\nu} {K^{\mu} }_{\a}) \D^\a \D_{\mu}\D_\nu \phi= -\D^{\mu} {K^{\nu}}_{\a} \D^\a \D_{\mu}\D_\nu \phi-\D^{\nu} {K^{\mu} }_{\a} \D^\a \D_{\nu}\D_\mu \phi\\
&=&-2\D^{\mu} {K^{\nu}}_{\a} \D^\a \D_{\mu}\D_\nu \phi=-2\D^{\mu} {K^{\nu}}_{\a}( \D_\mu \D^{\a}\D_\nu \phi+[\D^\a, \D_{\mu}]\D_\nu \phi ) \\
&=&-2\D^{\mu} {K^{\nu}}_{\a}( \D_\mu \D^{\a}\D_\nu \phi+{ {\R^\a}_{\mu \nu}}^\delta \D_\delta \phi) 
\eeaa
Observe that the first term on the right hand side is the same as the term on the left hand side. We therefore obtain:
\beaa
(\D^\mu {K^{\nu}}_{\a}) \D_\mu \D^{\a}\D_\nu \phi&=&-\frac 2 3 \D^{\a} {K^{\nu}}_{\mu}{ {\R^\mu}_{\a \nu}}^\delta \D_\delta \phi.
\eeaa
For the first term, starting with $\D_{\mu} K_{\nu\rho}+\D_{\nu} K_{\rho\mu}+\D_{\rho} K_{\mu\nu}=0$
and applying $\D^\mu$ we have
\beaa
0&=&\D^\mu \D_{\mu} K_{\nu\rho}+\D^\mu \D_{\nu} K_{\rho\mu}+\D^\mu \D_{\rho} K_{\mu\nu} \\
&=&\D^\mu \D_{\mu} K_{\nu\rho}+ \D_{\nu} \D^\mu K_{\rho\mu}+ {{\R^{\mu}}_{\nu\rho}}^\ep K_{\ep \mu}+ {{\R}_{\nu}}^\ep K_{\rho \ep}+ \D_{\rho} \D^\mu K_{\mu\nu} + {{\R^{\mu}}_{\rho \nu}}^\ep K_{\mu \ep} + {{\R}_{\rho}}^\ep K_{\ep \nu}.
\eeaa
This gives
\beaa
\D^\a \D_{\a} K_{\mu\nu}&=&- \D_{\nu} \D^\a K_{\mu\a}- \D_{\mu} \D^\a K_{\a\nu} - {{\R^{\a}}_{\nu\mu}}^\ep K_{\ep \a}- {{\R^{\a}}_{\mu \nu}}^\ep K_{\a \ep}-{{\R}_{\nu}}^\ep K_{\mu \ep}-{{\R}_{\mu}}^\ep K_{\ep \nu}.
\eeaa
Using \eqref{eq:property-Killing-tensor}, we can write $\D_\mu\D^\a K_{\a\nu}+\D_\nu\D^\a K_{\a\mu}= 2\D_\mu\D^\a K_{\a\nu}$,
which gives
\beaa
\D^\a \D_{\a} K_{\mu\nu}&=&- 2\D_\mu\D^\a K_{\a\nu} - \left({{\R^{\a}}_{\nu\mu}}^\ep+ {{\R^{\a}}_{\mu \nu}}^\ep \right) K_{\a \ep}-{{\R}_{\nu}}^\ep K_{\mu \ep}-{{\R}_{\mu}}^\ep K_{\ep \nu}.
\eeaa
Applying $\D^\mu$ to the above we have
\beaa
\D^\mu\D^\a \D_{\a} K_{\mu\nu}&=&- 2\D^\mu\D_\mu\D^\a K_{\a\nu} - \left({{\R^{\a}}_{\nu\mu}}^\ep+ {{\R^{\a}}_{\mu \nu}}^\ep \right) \D^\mu K_{\a \ep}\\
&&- \left(\D^\mu{{\R^{\a}}_{\nu\mu}}^\ep+ \D^\mu{{\R^{\a}}_{\mu \nu}}^\ep \right)  K_{\a \ep}- \D^\mu \left({{\R}_{\nu}}^\ep K_{\mu \ep}+{{\R}_{\mu}}^\ep K_{\ep \nu}\right),
\eeaa
which can be written as 
\begin{equation}
    \label{eq-intermediate-proof}
\begin{split}
\D^\mu\D^\a \D_{\a} K_{\mu\nu}+2\D^\mu\D_\mu\D^\a K_{\a\nu}&= - \left({{\R^{\a}}_{\nu\mu}}^\ep+ {{\R^{\a}}_{\mu \nu}}^\ep \right) \D^\mu K_{\a \ep}\\
&- \left(\D^\mu{{\R^{\a}}_{\nu\mu}}^\ep+ \D^\mu{{\R^{\a}}_{\mu \nu}}^\ep \right)  K_{\a \ep}- \D^\mu \left({{\R}_{\nu}}^\ep K_{\mu \ep}+{{\R}_{\mu}}^\ep K_{\ep \nu}\right).
\end{split}
\end{equation}
The left hand side of \eqref{eq-intermediate-proof} is given by
\beaa
\D^\mu\D^\a \D_{\a} K_{\mu\nu}+2\D^\mu\D_\mu\D^\a K_{\a\nu}&=&3\D^\a \D_{\a}  \D^\mu K_{\mu\nu}+[\D^\mu, \D^\a \D_{\a}] K_{\mu\nu}.
\eeaa
Using  Lemma \ref{general-commutator} to write
\beaa
  \, [\D^\mu, \square_\g]K_{\mu\nu}&=&2{{\R_{\a}}^{\mu}}_{\nu\ep} \D_\mu K^{\a\ep } + \R^{\a\ep} \D_\a K_{\ep\nu} + \D^\a{ \R_{\a} }^{\ep} K_{\ep\nu}+\D^\a{ \R_{\mu\a\nu }}^\ep{K^\mu}_{\ep},
\eeaa
we have 
\begin{equation}
\label{eq-intermediate-2-proof}
\begin{split}
\D^\mu\D^\a \D_{\a} K_{\mu\nu}+2\D^\mu\D_\mu\D^\a K_{\a\nu}&=3\D^\a \D_{\a}  \D^\mu K_{\mu\nu}+2{{\R_{\a}}^{\mu}}_{\nu\ep} \D_\mu K^{\a\ep }\\
&+ \R^{\a\ep} \D_\a K_{\ep\nu} + \D^\a{ \R_{\a} }^{\ep} K_{\ep\nu}+\D^\a{ \R_{\mu\a\nu }}^\ep{K^\mu}_{\ep}.
\end{split}
\end{equation}
Plugging in the above as the left hand side in \eqref{eq-intermediate-proof}, we arrive to
\beaa
\D^\a \D_{\a}  \D^\mu K_{\mu\nu}&=& - \left(\frac 1 3 {{\R^{\a}}_{\nu\mu}}^\ep+ {{\R^{\a}}_{\mu \nu}}^\ep \right) \D^\mu K_{\a \ep}+\frac 1 3 \Big[- \left(\D^\mu{{\R^{\a}}_{\nu\mu}}^\ep+ \D^\mu{{\R^{\a}}_{\mu \nu}}^\ep \right)  K_{\a \ep}\\
&&- \D^\mu \left({{\R}_{\nu}}^\ep K_{\mu \ep}+{{\R}_{\mu}}^\ep K_{\ep \nu}\right)- \R^{\a\ep} \D_\a K_{\ep\nu} - \D^\a{ \R_{\a} }^{\ep} K_{\ep\nu}-\D^\a{ \R_{\mu\a\nu }}^\ep{K^\mu}_{\ep} \Big].
\eeaa
Writing again from \eqref{eq-intermediate-2-proof},
\beaa
\D^\mu\D^\a \D_{\a} K_{\mu\nu}&=&\D^\a \D_{\a}  \D^\mu K_{\mu\nu}+2{{\R_{\a}}^{\mu}}_{\nu\ep} \D_\mu K^{\a\ep }+ \R^{\a\ep} \D_\a K_{\ep\nu} + \D^\a{ \R_{\a} }^{\ep} K_{\ep\nu}+\D^\a{ \R_{\mu\a\nu }}^\ep{K^\mu}_{\ep},
\eeaa
we obtain
\beaa
\D^\mu \square_\g K_{\mu\nu}&=&- \left(\frac 1 3 {{\R^{\a}}_{\nu\mu}}^\ep- {{\R^{\a}}_{\mu \nu}}^\ep \right) \D^\mu K_{\a \ep}+\frac 1 3 \Big[- \left(\D^\mu{{\R^{\a}}_{\nu\mu}}^\ep+ \D^\mu{{\R^{\a}}_{\mu \nu}}^\ep \right)  K_{\a \ep}\\
&&- \D^\mu \left({{\R}_{\nu}}^\ep K_{\mu \ep}+{{\R}_{\mu}}^\ep K_{\ep \nu}\right)+2 \R^{\a\ep} \D_\a K_{\ep\nu} +2 \D^\a{ \R_{\a} }^{\ep} K_{\ep\nu}+2\D^\a{ \R_{\mu\a\nu }}^\ep{K^\mu}_{\ep} \Big].
\eeaa
Plugging the above in \eqref{eq-intermediate-3-proof}, we obtain
\begin{equation}\label{eq:intermediate-4-proof}
\begin{split}
[\KK, \square_\g]\phi &= \D^\mu K_{\a \ep}  \left(\frac 1 3 {{\R^{\a\de}}_{\mu}}^\ep- {{\R^{\a}}_{\mu }}^{\de\ep} \right) \D_\de \phi-\frac 2 3 \D^{\mu} {K^{\ep}}_{\a}{ {\R^\a}_{\mu \ep}}^\delta \D_\delta \phi \\
& - \frac 1 3 \Big[- \left(\D^\mu{{\R^{\a}}_{\nu\mu}}^\ep+ \D^\mu{{\R^{\a}}_{\mu \nu}}^\ep \right)  K_{\a \ep}- \D^\mu \left({{\R}_{\nu}}^\ep K_{\mu \ep}+{{\R}_{\mu}}^\ep K_{\ep \nu}\right)\\
&+2 \R^{\a\ep} \D_\a K_{\ep\nu} +2 \D^\a{ \R_{\a} }^{\ep} K_{\ep\nu}+2\D^\a{ \R_{\mu\a\nu }}^\ep{K^\mu}_{\ep} \Big] \D_\nu \phi\\
&-\D_{\mu}( K^{\mu\nu}  {\R_{\nu}}^\ep \D_\ep \phi)+{\R_{\mu}}^\ep \D^\mu (K_{\ep\nu} \D^\nu \phi) + \left(\D^\ep \R -\D^\mu {\R_\mu}^\ep \right) K_{\ep\nu} \D^\nu \phi.
\end{split}
\end{equation}
We now show that the first line of \eqref{eq:intermediate-4-proof} cancels out. 
Using that ${ {\R^\a}_{\mu \ep}}^\delta=-{{ {\R^\a}_{\mu}}^\delta}_{\ep}$ we have
\beaa
&& \D^\mu K_{\a \ep}  \left(\frac 1 3 {{\R^{\a\de}}_{\mu}}^\ep- {{\R^{\a}}_{\mu }}^{\de\ep} \right) \D_\de \phi-\frac 2 3 \D^{\mu} {K^{\ep}}_{\a}{ {\R^\a}_{\mu \ep}}^\delta \D_\delta \phi\\
&=& \frac 1 3 \D^\mu K_{\a \ep}  \left( {{\R^{\a\de}}_{\mu}}^\ep- {{\R^{\a}}_{\mu }}^{\de\ep} \right) \D_\de \phi=\frac 1 3 (-\D_{\ep} {K_{\a}}^{\mu}{{\R^{\a\de}}_{\mu}}^\ep+\D_{\a} {K^{\mu}}_{\ep}{{\R^{\a}}_{\mu }}^{\de\ep}) \D_\de \phi=0
\eeaa
where we wrote once again $\D_{\mu} K_{\ep\a}= -\D_{\ep} K_{\a\mu}-\D_{\a} K_{\mu\ep} $,
and observe that the first term is symmetric in $\a\mu$ while the second Riemann tensor term in antisymmetric in $\a\mu$, and the second term in symmetric in $\mu\ep$ while the first Riemann tensor is antisymmetric in $\mu\nu$. 
Writing $\D^{\nu} {K_{\a}}^{\mu} {{\R}^{\a \de}}_{\mu\nu}=\D^{\nu} {K_{\a}}^{\mu} {{\R}_{\mu\nu}}^{\a \de}=\D^{\a} {K_{\nu}}^{\mu} {{\R}_{\mu\a}}^{\nu \de}=-\D^{\a} {K_{\nu}}^{\mu} {{\R}_{\a\mu}}^{\nu \de}$ we obtain the cancellation of the first line.

We now simplify the remaining three lines on the RHS of \eqref{eq:intermediate-4-proof}. We have
\beaa
RHS &=& - \frac 1 3 \Big[- \left(\D^\mu{{\R^{\a}}_{\nu\mu}}^\ep+ \D^\mu{{\R^{\a}}_{\mu \nu}}^\ep \right)  K_{\a \ep}- \D^\mu \left({{\R}_{\nu}}^\ep K_{\mu \ep}+{{\R}_{\mu}}^\ep K_{\ep \nu}\right)+2 \R^{\a\ep} \D_\a K_{\ep\nu} +2 \D^\a{ \R_{\a} }^{\ep} K_{\ep\nu}\\
&&+2\D^\a{ \R_{\mu\a\nu }}^\ep{K^\mu}_{\ep} \Big] \D_\nu \phi-\D_{\mu} K^{\mu\nu}  {\R_{\nu}}^\ep \D_\ep \phi- K^{\mu\nu} \D_{\mu} {\R_{\nu}}^\ep \D_\ep \phi- K^{\mu\nu}  {\R_{\nu}}^\ep \D_{\mu}\D_\ep \phi\\
&&+{\R_{\mu}}^\ep \D^\mu K_{\ep\nu} \D^\nu \phi +{\R_{\mu}}^\ep K_{\ep\nu}\D^\mu  \D^\nu \phi + \left(\D^\ep \R -\D^\mu {\R_\mu}^\ep \right) K_{\ep\nu} \D^\nu \phi\\
&=& \Big[ \frac 1 3 \left(\D^\mu{{\R^{\a}}_{\nu\mu}}^\ep+ \D^\mu{{\R^{\a}}_{\mu \nu}}^\ep \right)  K_{\a \ep}+  \frac 1 3 \D^\mu \left({{\R}_{\nu}}^\ep K_{\mu \ep}+{{\R}_{\mu}}^\ep K_{\ep \nu}\right)+{\R_{\mu}}^\ep \D^\mu {K_{\ep}}^{\nu}- \frac 2 3  \R^{\a\ep} \D_\a K_{\ep\nu} \\
&&- \frac 2 3  \D^\a{ \R_{\a} }^{\ep} K_{\ep\nu}- \frac 2 3 \D^\a{ \R_{\mu\a\nu }}^\ep{K^\mu}_{\ep} -\D_{\mu} K^{\mu\a}  {\R_{\a}}^\nu- K^{\mu\a} \D_{\mu} {\R_{\a}}^\nu+ \left(\D^\ep \R -\D^\mu {\R_\mu}^\ep \right) {K_{\ep}}^{\nu} \Big] \D_\nu \phi.
\eeaa
This gives
\beaa
RHS &=& \Big[ \frac 1 3 \left(\D^\mu{{\R^{\a}}_{\nu\mu}}^\ep- \D^\mu{{\R^{\a}}_{\mu \nu}}^\ep -2 \D^\a {{\R}_{\nu}}^\ep\right)  K_{\a \ep}+\big( \D^\a \R-\frac 4 3  \D^\mu {{\R}_{\mu}}^\a \big) K_{\a\nu}\\
&&-  \frac 2 3 {{\R}_{\nu}}^\ep  \D^\mu K_{\mu \ep}+ \frac 2 3{{\R}_{\mu}}^\ep  \D^\mu  K_{\ep \nu} \Big] \D^\nu \phi.
\eeaa
By writing  $\D^\mu{{\R^{\a}}_{\mu \nu}}^\ep=-\D_\nu \R^{\a\ep}+\D^\ep {\R^\a}_{\nu}$ and
$\D^\mu{{\R^\a}_{\nu\mu}}^{\ep}=\D^\a {\R^\ep}_{\nu}-\D_\nu \R^{\ep\a}$ we have
\beaa
RHS &=& \Big[- \frac 2 3 \D^\a {\R^\ep}_{\nu}   K_{\a \ep}+\left( \D^\a \R-\frac 4 3  \D^\mu {{\R}_{\mu}}^\a \right) K_{\a\nu}-  \frac 2 3 {{\R}_{\nu}}^\ep  \D^\mu K_{\mu \ep}+ \frac 2 3{{\R}_{\mu}}^\ep  \D^\mu  K_{\ep \nu} \Big] \D^\nu \phi,
\eeaa
as stated.

\end{document}